\documentclass{article}

\usepackage[preprint]{neurips_2026} 

\usepackage[utf8]{inputenc} 
\usepackage[T1]{fontenc}    
\usepackage{hyperref}       
\usepackage{footnotehyper}
\usepackage{url}            
\usepackage{booktabs}       
\usepackage{amsfonts}       
\usepackage{nicefrac}       
\usepackage{xcolor}         

\usepackage{dsfont} 
\usepackage{multicol}
\usepackage{multirow}   
\usepackage{makecell}   
\usepackage{array}      
\usepackage{lscape} 
\usepackage{enumitem} 
\usepackage{placeins} 
\usepackage{array} 
\usepackage{setspace} 
\usepackage{rotating}

\usepackage{tikz}
\usetikzlibrary{positioning, trees}
\newdimen\nodeDist
\usetikzlibrary{shapes.geometric, arrows}
\definecolor{gold}{rgb}{0.85,0.65,0.13}

\usepackage{algorithm, algpseudocode}

\usepackage{amsthm}
\usepackage{amsmath}
\usepackage{amssymb}
\usepackage{relsize} 
\usepackage{bm}
\usepackage{mathtools} 

\usepackage{ee}

\usepackage[acronym]{glossaries} 
\glsdisablehyper 

\newacronym{2sls}{2SLS}{Two-Stage Least Squares}
\newacronym{bart}{BART}{Bayesian Additive Regression Trees}
\newacronym{bcf}{BCF}{Bayesian Causal Forest}
\newacronym{bcf-iv}{BCF-IV}{Bayesian Instrumental Variable Causal Forest}
\newacronym{cace}{CACE}{Complier Average Causal Effect}
\newacronym{cart}{CART}{Classification and Regression Trees}
\newacronym{ccace}{cCACE}{conditional Complier Average Causal Effect}
\newacronym{citt}{cITT}{conditional Intention-To-Treat}
\newacronym{DR}{DR}{Detection Rate}
\newacronym{FDR}{FDR}{False Detection Rate}
\newacronym{itt}{ITT}{Intention-to-Treat}
\newacronym{iv}{IV}{Instrumental Variables}
\newacronym{sbcf}{SBCF}{Shrinkage Bayesian Causal Forest}
\newacronym{sbcf-iv}{SBCF-IV}{Shrinkage Bayesian Causal Forest with Instrumental Variable}
\newacronym{grf}{GRF}{generalized random forest}
\newacronym{grf-iv}{GRF-IV}{GRF-based Instrumental Variable Causal Forest}
\newacronym{SUTVA}{SUTVA}{Stable Unit Treatment Value Assumption}
\newacronym{DGP}{DGP}{data-generating process}
\newacronym{IRA}{IRA}{Individual Retirement Account}

\theoremstyle{plain}
\newtheorem{cor}{Corollary}[section]
\newtheorem{prop}{Proposition}[section]

\newtheorem{thm}{Theorem}[section]

\theoremstyle{definition}

\newtheorem{defn}{Definition}[section]
\newtheorem{assump}{Assumption}[section]

\hypersetup{colorlinks = true,
            citecolor  = blue,
            linkcolor  = blue,
            urlcolor   = blue
            }

\newenvironment{notes}
  {\begin{minipage}{\linewidth}\smallskip\footnotesize\emph{Notes:}}
  {\end{minipage}}

\newenvironment{keywords}
{\bgroup\leftskip 20pt\rightskip 20pt \small\noindent{\bfseries
Keywords:} \ignorespaces}%
{\par\egroup\vskip 0.25ex}
\newenvironment{jelclass}
{\bgroup\leftskip 20pt\rightskip 20pt \small\noindent{\bfseries
JEL classification:} \ignorespaces}%
{\par\egroup\vskip 0.25ex}

\title{Shrinkage Bayesian Causal Forest with Instrumental Variable}

\author{%
  Lennard Ma{\ss}mann\thanks{Ruhr Graduate School in Economics (RGS Econ), Research Academy Ruhr, Universitätsstr. 150, 44801 Bochum, Germany} \\
  University of Duisburg-Essen \\
  Faculty of Business Administration and Economics\\
  Universitätsstraße 12, 45117 Essen, Germany \\
  \texttt{lennard.massmann@uni-due.de} \\
  \And
  Jens Klenke \\
  University of Duisburg-Essen \\
  Faculty of Business Administration and Economics\\
  Universitätsstraße 12, 45117 Essen, Germany \\
  \texttt{jens.klenke@vwl.uni-due.de} \\
}

\begin{document}

\maketitle

\begin{abstract}
Discovering interpretable subgroups whose complier effects deviate from the average is a central goal of instrumental variable analysis under imperfect compliance, yet existing tree-based methods degrade when most covariates are irrelevant to the effect.
We propose Shrinkage Bayesian Causal Forest with Instrumental Variable (SBCF-IV) for discovering and estimating subgroups with heterogeneous Complier Average Causal Effects (CACE) in sparse high-dimensional settings. 
SBCF-IV places a sparsity-inducing Dirichlet prior on the splitting probabilities of the Bayesian Additive Regression Trees that estimate the conditional intention-to-treat and the complier share, concentrating posterior mass on the few covariates that moderate the complier effect and thereby regularizing effect estimation.
The posterior split frequencies additionally enter a downstream CART as variable-level costs that steer the partition toward relevant moderators, providing an interpretable division of the covariate space.
Monte Carlo experiments show that, as the share of irrelevant covariates grows, SBCF-IV recovers the true partition more reliably than its non-sparse predecessor BCF-IV at the tree and unit level, and retains nominal coverage where BCF-IV's intervals deteriorate. We apply the method to the Oregon Health Insurance Experiment and the 401(k) eligibility data.
\end{abstract}


\begin{jelclass}
    C11, C14, C21, C26
\end{jelclass}


\section{Introduction}\label{ch:intro}

The heterogeneous treatment effect (HTE) literature has expanded along two complementary dimensions: estimating the conditional average treatment effect (CATE) and discovering interpretable subgroups whose effects deviate from the population average \citep{lipkovich_tutorial_2017, kunzel_metalearners_2019, dwivedi_stable_2020}. For CATE estimation, nonparametric machine learning methods have become standard, including the causal forest \citep{athey_generalized_2019}, Bayesian Additive Regression Trees (BART) \citep{hill_bayesian_2011}, \gls{bcf}  \citep{hahn_bayesian_2020, caron_shrinkage_2022}, and doubly robust meta-learners \citep{kennedy_towards_2023, semenova_debiased_2021}. In parallel, a growing strand of work has focused on the data-driven discovery of interpretable subgroups, with decision-tree-based methods \citep{athey_recursive_2016, bargagli_stoffi_causal_2020, bargagli-stoffi_heterogeneous_2022, lee_discovering_2021} being particularly prominent due to their interpretability. Ensemble-based extensions such as the causal rule ensemble \citep{bargagli-stoffi_causal_2024} and causal distillation trees \citep{huang_distilling_2025} address the shortcomings of unstable single-tree methods by aggregating decision rules across many trees, yielding more stable and expressive subgroup representations when estimating CATE. 

A limitation of recent work on HTE is that the identification for CATE rests on the assumption of regular assignment mechanisms, which is rarely defensible in observational studies with unobserved heterogeneity in treatment uptake. When a valid instrument is available, a relevant target estimand becomes the \gls{cace}, or Local Average Treatment Effect (LATE), nonparametrically identified for the subpopulation of compliers under the standard \gls{iv} assumptions \citep{imbens_identification_1994, angrist_identification_1996}. The interpretation of complier effects has been the subject of ongoing debate: because compliance status is a counterfactual quantity that is never directly observed, critics have argued that the CACE pertains to an unidentified subgroup and is therefore of limited policy relevance \citep{deaton_instruments_2010, swanson_think_2014}. In contrast, the LATE is often the most one can learn nonparametrically in the presence of unmeasured confounding without imposing restrictive effect homogeneity assumptions, and it remains informative about the underlying causal structure. Crucially, when covariates explain a substantial share of the variation in compliance, complier effects effectively coincide with conditional effects in identifiable subgroups, and the concerns about an unknown target population largely dissolve \citep{kennedy_sharp_2020}. This observation provides direct motivation for studying the \gls{ccace}. The \gls{ccace} is the \gls{cace} as a function of observed characteristics. Characterizing how complier effects vary along interpretable covariate profiles transforms the CACE from a property of an unobserved subgroup into a set of policy-relevant statements about identifiable populations.
Several methods have been developed to recover heterogeneous \gls{iv} effects, each targeting a distinct inferential object. Forest-based estimators such as the instrumental variable forest of \citet{wang_instrumental_2022} and the instrumental forest within the \gls{grf} framework \citep{athey_generalized_2019} target a unit-level conditional \gls{iv} function, delivering pointwise-consistent estimates with valid asymptotic inference but no explicit partition of the covariate space. The drivers of heterogeneity are recovered post hoc through variable importance scores or best linear projections. Partition-based procedures like the \gls{iv} tree of \citet{wang_instrumental_2022} and the matching procedure of \citet{johnson_detecting_2022} return an interpretable partition over which subgroup-level effects can be read off directly, paired with closed-testing inference.

We propose \gls{sbcf-iv}, a generalization of \gls{bcf-iv} tailored to settings in which the share of covariates that drive effect heterogeneity is small relative to $P$. The \gls{bcf-iv} algorithm of \citet{bargagli-stoffi_heterogeneous_2022} combines a \gls{bart}-based \citep{chipman_bart_2010} sum-of-trees estimator for the \gls{citt} with a shallow \gls{cart} post-processing step \citep{breiman_classification_1984}, yielding interpretable subgroup-level estimates of \gls{ccace} within a two-step procedure based on stratification and \gls{iv} estimation. By doing so, it reconciles the predictive accuracy of ensemble methods \citep{athey_generalized_2019, hartford_deep_2017} with the interpretability of single-tree approaches \citep{athey_recursive_2016, bargagli_stoffi_causal_2020, johnson_detecting_2022}. However, \gls{bcf-iv} inherits \gls{bart}'s uniform split-variable prior and therefore performs no targeted feature selection: when the covariate vector contains many irrelevant variables, as is typical in modern administrative or biomedical data, both the ensemble and the downstream \gls{cart} can spread splits across spurious moderators, weakening the discovered subgroups.
Our contribution is twofold. First, we diagnose the failure mode of \gls{bcf-iv} in high dimensions. \gls{bcf-iv} increasingly splits on spurious moderator variables as the share of irrelevant covariates grows, which degrades subgroup recovery and effect estimation.
We address this by estimating the \gls{citt} with the \gls{sbcf} of \citet{caron_shrinkage_2022}, whose sparsity-inducing Dirichlet prior, in the spirit of SoftBART \citep{linero_bayesian_2018, linero_bayesian_2018-1}, concentrates posterior mass on the few covariates that moderate the effect. The resulting estimation gains in our setting are attributable to this prior.
Second, we feed the posterior split frequencies into the subgroup-discovery \gls{cart} as variable-level costs, a component that steers the partition toward ensemble-relevant covariates and is most beneficial when the ensemble is not already sparse. Throughout, our scope follows that of \citet{bargagli-stoffi_heterogeneous_2022}: a binary randomized instrument, a binary treatment, and standard \gls{iv} identification assumptions.
Section~\ref{ch:sim_study} accordingly benchmarks \gls{sbcf-iv} against \gls{bcf-iv} as the direct ancestor and natural comparator under this setting, with a supplementary comparison to the
instrumental forest of \citet{athey_generalized_2019} and a cost-weighting ablation in Appendix~\ref{append:further_precision_results}.

The paper proceeds as follows. Section~\ref{ch:PO_irreg} sets up the potential outcomes framework and identifying assumptions for the \gls{ccace} under an irregular assignment mechanism. Section~\ref{ch:BCF_IV} introduces \gls{sbcf-iv}. Section~\ref{ch:sim_study} reports Monte Carlo evidence on its performance relative to \gls{bcf-iv} in high-dimensional settings based on tree-level and unit-level performance criteria. Section~\ref{ch:emp_appl} applies \gls{sbcf-iv} to two empirical studies: the Oregon Health Insurance Experiment (OHIE) \citep{finkelstein_oregon_2012, johnson_detecting_2022} and the 401(k) retirement plans dataset \citep{poterba_401k_1992, poterba_401k_1995, chernozhukov_doubledebiased_2018}. Section~\ref{ch:conclusion} concludes.

\section{Potential outcomes and irregular assignment}
\label{ch:PO_irreg}

We follow \citet{bargagli_stoffi_causal_2020} and \citet{bargagli-stoffi_heterogeneous_2022} and adopt Rubin's causal model, working within the irregular assignment framework of \citet{imbens_causal_2015}. For $N$ units indexed $i = 1, \dots, N$, let $Y_i \in \mathbb{R}$ denote the observed outcome, $Z_i \in \{0,1\}$ a binary instrument (assignment), $W_i \in \{0,1\}$ the actual treatment received, and $X_i \in \mathbb{R}^P$ the $i$-th row of an $N \times P$ matrix $X$ of pre-treatment covariates. Each unit is endowed with potential outcomes $Y_i(Z_i=z, W_i=w)$ and potential treatments $W_i(z)$ for $z, w \in \{0,1\}$, related to observed quantities by the consistency relations $Y_i = Y_i(Z_i, W_i)$ and $W_i = W_i(Z_i)$. The instrument $Z_i$ is unconfounded but the receipt $W_i$ may be confounded. This is the canonical \gls{iv} setting, and our inferential target is the \gls{ccace} in Definition \ref{defn:cCACE}, based on the latent subpopulation of compliers.

\begin{defn}[Compliance subgroups]
\label{defn:compliance}
Each unit belongs to one of four latent compliance subgroups, defined by the joint values of its potential treatments:
\[
   G_i =
   \begin{cases}
      C,  & W_i(0) = 0,\ W_i(1) = 1 \quad \text{(compliers)} \\
      D,  & W_i(0) = 1,\ W_i(1) = 0 \quad \text{(defiers)} \\
      AT, & W_i(0) = 1,\ W_i(1) = 1 \quad \text{(always-takers)} \\
      NT, & W_i(0) = 0,\ W_i(1) = 0 \quad \text{(never-takers)},
   \end{cases}
\]
with conditional subgroup proportions $\pi_G(x) = \Pr(G_i = G \mid X_i = x)$ for $G \in \{C, D, AT, NT\}$.
\end{defn}

\begin{defn}[Conditional CACE]
\label{defn:cCACE}
The conditional Complier Average Causal Effect is the estimand
\[
   \tau^{\text{CACE}}(x) \;\coloneqq\; \mathbb{E}\bigl[Y_i(1, W_i(1)) - Y_i(0, W_i(0)) \,\bigm|\, G_i = C,\, X_i = x\bigr],
\]
which, under the exclusion restriction in Assumption \ref{assump:identification_irregular}(d), reduces to $\mathbb{E}[Y_i(1) - Y_i(0) \mid G_i = C, X_i = x]$.
\end{defn}

Definition~\ref{defn:cCACE} fixes the target as a property of the latent complier subpopulation and it is not, by itself, an object computable from the observed distribution of $(Y_i, W_i, Z_i, X_i)$. Identification proceeds through the conditional \gls{itt} effect and the conditional complier share, both of which admit clean expressions in terms of observed conditional means.

\begin{defn}[Conditional ITT and complier share]
\label{defn:cITT}
The conditional \gls{itt} effect and the conditional complier share are
\begin{align*}
   \text{ITT}_Y(x) &\;\coloneqq\; \mathbb{E}[Y_i \mid Z_i = 1, X_i = x] - \mathbb{E}[Y_i \mid Z_i = 0, X_i = x], \\
   \pi_C(x) &\;\coloneqq\; \Pr(G_i = C \mid X_i = x).
\end{align*}
By a mixture argument over the compliance subgroups of Definition~\ref{defn:compliance},
\begin{align*}
   \text{ITT}_Y(x) \;=\; \pi_C(x)\,\text{ITT}_{Y,C}(x) + \pi_D(x)\,\text{ITT}_{Y,D}(x) + \pi_{AT}(x)\,\text{ITT}_{Y,AT}(x) + \pi_{NT}(x)\,\text{ITT}_{Y,NT}(x),
\end{align*}
where $\text{ITT}_{Y,G}(x)$ denotes the conditional ITT among units of compliance types $G$ in Definition \ref{defn:compliance}.
\end{defn}

Definition~\ref{defn:cITT} makes explicit that $\text{ITT}_Y(x)$ is a covariate-weighted mixture and isolating $\tau^{\text{CACE}}(x) = \text{ITT}_{Y,C}(x)$ requires assumptions that neutralize the contributions of defiers, always-takers, and never-takers. The classical \citet{angrist_identification_1996} conditions for \gls{iv} estimation deliver point identification.

\begin{assump}[IV identification under irregular assignment] 
\label{assump:identification_irregular} \leavevmode
\begin{enumerate}[label=(\alph*)]
   \item \textit{\gls{SUTVA} / consistency:} $Y_i = Y_i(Z_i, W_i)$ and $W_i = W_i(Z_i)$, with no interference between units and no hidden treatment variants.
   \item \textit{Relevance:} $\pi_C(x) > 0$ almost surely.
   \item \textit{Unconfounded instrument:} $Z_i \perp\!\!\!\perp \bigl(\{Y_i(z, w)\}_{z,w \in \{0,1\}^2},\, W_i(0),\, W_i(1)\bigr) \,\bigm|\, X_i$.
   \item \textit{Exclusion restriction:} $Y_i(z, w) = Y_i(w)$ for all $z, w \in \{0,1\}$.
   \item \textit{Monotonicity:} $W_i(1) \geq W_i(0)$.
\end{enumerate}
\end{assump}

Assumption~\ref{assump:identification_irregular}(a) ensures consistency and rules out interference and hidden treatment variants. Relevance (b), together with monotonicity (e), guarantees $\pi_C(x) > 0$ almost surely so that the identification ratio below is well-defined. Monotonicity also rules out defiers ($\pi_D(x) = 0$), which must be defended on substantive grounds or enforced through a one-sided non-compliance design. The exclusion restriction (d) confines the effect of $Z_i$ on $Y_i$ to the channel through $W_i$. Under Assumption~\ref{assump:identification_irregular}, the complier estimand of Definition~\ref{defn:cCACE} is identified from the observed distribution.

\begin{prop}[Identification of $\tau^{\text{CACE}}(x)$]
\label{prop:identification}
Under Assumption~\ref{assump:identification_irregular},
\[
   \tau^{\text{CACE}}(x) \;=\; \frac{\text{ITT}_Y(x)}{\pi_C(x)} \;=\; \frac{\mathbb{E}[Y_i \mid Z_i = 1, X_i = x] - \mathbb{E}[Y_i \mid Z_i = 0, X_i = x]}{\mathbb{E}[W_i \mid Z_i = 1, X_i = x] - \mathbb{E}[W_i \mid Z_i = 0, X_i = x]}.
\]
The proof, adapting \citet{angrist_identification_1996} to the conditional target as in \citet{bargagli-stoffi_heterogeneous_2022}, is given in Appendix~\ref{append:proof_cCACE}.
\end{prop}

Given a partition $\{\mathbb{X}_j\}_j$ of the covariate space, Proposition~\ref{prop:identification} motivates a sample-moment estimator targeting the subgroup-averaged complier effect
\begin{align}\label{eq:subgroup-cace}
   \tau^{\text{CACE}}_{\mathbb{X}_j} \;\coloneqq\; \mathbb{E}\bigl[Y_i(W=1) - Y_i(W=0) \,\bigm|\, G_i = C,\, X_i \in \mathbb{X}_j\bigr],
\end{align}
obtained by replacing population conditional means in Proposition~\ref{prop:identification} with subgroup sample analogues.

\begin{defn}[Subgroup-wise 2SLS estimator]
\label{defn:cCACE_estimator}
For $X_i =x \in \mathbb{X}_j$, with $N_{z,j} = \sum_{l:\, X_l \in \mathbb{X}_j} \mathbf{1}\{Z_l = z\}$ the count of $\mathbb{X}_j$-units assigned to $Z_l = z \in \{0,1\}$,
\begin{align*}
   \widehat\tau^{\text{CACE}}(x) \coloneqq \widehat\tau^{\,\text{2SLS}}_{\mathbb{X}_j} \;=\; \frac{\dfrac{1}{N_{1,j}} \sum_{l:\, X_l \in \mathbb{X}_j} Y_l Z_l \;-\; \dfrac{1}{N_{0,j}} \sum_{l:\, X_l \in \mathbb{X}_j} Y_l (1 - Z_l)}{\dfrac{1}{N_{1,j}} \sum_{l:\, X_l \in \mathbb{X}_j} W_l Z_l \;-\; \dfrac{1}{N_{0,j}} \sum_{l:\, X_l \in \mathbb{X}_j} W_l (1 - Z_l)}
\end{align*}
targets $\tau^{\text{CACE}}_{\mathbb{X}_j}$.
\end{defn}

Equivalently, $\widehat\tau^{\,\text{2SLS}}_{\mathbb{X}_j}$ is the Two-Stage Least Squares estimator on the subgroup-restricted simultaneous system
\begin{equation}\label{eq:simultaneous_equations}
   Y_{i, \mathbb{X}_j} \;=\; \kappa_{\mathbb{X}_j} + \tau^{\text{CACE}}_{\mathbb{X}_j}\,W_{i, \mathbb{X}_j} + \varepsilon_{i, \mathbb{X}_j}, \qquad W_{i, \mathbb{X}_j} \;=\; \pi_{0, \mathbb{X}_j} + \pi_{C, \mathbb{X}_j}\,Z_{i, \mathbb{X}_j} + \nu_{i, \mathbb{X}_j},
\end{equation}
with intercepts $\kappa_{\mathbb{X}_j}$ and $\pi_{0, \mathbb{X}_j}$, error terms $\varepsilon_{i, \mathbb{X}_j}$ and $\nu_{i, \mathbb{X}_j}$, while $\mathbb{E}[\varepsilon_{i, \mathbb{X}_j}] = \mathbb{E}[\nu_{i, \mathbb{X}_j}] = 0$ and $\mathbb{E}[Z_{i, \mathbb{X}_j}\,\nu_{i, \mathbb{X}_j}] = 0$. Under Assumption~\ref{assump:identification_irregular} and a sufficient number of i.i.d.\ observations within each $\mathbb{X}_j$, $\widehat\tau^{\,\text{2SLS}}_{\mathbb{X}_j}$ is consistent and asymptotically normal for $\tau^{\text{CACE}}_{\mathbb{X}_j}$, with the reduced form and formal asymptotic results collected in Appendix~\ref{append:theorem_2SLS}.

Definition~\ref{defn:cCACE_estimator} presumes that the partition $\{\mathbb{X}_j\}_j$ is known. The contribution of \gls{bcf-iv} \citep{bargagli-stoffi_heterogeneous_2022} is to discover this partition from the data through (i) an honest split of the sample into disjoint discovery and inference subsets $\mathcal{I}_{\text{disc}}$ and $\mathcal{I}_{\text{inf}}$; (ii) interpretable discovery of heterogeneity on $\mathcal{I}_{\text{disc}}$; and (iii) inference for $\tau^{\text{CACE}}_{\mathbb{X}_j}$ on $\mathcal{I}_{\text{inf}}$. Our \gls{sbcf-iv} algorithm in Algorithm~\ref{alg:sbcf-iv} inherits this three-step structure of \citet{bargagli-stoffi_heterogeneous_2022} with adaptations for a high-dimensional covariate setup described in the next section.

\section{Shrinkage Bayesian Causal Forest with Instrumental Variable}
\label{ch:BCF_IV}

We propose \gls{sbcf-iv}, an extension of \gls{bcf-iv} \citep{bargagli-stoffi_heterogeneous_2022} to settings with many irrelevant covariates. The overall structure of honest sample splitting, data-driven discovery of heterogeneous subgroups on $\mathcal{I}_{\text{disc}}$, and \gls{2sls} inference on $\mathcal{I}_{\text{inf}}$ is preserved. Our adaptations concern the discovery step where we replace BCF with the sparsity-inducing Shrinkage BCF (SBCF) of \citet{caron_shrinkage_2022} and feed its posterior variable-selection frequencies into the subgroup-finding tree as variable-level costs. Algorithm~\ref{alg:sbcf-iv} summarizes the full procedure.

Working on $\mathcal{I}_{\text{disc}}$, we separately model the numerator and denominator of the identification ratio in Proposition~\ref{prop:identification}. Following \citet{hahn_bayesian_2020}, we adopt the semi-parametric specification%
\footnote{Background on CART, BART, BCF, and SBCF is collected in Appendix~\ref{append:bart_bcf_cart}.}
\begin{equation}\label{eq:cond_exp_Y}
   \mathbb{E}[Y_i \mid Z_i = z, X_i = x] \;=\; \mu\bigl(e(x), x\bigr) \;+\; \text{ITT}_Y(x)\,z,
\end{equation}
where $e(x) = \Pr(Z_i = 1 \mid X_i = x)$ is the instrument's propensity score, included as a covariate in the control function $\mu(e(x), x)$ to mitigate regularization-induced confounding and targeted selection. Independent BART priors are placed on $\mu(e(x), x)$ and $\text{ITT}_Y(x)$, with depth-penalty parameters $(\eta, \beta) = (0.25, 3)$ on $\text{ITT}_Y(x)$ favoring shallow trees and hence simpler heterogeneity patterns. The compliance component is modeled analogously via
\begin{equation}\label{eq:cond_exp_W}
   \mathbb{E}[W_i \mid Z_i = z, X_i = x] \;=\; \delta(z, x),
\end{equation}
with a BART probit prior on $\delta(z, x)$ \citep{hill_bayesian_2011}. Combining estimates yields a pointwise complier-share estimate $\widehat\pi_C^{\,\text{SBCF}}(x) = \widehat\delta(1, x) - \widehat\delta(0, x)$ and a pointwise complier-effect estimate $\widehat\tau^{\,\text{SBCF}}(x) = \widehat{\text{ITT}}_Y^{\,\text{SBCF}}(x) / \widehat\pi_C^{\,\text{SBCF}}(x)$. We emphasize that $\widehat\tau^{\,\text{SBCF}}(x)$ is a posterior estimate on $\mathcal{I}_{\text{disc}}$ used solely as a heterogeneity signal for the tree-fitting step, while we conduct final inference for the subgroup target $\tau^{\text{CACE}}_{\mathbb{X}_j}$ with $\widehat\tau^{\,\text{2SLS}}_{\mathbb{X}_j}$ from Definition \ref{defn:cCACE_estimator} on $\mathcal{I}_{\text{inf}}$.

\gls{sbcf-iv} departs from \gls{bcf-iv} in the prior on the split-variable selection probabilities $s = (s_1, \ldots, s_P)$. BCF uses a uniform $s_j = 1/P$. Instead, SBCF imposes a sparsity-inducing Dirichlet prior,
\begin{equation}\label{eq:sbcf_prior}
   s \sim \text{Dirichlet}\!\left(\tfrac{\alpha}{P}, \ldots, \tfrac{\alpha}{P}\right), \qquad \frac{\alpha}{\alpha + \rho} \sim \text{Beta}(a, b),
\end{equation}
with defaults $(a, b, \rho) = (0.5, 1, P)$ \citep{caron_shrinkage_2022}. Small $\alpha$ concentrates mass on few covariates, and the hyperprior on $\alpha/(\alpha + \rho)$ lets the data determine the degree of sparsity, with the preference for sparsity strengthening as $P$ grows. We place separate Dirichlet priors on the split probabilities $s_\mu$ and $s_{\text{ITT}_Y}$ of the two components of \eqref{eq:cond_exp_Y}, using $\rho_\mu = P + 1$ (to accommodate the propensity score as an extra covariate in $\mu$) and $\rho_{\text{ITT}_Y} = P/2$ (to further concentrate mass on few active splits in the treatment effect component). The same sparsity-inducing prior in \eqref{eq:sbcf_prior} is applied in the SoftBART probit model used for $\delta(z, x)$ in \eqref{eq:cond_exp_W}, with $\rho_\delta = P+1$. Appendix~\ref{append:sbcf_priors} provides more information on the full sparsity-inducing prior specifications. We fit a shallow CART \citep{breiman_classification_1984} to the pointwise estimates $\widehat\tau^{\,\text{SBCF}}(X_i)$ to recover an interpretable partition $\{\mathbb{X}_j\}_j$. The posterior split frequencies $\widehat{s}_{\text{ITT}_Y}$ indicate which covariates drive heterogeneity and we pass these split frequencies to \texttt{rpart}  \citep{rpart} through the \texttt{cost} argument, setting variable-level costs to
\begin{equation}\label{eq:cost}
   c_{\text{psp}} \;=\; \frac{\max\{\widehat{s}_{\text{ITT}_Y}\}}{\widehat{s}_{\text{ITT}_Y}}.
\end{equation}
Because \texttt{rpart} divides split-improvement by the candidate variable's cost, \eqref{eq:cost} upweights covariates with a higher posterior inclusion probability at each split, in a manner analogous to variable-importance weighting in random forests \citep{breiman_random_2001}.
The cost vector $c_{\text{psp}}$ is defined via $\widehat s_{\text{ITT}_Y}$, the posterior
split frequencies of the cITT component, mirroring the ITT-anchored discovery step of \gls{bcf-iv} \citep{bargagli-stoffi_heterogeneous_2022}. Posterior split frequencies can be exactly zero for covariates the ensemble never selects given finitely many draws. Therefore, we floor $\widehat s_{\text{ITT}_Y}$ at a small $\varepsilon_c > 0$ before inversion in \eqref{eq:cost}, capping the cost of a never-selected covariate at $\max\{\widehat s_{\text{ITT}_Y}\}/\varepsilon_c$. Equation \eqref{eq:cost} is thus always well-defined and assigns such covariates a large but finite penalty.
The choice of $\widehat s_{\text{ITT}_Y}$ for the cost vector is justified by a coherence property of the two-step procedure: within a leaf, the discovery average and the \gls{2sls} estimand differ only by the within-leaf covariance between the complier share and the complier effect, a leaf-level instance of the compliance-weighting of instrumental variable estimands \citep{angrist_twostage_1995, abadie_semiparametric_2003, frolich_nonparametric_2007}. 
When within-leaf variation of $\pi_C(x)$ is small, heterogeneity in $\tau^{\text{CACE}}(x)$ is inherited almost entirely from $\text{ITT}_Y(x)$ and $\widehat s_{\text{ITT}_Y}$ is the natural cost weight (see Corollary~\ref{cor:coherence} in Appendix~\ref{append:proof_cCACE}). Aggregating $\widehat s_{\text{ITT}_Y}$ with the compliance component's split frequencies, for regimes with substantial denominator-driven heterogeneity, is an extension left to future work.

Inference on the discovered partition follows \gls{bcf-iv} without modification. For each node $\mathbb{X}_j$ of the tree learned on $\mathcal{I}_{\text{disc}}$, we compute $\widehat\tau^{\,\text{2SLS}}_{\mathbb{X}_j}$ on $\mathcal{I}_{\text{inf}}$ via the simultaneous system in \eqref{eq:simultaneous_equations}. Consistency and asymptotic normality under subgroup-level moment conditions are collected in Appendix~\ref{append:theorem_2SLS}. To guard against spurious heterogeneity, nodes flagged by a first-stage $F$-test for weak instruments are discarded, and $p$-values across leaves are adjusted for the familywise error rate using Holm-corrections for adjusted $p$-values \citep{holm_simple_1979, bargagli-stoffi_heterogeneous_2022}.

\begin{algorithm}[t]
\caption{Shrinkage Bayesian Causal Forest with Instrumental Variable (SBCF-IV)}\label{alg:sbcf-iv}
\begin{algorithmic}[1]
    \Require $N$ units $\{(X_i, Z_i, W_i, Y_i)\}_{i=1}^N$.
    \Ensure Tree-structured partition of the covariate space with node-level CACE estimates.
    \Statex
    \Statex \textbf{Step 1: Honest splitting.}
    \State Randomly partition the sample into $\mathcal{I}_{\text{disc}}$ and $\mathcal{I}_{\text{inf}}$ (defaults to half-size splits).
    \Statex
    \Statex \textbf{Step 2: Discovery} (on $\mathcal{I}_{\text{disc}}$).
    \State Estimate $\widehat{\text{ITT}}_Y^{\,\text{SBCF}}(x)$ via SBCF under the sparsity prior \eqref{eq:sbcf_prior}; save posterior split frequencies $\widehat{s}_{\text{ITT}_Y}$.
    \State Estimate $\widehat\pi_C^{\,\text{SBCF}}(x)$ via a SoftBART probit with similar sparsity prior.
    \State Form the pointwise heterogeneity signal $\widehat\tau^{\,\text{SBCF}}(x) = \widehat{\text{ITT}}_Y^{\,\text{SBCF}}(x) / \widehat\pi_C^{\,\text{SBCF}}(x)$.
    \State Fit a shallow CART to $\{(\widehat\tau^{\,\text{SBCF}}(x), X_i)\}$ with variable-level costs $c_{\text{psp}}$ as in \eqref{eq:cost}; the resulting partition is $\{\mathbb{X}_j\}_j$.
    \Statex
    \Statex \textbf{Step 3: Inference} (on $\mathcal{I}_{\text{inf}}$).
    \State For every node $\mathbb{X}_j$, compute $\widehat\tau^{\,\text{2SLS}}_{\mathbb{X}_j}$ via \eqref{eq:simultaneous_equations}, targeting $\tau^{\text{CACE}}_{\mathbb{X}_j}$.
    \State Run first-stage weak-instrument tests and adjust leaf-level $p$-values for the familywise error rate.
    \Statex
    \State \Return The pruned tree with node-level CACE estimates and adjusted inference.
\end{algorithmic}
\end{algorithm}

\section{Simulation study}
\label{ch:sim_study}

We combine the design of \citet{bargagli-stoffi_heterogeneous_2022} with the high-dimensional setup of \citet{caron_shrinkage_2022} to assess \gls{sbcf-iv} and \gls{bcf-iv} in settings with many irrelevant covariates. Appendix~\ref{append:sim_design} presents a detailed description of the complete simulation design. For each of $N = 1{,}000$ units, we generate a binary instrument $Z_i \sim \text{Bernoulli}(0.5)$, a covariate vector $X_i \in \mathbb{R}^P$ with $P \in \{10, 50, 100\}$ (half binary, half continuous), and potential outcomes and treatments according to
\begin{align}\label{eq:dgp}
\begin{split}
   W_i(0) = 0, \quad W_i(1) &\sim \text{Bernoulli}(\pi_{\text{comp}} = 0.75), \\
   Y_i(0) = \mu(X_i) + \epsilon_i, \quad Y_i(1) &= Y_i(0) + W_i(1)\,\tau^{\text{CACE}}(X_i),
\end{split}
\end{align}
with observed values $W_i = Z_i W_i(1)$ and $Y_i$ defined analogously under \gls{SUTVA}. The Gaussian error term reads $\epsilon_i \sim \mathcal{N}(0,1)$. Three features of \eqref{eq:dgp} are central to the estimand of interest. First, $W_i(0) = 0$ imposes one-sided non-compliance, ruling out defiers and always-takers by design and making monotonicity (Assumption~\ref{assump:identification_irregular}(e)) hold exactly. Second, we choose the compliance rate $\pi_{\text{comp}} = 0.75$ to control the strength of the instrument. Third, heterogeneity in the conditional CACE is confined to the first two binary covariates,
\begin{equation}\label{eq:cace_truth}
   \tau^{\text{CACE}}(X_i) =
   \begin{cases}
      \phantom{-}k, & X_i \in l_1 = \{X_{i,1} = 0,\, X_{i,2} = 0\}, \\
      -k,           & X_i \in l_2 = \{X_{i,1} = 1,\, X_{i,2} = 1\}, \\
      \phantom{-}0, & \text{otherwise},
   \end{cases}
\end{equation}
with effect size $k \in \{0, 0.2, 0.4, \dots, 2\}$.\footnote{We use $l_1, l_2$ to denote the true heterogeneity subgroups in the \gls{DGP}, distinct from the discovered subgroups $\{\mathbb{X}_j\}_j$ produced by the CART step of \gls{sbcf-iv}. A successful run recovers $l_1$ and $l_2$ as leaves of the tree, i.e., $\mathbb{X}_j = l_j$ for $j \in \{1, 2\}$ up to labeling.} The remaining covariates, including all continuous ones, are irrelevant for $\tau^{\text{CACE}}(x)$, so the true sparsity level in the treatment effect component grows with $P$. This isolates the setting \gls{sbcf-iv} is designed for: relevance concentrated on two covariates while $P - 2$ noise variables compete for splits. The control function $\mu(X_i)$ is adopted from \citet{caron_shrinkage_2022} and depends only on three continuous covariates through a nonlinear combination of sine, quadratic, and absolute-value terms. Its explicit form is given in Appendix~\ref{append:sim_design}.

We evaluate \gls{sbcf-iv} along three complementary dimensions: tree-level subgroup detection, unit-level classification, and unit-level estimation precision with uncertainty quantification. Tree-level recovery is captured by the \gls{DR} and \gls{FDR}, which record whether the true heterogeneity regions $l_1, l_2$ are recovered. Unit-level classification is assessed on $\mathcal{I}_{\text{inf}}$ via Recall, Precision, False Positive Rate (FPR), and $F$-score, translating structural detection into how reliably individual observations are sorted into statistically significant leaves. Estimation precision is then quantified by the per-unit bias, MSE, and $95\%$ coverage of the resulting conditional CACE estimates. Formal definitions and further discussions are deferred to Appendix~\ref{append:sim_metrics}.

%
Figure~\ref{fig:dr_fdr} shows that \gls{sbcf-iv} dominates \gls{bcf-iv} on both tree-level criteria across all three covariate dimensions. For detection, \gls{sbcf-iv}'s \gls{DR} rises steeply once $k \geq 0.8$ and saturates near one for $k \geq 1.4$, essentially uniformly in $P$: adding irrelevant covariates does not meaningfully degrade the algorithm's ability to recover $l_1$ and $l_2$ as leaves. \gls{bcf-iv}, by contrast, only begins to detect the true subgroups for $k \geq 1$ in the low-dimensional case ($P = 10$), and its \gls{DR} flattens near 0.3 at $k = 2$. For $P \in \{50, 100\}$, \gls{bcf-iv}'s \gls{DR} stays close to zero throughout the grid, indicating that the uniform split-variable prior fails to concentrate on the two binary covariates that drive heterogeneity once they are buried among many noise variables. The ordering on \gls{FDR} is equally clear. \gls{sbcf-iv} maintains an \gls{FDR} near zero across the entire $(k, P)$ grid, so its detection gains do not come at the cost of spurious discoveries. \gls{bcf-iv}, by contrast, exhibits an \gls{FDR} that rises with $k$ and is largest at $P=10$, reaching roughly $0.75$ at $k = 2$ with $P = 10$ and remaining above $0.25$ for larger $P$. Most of \gls{bcf-iv}'s rare "discoveries" in high-dimensional settings are false positives. The two panels indicate that the sparsity-inducing prior improves both detection and false-discovery control, recovering true heterogeneity where \gls{bcf-iv} misses it while controlling spurious leaves where \gls{bcf-iv} generates them.
\begin{figure}[h]
   \centering
   \caption{Tree-level subgroup detection across effect size $k$ and covariate dimension $P$.} 
    \label{fig:dr_fdr}
   \resizebox{0.9\linewidth}{!}{%
     \input{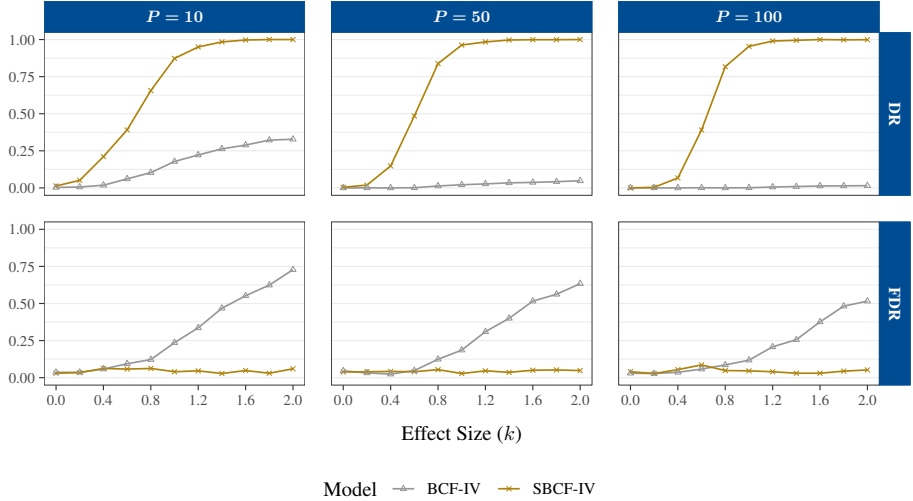}%
   }
\begin{notes}
    The top row reports the Detection Rate (DR), defined in \eqref{eq:DR} as the average share of true heterogeneity subgroups $l_1, l_2$ recovered as leaves of the discovered tree. The bottom row reports the False Detection Rate (FDR), defined in \eqref{eq:FDR} as the share of replications in which at least one spurious leaf is flagged as significant at level $\alpha = 0.05$. Results are averaged over $M=500$ Monte Carlo replications with $N = 1{,}000$, and compare \gls{sbcf-iv} (orange) with \gls{bcf-iv} (gray). Full results in Table \ref{tab:rule} of Appendix \ref{append:full_precision_results}.
\end{notes}
\end{figure}

Figure~\ref{fig:precision.Fscore} confirms that \gls{sbcf-iv}'s tree-level advantages in Figure~\ref{fig:dr_fdr} translate directly into sharper unit-level sorting. Here, we report Precision and the $F$-score. The full set of classification metrics is reported in Figure \ref{fig:full_classification} of Appendix~\ref{append:supmat_simstudy}. Precision follows the same pattern across both algorithms: for $k \geq 1$, \gls{sbcf-iv}'s Precision rises steeply and saturates near one by $k = 1.4$ across all three values of $P$, meaning that nearly every unit assigned to a significant leaf is in a truly heterogeneous subgroup. \gls{bcf-iv} reaches roughly $\text{Precision} = 0.4$ at $k = 2$ under $P = 10$ and stays close to zero for $P \in \{50, 100\}$. In the high-dimensional regime, nearly all of \gls{bcf-iv}'s unit-level positive classifications are misclassifications, which is a consequence of \gls{FDR} results in Figure \ref{fig:dr_fdr}. The $F$-score mirrors this ranking and is the more informative single summary since it penalizes both missed heterogeneity and spurious flags. \gls{sbcf-iv}'s $F$-score converges to one for $k \geq 1.4$ uniformly in $P$, indicating that the algorithm both identifies the correct units and avoids misclassification. \gls{bcf-iv}'s $F$-score remains below $0.4$ throughout the $(k, P)$ grid and flatlines near zero once $P \geq 50$. 

\begin{figure}[h]
   \centering
   \caption{Unit-level classification performance across effect size $k$ and covariate dimension $P$.}
    \label{fig:precision.Fscore}
   \resizebox{0.9\linewidth}{!}{%
   \input{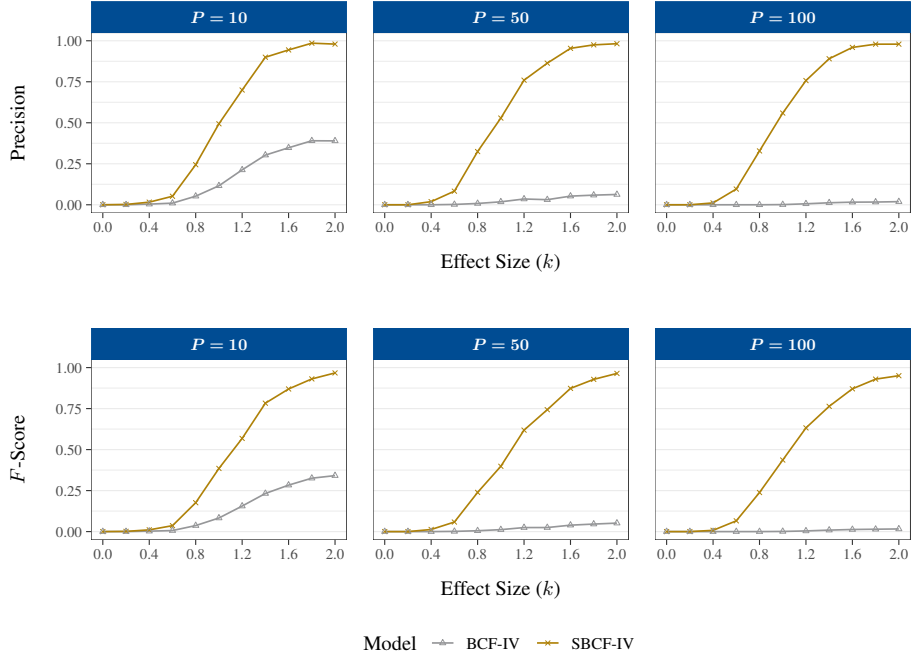}
   }
    \begin{notes}
        The top row reports Precision and the bottom row the $F$-score, both defined in \eqref{eq:class_metrics} and computed over all units in $\mathcal{I}_{\text{inf}}$ with significance evaluated at $\alpha = 0.05$ using Holm-adjusted $p$-values. Results are averaged over $M$ Monte Carlo replications with $N = 1{,}000$, and compare \gls{sbcf-iv} (orange) with \gls{bcf-iv} (gray). We report full results in Table \ref{tab:full_classification} of Appendix \ref{append:full_precision_results}.
    \end{notes}
\end{figure}
%

Table~\ref{tab:short_precision_uncertainty} reports estimation precision and uncertainty quantification for $\widehat\tau^{\text{CACE}}(x)$ and $k \in \{0, 1, 2\}$, while qualitatively similar results for the full $k$-range are deferred to Appendix~\ref{append:full_precision_results}. We summarize three emerging patterns. 
First, point estimation accuracy diverges with effect size. For $k \geq 1$, \gls{sbcf-iv} consistently delivers lower MSE and absolute bias than \gls{bcf-iv}, with the gap widening as $k$ grows: \gls{sbcf-iv}'s MSE is about a quarter of BCF-IV's at $(k, P) = (2, 10)$ and about a tenth at $(k, P) = (2, 100)$.
\gls{bcf-iv}'s MSE grows with $P$, while \gls{sbcf-iv}'s remains flat, confirming that the dimension-invariance seen in Figures~\ref{fig:dr_fdr} and~\ref{fig:precision.Fscore} extends to point estimation accuracy. 
Second, \gls{bcf-iv}'s absolute bias grows from $0.29$ at $(k, P) = (0, 10)$ to $1.5$ at $(k, P) = (2, 100)$ because its uniform split-variable prior dilutes the estimated effect across noise covariates. \gls{sbcf-iv}'s absolute bias stays below $0.5$ throughout the grid. 
Third, uncertainty quantification degrades with increasing $k$ and $P$ for \gls{bcf-iv}. \gls{sbcf-iv} maintains coverage near the nominal $0.95$ level across the entire $(k, P)$ grid, so its confidence intervals remain well-calibrated even when the true effect is large or the covariate space is high-dimensional. \gls{bcf-iv}'s coverage, by contrast, drops with $k$. Taken together, the three simulation dimensions indicate that the sparsity prior improves discovery, classification, point estimation, and inference simultaneously, with \gls{sbcf-iv} delivering near-nominal coverage at comparable interval width.
\begin{table}[H]
\centering
\caption{Estimation precision and uncertainty quantification for $\widehat{\tau}^{\text{CACE}}(x)$.}
\label{tab:short_precision_uncertainty}
\centering
\resizebox{\ifdim\width>\linewidth\linewidth\else\width\fi}{!}{
\begin{tabular}[t]{ll>{\raggedleft\arraybackslash}p{0.25cm}rrrrr>{\raggedleft\arraybackslash}p{0.25cm}rrrrr}
\toprule
\multicolumn{1}{c}{} & \multicolumn{1}{c}{} & \multicolumn{1}{c}{} & \multicolumn{5}{c}{\textbf{BCF-IV}} & \multicolumn{1}{c}{} & \multicolumn{5}{c}{\textbf{SBCF-IV}} \\
\cmidrule(l{3pt}r{3pt}){4-8} \cmidrule(l{3pt}r{3pt}){10-14}
$P$ & $k$ &  & \multicolumn{1}{c}{MSE} & \multicolumn{1}{c}{Bias} & \multicolumn{1}{c}{Abs. Bias} & \multicolumn{1}{c}{Coverage} & \multicolumn{1}{c}{Width} &  & \multicolumn{1}{c}{MSE} & \multicolumn{1}{c}{Bias} & \multicolumn{1}{c}{Abs. Bias} & \multicolumn{1}{c}{Coverage} & \multicolumn{1}{c}{Width}\\
\midrule
 & 0 &  & \makecell{0.188 (0.163)} & \makecell{-0.004 (0.225)} & \makecell{0.286 (0.140)} & \makecell{0.948 (0.176)} & \makecell{1.434 (0.219)} &  & \makecell{0.191 (0.161)} & \makecell{-0.004 (0.222)} & \makecell{0.314 (0.131)} & \makecell{0.957 (0.131)} & \makecell{1.600 (0.176)}\\





 & 1 &  & \makecell{0.473 (0.333)} & \makecell{0.000 (0.286)} & \makecell{0.569 (0.240)} & \makecell{0.695 (0.337)} & \makecell{1.600 (0.210)} &  & \makecell{0.298 (0.327)} & \makecell{-0.002 (0.308)} & \makecell{0.426 (0.256)} & \makecell{0.899 (0.235)} & \makecell{1.794 (0.160)}\\





\multirow{-3}{*}{\raggedright\arraybackslash 10} & 2 &  & \makecell{0.931 (0.729)} & \makecell{0.014 (0.446)} & \makecell{0.811 (0.354)} & \makecell{0.478 (0.335)} & \makecell{1.697 (0.174)} &  & \makecell{0.214 (0.211)} & \makecell{0.012 (0.337)} & \makecell{0.368 (0.199)} & \makecell{0.953 (0.149)} & \makecell{1.840 (0.141)}\\
\cmidrule{1-14}
 & 0 &  & \makecell{0.205 (0.244)} & \makecell{-0.003 (0.223)} & \makecell{0.273 (0.145)} & \makecell{0.953 (0.178)} & \makecell{1.399 (0.244)} &  & \makecell{0.214 (0.151)} & \makecell{-0.003 (0.222)} & \makecell{0.358 (0.140)} & \makecell{0.960 (0.111)} & \makecell{1.769 (0.143)}\\





 & 1 &  & \makecell{0.785 (0.472)} & \makecell{0.012 (0.241)} & \makecell{0.744 (0.288)} & \makecell{0.454 (0.404)} & \makecell{1.495 (0.232)} &  & \makecell{0.213 (0.232)} & \makecell{-0.003 (0.309)} & \makecell{0.362 (0.203)} & \makecell{0.949 (0.151)} & \makecell{1.830 (0.140)}\\





\multirow{-3}{*}{\raggedright\arraybackslash 50} & 2 &  & \makecell{2.088 (1.446)} & \makecell{-0.008 (0.299)} & \makecell{1.291 (0.508)} & \makecell{0.202 (0.256)} & \makecell{1.593 (0.255)} &  & \makecell{0.204 (0.203)} & \makecell{0.002 (0.320)} & \makecell{0.359 (0.192)} & \makecell{0.957 (0.145)} & \makecell{1.836 (0.138)}\\
\cmidrule{1-14}
 & 0 &  & \makecell{0.196 (0.234)} & \makecell{-0.020 (0.226)} & \makecell{0.253 (0.143)} & \makecell{0.957 (0.180)} & \makecell{1.325 (0.215)} &  & \makecell{0.217 (0.162)} & \makecell{-0.021 (0.226)} & \makecell{0.360 (0.143)} & \makecell{0.951 (0.111)} & \makecell{1.785 (0.136)}\\





 & 1 &  & \makecell{0.954 (0.418)} & \makecell{0.019 (0.220)} & \makecell{0.872 (0.251)} & \makecell{0.273 (0.351)} & \makecell{1.417 (0.243)} &  & \makecell{0.230 (0.215)} & \makecell{0.005 (0.319)} & \makecell{0.387 (0.199)} & \makecell{0.950 (0.149)} & \makecell{1.826 (0.139)}\\





\multirow{-3}{*}{\raggedright\arraybackslash 100} & 2 &  & \makecell{2.551 (1.522)} & \makecell{0.013 (0.257)} & \makecell{1.454 (0.526)} & \makecell{0.157 (0.222)} & \makecell{1.535 (0.235)} &  & \makecell{0.230 (0.238)} & \makecell{0.022 (0.340)} & \makecell{0.376 (0.212)} & \makecell{0.942 (0.164)} & \makecell{1.835 (0.146)}\\
\bottomrule
\end{tabular}}
\begin{notes}
    MSE, Bias, Absolute Bias, as well as Coverage and Width of 95\% confidence intervals. All values are averaged over $M = 500$ Monte Carlo replications with $N = 1{,}000$ and standard deviations across Monte Carlo replications are reported in parentheses. Results pool subgroups $l_1$ and $l_2$ and compare BCF-IV (left) with SBCF-IV (right). We focus on $k \in \{ 0, 1, 2\}$ and $P \in \{10, 50, 100\}$ for visibility. Table \ref{tab:full_precision_results} of Appendix~\ref{append:full_precision_results} reports full results.
\end{notes}
\end{table}

To check whether the improvements above are specific to the comparison of \gls{sbcf-iv} with \gls{bcf-iv} or to the cost weighting of \eqref{eq:cost}, Appendix~\ref{append:further_precision_results} reports an analysis that considers the instrumental forest of the \gls{grf} framework as an additional benchmark method \citep{athey_generalized_2019}. We further cross each method that estimates the heterogeneity signal (\gls{bcf}, \gls{sbcf}, \gls{grf}) in line four of Algorithm \ref{alg:sbcf-iv} with the split-frequency cost of \eqref{eq:cost}. Table \ref{tab:competitors} provides an overview of all six methods as variants of Algorithm \ref{alg:sbcf-iv}. The simulation shows that the sparsity-inducing prior, not the cost weighting, is the major source of the improvement. The two \gls{sbcf-iv} variants (with and without cost weighting) are nearly indistinguishable across the $(k,P)$ grid, consistent with Corollary~\ref{cor:coherence}. In contrast, \gls{grf-iv} degrades with $P$ similary to \gls{bcf-iv}. 
The cost weighting helps the non-sparse methods that are based on either \gls{grf} or \gls{bcf}, for which it partially restores detection and coverage while improving precision of the final estimates but without matching the performance of \gls{sbcf-iv}. We present full tables and figures in Appendix~\ref{append:further_precision_results}.


\section{Empirical application}
\label{ch:emp_appl}

We apply \gls{sbcf-iv} to two empirical datasets, the Oregon Health Insurance Experiment (OHIE) \citep{finkelstein_oregon_2012, johnson_detecting_2022} and the 401(k) dataset \citep{poterba_401k_1992, chernozhukov_doubledebiased_2018, belloni_program_2017}. The OHIE was a randomized controlled trial conducted in 2008 to assess the effects of expanding Medicaid coverage on health outcomes, financial security, and healthcare utilization. The US state Oregon used a lottery system to allocate a limited number of Medicaid spots to uninsured, low-income adults \citep{finkelstein_oregon_2012, johnson_detecting_2022}. This created a natural experiment, allowing researchers to compare those who received Medicaid to those who did not. Following \citet{johnson_detecting_2022}, we apply \gls{sbcf-iv} to the OHIE data. The outcome is the number of days in the past month on which poor physical or mental health did not impair the respondent's usual activities (a survey count between 0 and 30, with higher values indicating better health). \gls{sbcf-iv} identifies a single complier subgroup with a positive, statistically significant Medicaid effect: English-preferring individuals aged between 38 and 59 (effect: $2.263$; adjusted $p$-value: $0.0945$), in contrast to the two subgroups reported by \citet{johnson_detecting_2022}. Estimated subgroup compliance rates range from $19\%$ to $32\%$, in line with \citet{johnson_detecting_2022} and exhibiting only modest variation relative to the heterogeneity in conditional \gls{itt} effects. Both analyses thus locate the dominant source of complier-effect heterogeneity in the variation of the \gls{ccace} numerator rather than in compliance probabilities. Appendix \ref{sec:OHIE} re-investigates the empirical analysis presented in \citet{johnson_detecting_2022} in more detail.

On the 401(k) data, \gls{sbcf-iv} partitions eligible households into seven subgroups, of which two possess adjusted $p$-values at the 10\% level: the bulk subgroup with income below \$68{,}810 (89\% of the inference sample, $\widehat{\tau}^{\text{CACE}}(x) = \$17{,}818$) and a small upper-middle-income subgroup with household income between $\$92{,}690$ and $\$110{,}400$ with $\widehat{\tau}^{\text{CACE}}(x) = \$53{,}422$. 
The remaining five leaves show adjusted $p$-values without statistical significance. Four carry inference-sample shares at or below 1\% and the fifth below 7\%, which we read as small-sample artifacts.
The discovered splits run almost entirely along income, aligning with the lifecycle and earnings-gradient emphasis of the classical 401(k) saving literature \citep{poterba_401k_1992, engen_effects_2000}. As in the OHIE application, complier shares vary far less than the conditional CACE across leaves. The high-mass leaves cluster near 
 a complier share of 0.70, so the heterogeneity is driven by variation in the conditional \gls{itt} numerator rather than in the complier denominator. This mirrors the simulation evidence in Section~\ref{ch:sim_study} that \gls{sbcf-iv} recovers \gls{itt}-driven partitions reliably. Both statistically significant leaf estimates exceed the overall cross-fitted
\gls{cace} of $\$9{,}000$--$\$13{,}000$ reported by \citet{chernozhukov_doubledebiased_2018}, the bulk leaf modestly (by $\$5{,}000$--$\$9{,}000$) and the upper-middle-income leaf substantially. We present this as a descriptive comparison only: the gap
is consistent with several explanations (genuine
income-based heterogeneity in the \gls{cace}, post-selection effects from
reporting leaves after screening, and the
upward bias from unobserved saver heterogeneity flagged by
\citet{engen_illusory_1996}) which the present design cannot separate. We provide an extended discussion of this second empirical application in Appendix~\ref{sec:401k}.

\section{Conclusion}
\label{ch:conclusion}

This paper introduces \gls{sbcf-iv}, an extension of \gls{bcf-iv} \citep{bargagli-stoffi_heterogeneous_2022} that combines Bayesian shrinkage with subgroup-level estimation of the conditional CACE under imperfect compliance in high-dimensional covariate settings. \gls{sbcf-iv} preserves the interpretive structure of \gls{bcf-iv} with an explicit partition over compliers and subgroup-level \gls{cace} while introducing regularization needed to make that structure stable in high dimensions \citep{caron_shrinkage_2022}. The simulation study in Section \ref{ch:sim_study} demonstrates for varying effect sizes $k$ that, as the covariate dimension $P$ grows, \gls{sbcf-iv} preserves subgroup discovery rates, precision, and nominal coverage. In comparison, \gls{bcf-iv} exhibits a degradation along most criteria, leading to a loss of inferential reliability. Applied to two empirical applications, \gls{sbcf-iv} recovers interpretable partitions of compliers whose \gls{ccace} estimates differ across subgroups, illustrating its practical value beyond the simulated regime.

Several limitations warrant emphasis. First, we establish no selection- or partition-consistency guarantee for the discovery step. The sparsity prior's variable-selection benefit is that of \citet{caron_shrinkage_2022}, which we demonstrate empirically in our \gls{ccace} setting. The discovery--inference coherence we invoke is a population-target alignment rather than a finite-sample recovery result. Second, the current formulation is restricted to binary instruments and binary treatments, and the honest sample split between $\mathcal{I}_{\text{disc}}$ and $\mathcal{I}_{\text{inf}}$ reduces the effective sample size that may hinder subgroup discovery in smaller empirical studies. Extensions to continuous instruments and multi-valued treatments combined with cross-fitting are natural directions for future work.

\begin{ack}
We thank Christoph Hanck and four anonymous referees for valuable comments that improved the paper. The paper has been presented at the Statistical Week 2025 and the ICSDS 2025. We thank all participants for insightful discussions that enhanced the focus of the paper. The authors acknowledge partial financial support from TRR 391 Spatio-temporal Statistics for the Transition of Energy and Transport (520388526) by the Deutsche Forschungsgemeinschaft (DFG, German Research Foundation) and from the Rhine-Ruhr Center for Scientific Data Literacy (DKZ.2R) by the German Federal Ministry of Education and Research (BMBF). We acknowledge the use of Anthropic's Claude Opus 4.8 for proofreading and editorial assistance in preparing this manuscript. All content was checked and verified by the authors, who remain responsible for any errors. Replication code for simulations and both empirical applications is available at \url{https://github.com/jens-klenke/SBART-IV}.
\end{ack}

\bibliographystyle{agsm}
\bibliography{BART_IV_paper}


\appendix

\section{Identification of the conditional CACE}
\label{append:proof_cCACE}

This appendix proves Proposition~\ref{prop:identification}: under Assumption~\ref{assump:identification_irregular}, the complier estimand $\tau^{\text{CACE}}(x)$ of Definition~\ref{defn:cCACE} equals the Wald ratio $\text{ITT}_Y(x)/\pi_C(x)$ and is therefore identified from the observed distribution of $(Y_i, W_i, Z_i, X_i)$. The argument adapts \citet{angrist_identification_1996} to the conditional target as in \citet{bargagli-stoffi_heterogeneous_2022}.\footnote{Within this appendix, $\mu_z(x)$ denotes an observable conditional mean and is distinct from the BART prognostic control function $\mu(\cdot)$ of the main text.}

For brevity, contrast this irregular assignment setting with the regular assignment mechanism, where unconfoundedness $W_i \perp\!\!\!\perp \bigl(Y_i(0), Y_i(1)\bigr) \mid X_i$ and overlap $0 < \Pr(W_i = 1 \mid X_i = x) < 1$ identify the conditional average treatment effect $\tau(x) = \mathbb{E}[Y_i \mid W_i = 1, X_i = x] - \mathbb{E}[Y_i \mid W_i = 0, X_i = x]$ directly from observed conditional means \citep{imbens_causal_2015}. Under irregular assignment, $W_i$ is potentially confounded and direct identification of $\tau(x)$ fails. Identification of the complier-restricted target $\tau^{\text{CACE}}(x)$ uses the IV approach below.

\begin{proof}
Fix $x$ in the support of $X$ and define the observable nuisance functions
\begin{align*}
   \mu_z(x) &\coloneqq \mathbb{E}[Y_i \mid Z_i = z, X_i = x], \\
   \delta_z(x) &\coloneqq \mathbb{E}[W_i \mid Z_i = z, X_i = x],\\
   \qquad z &\in \{0,1\}.
\end{align*}
Both are identified from the observed distribution. By Definition~\ref{defn:cITT}, $\text{ITT}_Y(x) = \mu_1(x) - \mu_0(x)$, and we will show that $\pi_C(x) = \delta_1(x) - \delta_0(x)$ under monotonicity. The remaining task is to establish
\begin{equation}\label{eq:target}
   \mu_1(x) - \mu_0(x) \;=\; \tau^{\text{CACE}}(x)\,\bigl(\delta_1(x) - \delta_0(x)\bigr).
\end{equation}

Let us express $\mu_z(x)$ in potential outcomes. By Assumption~\ref{assump:identification_irregular}(a) (consistency) and (c) (unconfounded instrument),
\[
   \mu_z(x) \;=\; \mathbb{E}[Y_i(z, W_i(z)) \mid Z_i = z, X_i = x] \;=\; \mathbb{E}[Y_i(z, W_i(z)) \mid X_i = x].
\]
Hence
\begin{align}
   \mu_1(x) - \mu_0(x) 
   &= \mathbb{E}\bigl[Y_i(1, W_i(1)) - Y_i(0, W_i(0)) \bigm| X_i = x\bigr] \notag \\
   &\overset{\text{(d)}}{=} \mathbb{E}\bigl[Y_i(W_i(1)) - Y_i(W_i(0)) \bigm| X_i = x\bigr], \label{eq:step1}
\end{align}
where the second equality invokes the exclusion restriction. Since $W_i(z) \in \{0,1\}$, a case analysis over the four values of $(W_i(0), W_i(1))$ yields
\begin{equation}\label{eq:algebraic}
   Y_i(W_i(1)) - Y_i(W_i(0)) \;=\; \bigl[Y_i(1) - Y_i(0)\bigr]\bigl[W_i(1) - W_i(0)\bigr],
\end{equation}
which requires no identifying assumption beyond binarity of $W_i$. Substituting \eqref{eq:algebraic} into \eqref{eq:step1},
\begin{equation}\label{eq:step2}
   \mu_1(x) - \mu_0(x) \;=\; \mathbb{E}\bigl[\{Y_i(1) - Y_i(0)\}\{W_i(1) - W_i(0)\} \bigm| X_i = x\bigr].
\end{equation}

By Assumption~\ref{assump:identification_irregular}(e), $W_i(1) - W_i(0) \in \{0, 1\}$, so $W_i(1) - W_i(0) = \mathbf{1}\{G_i = C\}$. Combining with \eqref{eq:step2} and iterating the expectation,
\begin{equation}\label{eq:step3}
   \mu_1(x) - \mu_0(x) 
   \;=\; \mathbb{E}\bigl[Y_i(1) - Y_i(0) \bigm| G_i = C,\, X_i = x\bigr] \cdot \pi_C(x)
   \;=\; \tau^{\text{CACE}}(x) \cdot \pi_C(x),
\end{equation}
where the second equality uses Definition~\ref{defn:cCACE} together with (d) (so that $Y_i(w)$ is unambiguous on the complier subpopulation).
By consistency and unconfoundedness in Assumption~\ref{assump:identification_irregular}, $\delta_z(x) = \mathbb{E}[W_i(z) \mid X_i = x]$. Under monotonicity in Assumption~\ref{assump:identification_irregular}(e),
\begin{equation}\label{eq:pi_C}
   \delta_1(x) - \delta_0(x) \;=\; \mathbb{E}[W_i(1) - W_i(0) \mid X_i = x] \;=\; \Pr(G_i = C \mid X_i = x) \;=\; \pi_C(x),
\end{equation}
and Assumption~\ref{assump:identification_irregular}(b) ensures $\pi_C(x) > 0$ almost surely, so the ratio in \eqref{eq:target} is well-defined.
Combining \eqref{eq:step3} and \eqref{eq:pi_C} yields \eqref{eq:target}, hence
\begin{align}
   \tau^{\text{CACE}}(x) \;=\; \frac{\mu_1(x) - \mu_0(x)}{\delta_1(x) - \delta_0(x)} \;=\; \frac{\text{ITT}_Y(x)}{\pi_C(x)},
\end{align}
which is identified from the observed distribution of $(Y_i, W_i, Z_i, X_i)$.
\end{proof}

We further note that the cost anchor in \eqref{eq:cost} rests on a simple relationship between the two within-leaf summaries our procedure in Algorithm \ref{alg:sbcf-iv} uses. The shallow CART fitted to $\widehat\tau^{\,\text{SBCF}}(x)$ in Algorithm \ref{alg:sbcf-iv} assigns leaf $\mathbb{X}_j$ the covariate-averaged effect $A_j := \mathbb{E}[\tau^{\text{CACE}}(X_i)\mid X_i\in\mathbb{X}_j]$, an average of pointwise ratios. Subgroup 2SLS targets $\tau^{\text{CACE}}_{\mathbb{X}_j}$ of \eqref{eq:subgroup-cace}, the complier-averaged (ratio-of-averages) effect.

\begin{cor}[Discovery--inference coherence]\label{cor:coherence}
Let $\{\mathbb{X}_j\}_j$ partition the covariate space and fix a leaf $\mathbb{X}_j$ with
$\mathbb{E}_j[\pi_C(X_i)] := \mathbb{E}[\pi_C(X_i)\mid X_i\in\mathbb{X}_j] > 0$. Let
$A_j := \mathbb{E}[\tau^{\text{CACE}}(X_i)\mid X_i\in\mathbb{X}_j]$ be the covariate-averaged complier effect targeted by the discovery step in Algorithm \ref{alg:sbcf-iv} (the leaf mean of the pointwise signal
$\widehat\tau^{\,\text{SBCF}}(x)$). Let $\tau^{\text{CACE}}_{\mathbb{X}_j}$ of
\eqref{eq:subgroup-cace} be the complier-averaged effect targeted by subgroup 2SLS. Under
Assumption~\ref{assump:identification_irregular},
\[
  A_j - \tau^{\text{CACE}}_{\mathbb{X}_j}
  = -\,\frac{\operatorname{Cov}_j\!\big(\pi_C(X_i),\,\tau^{\text{CACE}}(X_i)\big)}{\mathbb{E}_j[\pi_C(X_i)]}.
\]
In particular, $A_j = \tau^{\text{CACE}}_{\mathbb{X}_j}$ whenever, within the leaf, compliance and complier effect are uncorrelated (i.e., if either is conditionally constant).
\end{cor}

\begin{proof}
By Proposition~\ref{prop:identification}, $\text{ITT}_Y(x)=\pi_C(x)\,\tau^{\text{CACE}}(x)$.
Taking the expectation operator $\mathbb{E}_j$ and applying $\mathbb{E}_j[\pi_C(X_i)\,\tau^{\text{CACE}}(X_i)]
= \mathbb{E}_j[\pi_C(X_i)]\,\mathbb{E}_j[\tau^{\text{CACE}}(X_i)] + \operatorname{Cov}_j(\pi_C(X_i),\tau^{\text{CACE}}(X_i))$ yields
\begin{align*}
\tau^{\text{CACE}}_{\mathbb{X}_j} &= \mathbb{E}_j[\text{ITT}_Y(X_i)]/\mathbb{E}_j[\pi_C(X_i)] \\
&= A_j + \operatorname{Cov}_j(\pi_C(X_i),\tau^{\text{CACE}}(X_i))/\mathbb{E}_j[\pi_C(X_i)],
\end{align*}
with $\mathbb{E}_j[\pi_C(X_i)]>0$ by Assumption~\ref{assump:identification_irregular}(b).
\end{proof}

Corollary~\ref{cor:coherence} is the leaf-level form of the fact that an aggregated Wald estimand is a compliance-weighted average of covariate-specific complier effects \citep{angrist_twostage_1995, abadie_semiparametric_2003, frolich_nonparametric_2007}.
It concerns the alignment of population targets and does not establish selection or partition consistency of the discovery step, which we do not claim.

\section{Conditional 2SLS: reduced form and asymptotic properties}
\label{append:theorem_2SLS}

This appendix collects the reduced form of the subgroup-restricted simultaneous system \eqref{eq:simultaneous_equations} and the asymptotic properties of the 2SLS estimator $\widehat\tau^{\,\text{2SLS}}_{\mathbb{X}_j}$ defined in Definition~\ref{defn:cCACE_estimator}. The results restate the unconditional 2SLS theory of \citet{wooldridge_2010} with identification of the unconditional CACE following \citet{angrist_identification_1996} and \citet{imbens_causal_2015}, transferred to the conditional target $\tau^{\text{CACE}}_{\mathbb{X}_j}$ \citep{bargagli-stoffi_heterogeneous_2022}.

Substituting the first-stage equation of \eqref{eq:simultaneous_equations} into the outcome equation gives the reduced form
\begin{align*}
   Y_{i, \mathbb{X}_j} 
   &= \bigl( \kappa_{\mathbb{X}_j} + \tau^{\text{CACE}}_{\mathbb{X}_j}\,\pi_{0, \mathbb{X}_j} \bigr) 
    + \underbrace{\bigl( \tau^{\text{CACE}}_{\mathbb{X}_j}\,\pi_{C, \mathbb{X}_j} \bigr)}_{\gamma_{\mathbb{X}_j}} Z_{i, \mathbb{X}_j} 
    + \bigl( \varepsilon_{i, \mathbb{X}_j} + \tau^{\text{CACE}}_{\mathbb{X}_j}\,\nu_{i, \mathbb{X}_j} \bigr).
\end{align*}
Under $\mathbb{E}[\varepsilon_{i,\mathbb{X}_j}] = \mathbb{E}[\nu_{i, \mathbb{X}_j}] = 0$ and $\mathbb{E}[Z_{i,\mathbb{X}_j}\,\nu_{i, \mathbb{X}_j}] = \mathbb{E}[Z_{i, \mathbb{X}_j}\,\varepsilon_{i, \mathbb{X}_j}] = 0$, Ordinary Least Squares (OLS) consistently estimates $\pi_{C, \mathbb{X}_j}$ and $\gamma_{\mathbb{X}_j}$, such that
\begin{equation}\label{eq:2sls_ratio}
   \widehat\tau^{\,\text{2SLS}}_{\mathbb{X}_j} \;=\; \frac{\widehat\gamma_{\mathbb{X}_j}}{\widehat\pi_{C, \mathbb{X}_j}}
\end{equation}
coincides with the moment-based ratio of Definition~\ref{defn:cCACE_estimator} \citep{wooldridge_2010}.
Let $N_{\mathbb{X}_j}$ denote the number of observations in $\mathbb{X}_j$, and consider the subgroup-restricted moment conditions
\begin{enumerate}[label=(A\arabic*)]
   \item \label{A1} $\mathbb{E}[Z_{i, \mathbb{X}_j}^2] \neq 0$;
   \item \label{A2} $\mathbb{E}[Z_{i, \mathbb{X}_j}\,\varepsilon_{i, \mathbb{X}_j}] = 0$;
   \item \label{A3} $\pi_{C, \mathbb{X}_j} \neq 0$;
   \item \label{A4} $\mathbb{E}[Z_{i, \mathbb{X}_j}^2\,\varepsilon_{i, \mathbb{X}_j}^2] < \infty$.
\end{enumerate}
Conditions~\ref{A1}--\ref{A3} are the subgroup analogs of the standard IV moment restrictions and are implied within $\mathbb{X}_j$ by Assumption~\ref{assump:identification_irregular} together with the exogeneity of $Z_i$.

\begin{thm}[Consistency]\label{thm:consistency_2sls}
   Under \ref{A1}--\ref{A3}, $\widehat\tau^{\,\text{2SLS}}_{\mathbb{X}_j} - \tau^{\text{CACE}}_{\mathbb{X}_j} \overset{p}{\longrightarrow} 0$ as $N_{\mathbb{X}_j} \to \infty$.
\end{thm}

\begin{thm}[Asymptotic normality]\label{thm:normality_2sls}
   Under \ref{A1}--\ref{A4},
   \[
      \sqrt{N_{\mathbb{X}_j}}\,\bigl(\widehat\tau^{\,\text{2SLS}}_{\mathbb{X}_j} - \tau^{\text{CACE}}_{\mathbb{X}_j}\bigr) 
      \;\overset{d}{\longrightarrow}\; \mathcal{N}\bigl(0,\; N_{\mathbb{X}_j}\cdot\mathrm{avar}(\widehat\tau^{\,\text{2SLS}}_{\mathbb{X}_j})\bigr) \qquad \text{as } N_{\mathbb{X}_j} \to \infty,
   \]
   with $\mathrm{avar}(\widehat\tau^{\,\text{2SLS}}_{\mathbb{X}_j})$ the asymptotic variance of the 2SLS estimator, approximated as in \citet{wooldridge_2010}.
\end{thm}

\begin{proof}[Proof sketch]
Both results follow from the unconditional 2SLS case applied to the i.i.d.\ subsample $\{(Y_l, W_l, Z_l) : X_l \in \mathbb{X}_j\}$ \citep{wooldridge_2010, bargagli-stoffi_heterogeneous_2022}. Consistency uses the continuous mapping theorem applied to \eqref{eq:2sls_ratio} together with the law of large numbers for the numerator and denominator while normality follows from the delta method combined with a central limit theorem for $(\widehat\gamma_{\mathbb{X}_j}, \widehat\pi_{C, \mathbb{X}_j})$, where \ref{A4} ensures a finite asymptotic variance.
\end{proof}

We note that Assumption~\ref{A3} is the subgroup-level relevance condition which does not need to follow from overall instrument strength. In \citet{bargagli-stoffi_heterogeneous_2022} and \citet{ding_decomposing_2019}, the approximations in Theorems~\ref{thm:consistency_2sls}--\ref{thm:normality_2sls} require a sufficient number of observations within each $\mathbb{X}_j$. Shallower trees produce subgroups with both higher interpretability and higher statistical power \citep{athey_recursive_2016, lee_discovering_2021}, motivating the trimming of small or weak-instrument nodes in \gls{sbcf-iv}.

\section{Background: BART, BCF, SBCF, and CART}
\label{append:bart_bcf_cart}

This appendix recalls the ingredients on which \gls{sbcf-iv} is built. We keep the exposition minimal and refer for full treatments to the cited references.

\paragraph{Classification and Regression Tree (CART).} The CART algorithm \citep{breiman_classification_1984} recursively partitions the covariate space along axis-aligned splits, choosing at each node the split that maximizes within-node homogeneity of the response. The terminal-node piecewise-constant predictor is interpretable but a high-variance single learner.

\paragraph{Bayesian Additive Regression Trees (BART).} BART \citep{chipman_bart_2010} models the conditional mean as a sum of regression trees, and place three regularizing priors: one penalizing tree depth, one shrinking leaf-level predictions toward the response center, and one bounding the error variance away from zero. Jointly, these priors prevent any single tree from dominating the fit and yield coherent posterior uncertainty. At each split, BART draws the candidate splitting variable according to a vector of selection probabilities $s = (s_1, \ldots, s_P)$, which by default is uniform, $s_j = 1/P$.

\paragraph{Bayesian Causal Forest (BCF).} The BCF of \citet{hahn_bayesian_2020} adapts BART to causal inference under a regular assignment mechanism. Two features carry over to our setting: (i)~the conditional mean of the outcome is decomposed into a control function and a treatment effect function, each assigned an independent BART prior; (ii)~the propensity score is included as a covariate in the control function to mitigate regularization-induced confounding and targeted selection. Different depth-penalty parameters across the two components encourage shallower, more interpretable trees for the treatment effect component. Equation~\eqref{eq:cond_exp_Y} is the IV analog of this decomposition, with $Z_i$ and $\text{ITT}_Y(x)$ replacing the treatment indicator and treatment effect, respectively.

\paragraph{Shrinkage Bayesian Causal Forest (SBCF).} SBCF \citep{caron_shrinkage_2022} augments BCF with a sparsity-inducing Dirichlet prior on the split-variable selection probabilities, replacing BART's default uniform $s_j = 1/P$ with $s \sim \mathrm{Dir}(\alpha/P, \ldots, \alpha/P)$ and a hyperprior $\alpha/(\alpha + \rho) \sim \mathrm{Beta}(a, b)$ on the concentration parameter. Small $\alpha$ concentrates posterior mass on a few covariates. The hyperprior lets the data determine how aggressive this concentration should be, with the default $(a, b, \rho) = (0.5, 1, P)$ strengthening the sparsity preference as $P$ grows. The prior acts at each split independently, so covariates with low posterior inclusion probability are rarely chosen throughout the ensemble. \citet{caron_shrinkage_2022} show that SBCF improves CATE estimation over BCF when many covariates are irrelevant and in confounded-data settings, and that the posterior split frequencies $\widehat{s}$ provide a natural variable-importance summary. \gls{sbcf-iv} inherits both the estimation benefit and the interpretability of $\widehat{s}$: the latter is exploited in \eqref{eq:cost} to steer the subgroup-finding CART toward the covariates the ensemble deemed most relevant.

\section{Prior specifications for SBCF-IV}
\label{append:sbcf_priors}

\gls{sbcf-iv} combines the BCF decomposition in \eqref{eq:cond_exp_Y} and \eqref{eq:cond_exp_W} with the sparsity-inducing Dirichlet prior on split-variable selection probabilities proposed by \citet{caron_shrinkage_2022}. The main text states the prior generically in~\eqref{eq:sbcf_prior}. We record the component-specific hyperparameter choices in the following.

Let $s_\bullet = (s_{\bullet,1}, \ldots, s_{\bullet, P_\bullet})$ denote the split-probability vector for each of the three functions $\mu(e(x), x)$, $\text{ITT}_Y(x)$, and $\delta(z, x)$ in \eqref{eq:cond_exp_Y}--\eqref{eq:cond_exp_W}. We place independent priors
\begin{align}
   s_\mu &\sim \mathrm{Dir}\!\left(\tfrac{\alpha_\mu}{P+1}, \ldots, \tfrac{\alpha_\mu}{P+1}\right), 
   & \tfrac{\alpha_\mu}{\alpha_\mu + \rho_\mu} &\sim \mathrm{Beta}(a, b), 
   & \rho_\mu &= P + 1, \label{eq:prior_mu} \\
   s_{\text{ITT}_Y} &\sim \mathrm{Dir}\!\left(\tfrac{\alpha_{\text{ITT}_Y}}{P}, \ldots, \tfrac{\alpha_{\text{ITT}_Y}}{P}\right), 
   & \tfrac{\alpha_{\text{ITT}_Y}}{\alpha_{\text{ITT}_Y} + \rho_{\text{ITT}_Y}} &\sim \mathrm{Beta}(a, b), 
   & \rho_{\text{ITT}_Y} &= P/2, \label{eq:prior_itt} \\
   s_\delta &\sim \mathrm{Dir}\!\left(\tfrac{\alpha_\delta}{P+1}, \ldots, \tfrac{\alpha_\delta}{P+1}\right),
   & \tfrac{\alpha_\delta}{\alpha_\delta + \rho_\delta} &\sim \mathrm{Beta}(a, b),
   & \rho_\delta &= P+1, \label{eq:prior_delta}
\end{align}
with shared defaults $(a, b) = (0.5, 1)$ throughout. Three comments are in order.

\textit{(i) Dimension of $s_\mu$.} The propensity score $e(x)$ enters $\mu$ as an additional covariate, so $P_\mu = P+1$. The extra dimension is reflected both in the Dirichlet parameter and in the scale $\rho_\mu$ \citep{hahn_bayesian_2020}.

\textit{(ii) Targeted shrinkage via $\rho$.} Under $(a,b) = (0.5, 1)$, the hyperprior on $\alpha/(\alpha + \rho)$ concentrates near zero. Since $\alpha$ is increasing in $\rho$ at any fixed value
of the ratio, a smaller $\rho$ pulls $\alpha$ toward zero and strengthens the Dirichlet's concentration on few covariates.
Setting $\rho_{\text{ITT}_Y} = P/2 < P+1 = \rho_\mu$ imposes a sharper sparsity preference on the treatment effect component than on the prognostic function, consistent with the premise that treatment effect heterogeneity is driven by a smaller subset of covariates than the outcome's prognostic effects \citep{hahn_bayesian_2020, caron_shrinkage_2022}. The case $(a, b) = (1, 1)$ with $\alpha \to \infty$ recovers the uniform $s_{\bullet, j} = 1/P_\bullet$ used by \gls{bcf-iv} \citep{bargagli-stoffi_heterogeneous_2022}.

\textit{(iii) Unified treatment of $Y_i$ and $W_i$.} The instrument $z$
enters $\delta$ as an additional splittable covariate, so $P_\delta = P+1$;
the extra dimension is reflected in the Dirichlet parameter and in the
scale $\rho_\delta$, paralleling the treatment of $e(x)$ in $\mu$.
Specification~\eqref{eq:prior_delta} thus applies the same Dirichlet-based
sparsity mechanism to the compliance model, implemented as a SoftBART
probit in the sense of \citet{hill_bayesian_2011}. No separate
justification is needed: the sparsity concern (only a subset of
covariates drives either outcome heterogeneity or differential compliance) is identical in both components, and the Dirichlet prior is the natural vehicle for it in both cases. Because $\delta(z,x)$ is not further
decomposed as in~\eqref{eq:cond_exp_Y}, a single $s_\delta$ suffices
rather than a pair of component-specific priors.

\section{Supplementary materials for the simulation study}
\label{append:supmat_simstudy}

\subsection{Full simulation design}\label{append:sim_design}

This appendix collects the full specification of the \gls{DGP} summarized in Section~\ref{ch:sim_study}. For each unit $i = 1, \ldots, N$ with $N = 1{,}000$, potential outcomes, potential treatments, and covariates are drawn independently as
\begin{align*}
   Z_i &\sim \text{Bernoulli}(0.5), \\
   W_i(1) &\sim \text{Bernoulli}(\pi_{\text{comp}} = 0.75), \qquad W_i(0) = 0, \\
   W_i &= Z_i\, W_i(1) + (1 - Z_i)\, W_i(0), \\
   Y_i(0) &=  \mu(X_i) + \epsilon_i, \quad \epsilon_i \sim \mathcal{N}(0, 1), \\
   Y_i(1) &= Y_i(0) + W_i(1)\, \tau^{\text{CACE}}(X_i),
\end{align*}
with observed outcome $Y_i = Z_i\, Y_i(1) + (1 - Z_i)\, Y_i(0)$ under \gls{SUTVA} (Assumption~\ref{assump:identification_irregular}(a)). The covariate vector $X_i = (X_{i,1}, \ldots, X_{i,P}) \in \mathbb{R}^P$ consists of $P/2$ binary and $P/2$ continuous independent components,
\begin{equation*}
   X_{i,1}, \ldots, X_{i, P/2} \overset{\text{i.i.d.}}{\sim} \text{Bernoulli}(0.5), \qquad X_{i, P/2+1}, \ldots, X_{i, P} \overset{\text{i.i.d.}}{\sim} \mathcal{N}(0, 1),
\end{equation*}
with $P \in \{10, 50, 100\}$. The three values of $P$ probe increasing degrees of sparsity: only $X_{i,1}$ and $X_{i,2}$ drive heterogeneity in $\tau^{\text{CACE}}(x)$, so the share of relevant covariates in the treatment effect component is $2/P \in \{0.2, 0.04, 0.02\}$. The compliance rate $\pi_{\text{comp}} = 0.75$ governs the strength of $Z_i$ as an instrument for $W_i$. One-sided non-compliance is built in via $W_i(0) = 0$: units not assigned to treatment cannot receive it, so the population decomposes into compliers ($W_i(1) = 1$) and never-takers ($W_i(1) = 0$), with $\pi_C = \pi_{\text{comp}}$ and $\pi_{NT} = 1 - \pi_{\text{comp}}$. Defiers and always-takers are ruled out by design, and the monotonicity assumption~\ref{assump:identification_irregular}(e) holds with equality for the never-takers and strictly for the compliers.

The conditional CACE follows the setup in \citep{bargagli-stoffi_heterogeneous_2022} and is piecewise constant over three regions of the binary covariate space,
\begin{equation}\label{eq:tau_CACE_sim}
   \tau^{\text{CACE}}(X_i) =
   \begin{cases}
      \phantom{-}k, & X_i \in l_1 = \{X_{i,1} = 0,\, X_{i,2} = 0\}, \\
      -k,           & X_i \in l_2 = \{X_{i,1} = 1,\, X_{i,2} = 1\}, \\
      \phantom{-}0, & X_i \in l_0 = \{X_i \notin l_1 \cup l_2\},
   \end{cases}
\end{equation}
with effect size $k \in \{0, 0.2, 0.4, 0.6, 0.8, 1, 1.2, 1.4, 1.6, 1.8, 2\}$ to distinguish no, moderate, and large heterogeneity regimes. The two non-null subgroups $l_1$ and $l_2$ each have population mass $0.25$ under the Bernoulli distribution of $(X_{i,1}, X_{i,2})$, and carry opposite-signed effects, so the marginal CACE averages to zero. This design makes a naive marginal estimator uninformative and forces any algorithm targeting $\tau^{\text{CACE}}(x)$ to recover the heterogeneity structure. The subgroups $l_1, l_2, l_0$ are ground-truth objects defined by the \gls{DGP} and we distinguish them from the data-driven partition $\{\mathbb{X}_j\}_j$ produced by \gls{sbcf-iv} on $\mathcal{I}_{\text{disc}}$, whose leaves approximate $l_1$ and $l_2$ when discovery succeeds.
The baseline prognostic outcome function $\mu(X_i)$, adopted from \citet{caron_shrinkage_2022}, depends only on the first three continuous covariates and takes the explicit form
\begin{equation}\label{eq:mu_control}
   \mu(X_i) \;=\; 3 + 1.5\sin\!\bigl(\pi\, X_{i, P/2+1}\bigr) + 0.5\,\bigl(X_{i, P/2+2} - 0.5\bigr)^2 + 1.5\,\bigl(2 - |X_{i, P/2+3}|\bigr).
\end{equation}
The three terms introduce nonlinearity, curvature, and non-differentiability into the control component, providing a realistic stress test for the control-function BART prior on $\mu(e(x), x)$ in \eqref{eq:cond_exp_Y}. The relevant covariates for $\mu(X_i)$ in Equation \eqref{eq:mu_control} and for $\tau^{\text{CACE}}(X_i)$ in Equation \eqref{eq:tau_CACE_sim} are disjoint (continuous covariates for the former, binary for the latter), which allows us to assess the discovery ability of \gls{sbcf-iv} cleanly, without confounding the signal in the treatment effect component with the signal in the control function.

The full Monte Carlo design varies $(k, P)$ on the grid $\{0, 0.2, 0.4, 0.6, 0.8, 1, 1.2, 1.4, 1.6, 1.8, 2\} \times \{10, 50, 100\}$, yielding $33$ scenarios. Each scenario is replicated $M=500$ times, and performance is summarized via the classification metrics (Recall, Precision, $F$-score, FPR) for correctly identifying $l_1$ and $l_2$, together with MSE, bias, and coverage for the estimated conditional CACE.

\subsection{Performance measures for the simulation study}\label{append:sim_metrics}

Our evaluation follows the framework of \citet{bargagli-stoffi_heterogeneous_2022} but refines it by separating structural recovery from inferential significance. We benchmark structural discovery with tree-level metrics, while unit-level metrics provide the evaluation criteria relevant for applications where the true partition is unobserved. 
Thus, we evaluate \gls{sbcf-iv} and \gls{bcf-iv} along three dimensions: (i) tree-level subgroup detection, measuring whether the true heterogeneity regions $l_1, l_2$ appear as leaves of the discovered tree; (ii) unit-level individual classification, measuring whether each observation in $\mathcal{I}_{\text{inf}}$ is assigned to a correctly signed leaf; and (iii) unit-level estimation precision of the resulting CACE estimates. All metrics are computed on $\mathcal{I}_{\text{inf}}$ in each simulation run $m = 1, \ldots, M$ and averaged across runs. Throughout, $\mathcal{T}_m$ denotes the tree with corresponding subgroups discovered in run $m$, with leaves $\{\mathbb{X}_j(\mathcal{T}_m)\}_{j=1}^{J}$, and $p^{\text{adj}}_{\mathbb{X}_j(\mathcal{T}_m)}$ denotes the Holm-adjusted $p$-value of the estimated CACE within the discovered leaf $\mathbb{X}_j(\mathcal{T}_m)$. We set the significance level $\alpha = 0.05$.

\subsubsection{Subgroup detection: DR and FDR}\label{append:dr_fdr}

At the tree level, we assess whether the true heterogeneous subgroups $l_1, l_2$ from \eqref{eq:cace_truth} are recovered as leaves of $\mathcal{T}_m$. The \gls{DR} is the average share of true subgroups recovered, computed without a significance requirement to isolate the structural recovery ability of the algorithm from inferential uncertainty. We report \gls{DR} both overall and separately for $l_1$ and $l_2$. The \gls{FDR} is the share of replications in which at least one spurious leaf (a leaf not corresponding to a truly heterogeneous subgroup) is flagged as significant. \gls{FDR} is evaluated at the tree level rather than normalized by the number of non-effect subgroups, since the latter is a random quantity that varies across replications and therefore cannot serve as a stable denominator.\footnote{Normalizing by the tree would yield a less conservative metric and the tree-level indicator imposes a stricter penalty.} Formally,
\begin{align}
   \text{\gls{DR}} &= \frac{1}{M} \sum_{m=1}^{M} \frac{\sum_{j=1}^{J} \mathbf{1}\{\mathbb{X}_j(\mathcal{T}_m) \in \{l_1, l_2\}\}}{\#\{l_1, l_2\}}, \label{eq:DR}\\
   \text{\gls{FDR}} &= \frac{1}{M} \sum_{m=1}^{M} \mathbf{1}\!\left\{ \sum_{j=1}^{J} \mathbf{1}\!\bigl\{\mathbb{X}_j(\mathcal{T}_m) \notin \{l_1, l_2\},\; p^{\text{adj}}_{\mathbb{X}_j(\mathcal{T}_m)} \leq \alpha\bigr\} \geq 1 \right\}. \label{eq:FDR}
\end{align}
Because $l_1, l_2$ are only observable in simulation, \gls{DR} and \gls{FDR} serve as benchmarking tools. For empirical applications, the unit-level classification and estimation metrics below are the relevant ones.

\subsubsection{Individual classification: Recall, Precision, FPR, $F$-score}\label{append:class}

At the unit level, we ask whether each observation in $\mathcal{I}_{\text{inf}}$ is assigned to a leaf whose significance verdict matches its true heterogeneity status. Let $g_i \in \{l_0, l_1, l_2\}$ denote the true subgroup of unit $i$ under \eqref{eq:cace_truth}, and let $p^{\text{adj}}_{\mathbb{X}_j(\mathcal{T}_m)}$ denote the adjusted $p$-value of the leaf $\mathbb{X}_j(\mathcal{T}_m)$ containing $i$. Each unit contributes to one of four classification cells:
\begin{align*}
   \text{TP}(\mathcal{T}_m) &= \#\bigl\{i : g_i \in \{l_1, l_2\},\; p^{\text{adj}}_{\mathbb{X}_j(\mathcal{T}_m)} \leq \alpha\bigr\}, &
   \text{FN}(\mathcal{T}_m) &= \#\bigl\{i : g_i \in \{l_1, l_2\},\; p^{\text{adj}}_{\mathbb{X}_j(\mathcal{T}_m)} > \alpha\bigr\}, \\
   \text{FP}(\mathcal{T}_m) &= \#\bigl\{i : g_i = l_0,\; p^{\text{adj}}_{\mathbb{X}_j(\mathcal{T}_m)} \leq \alpha\bigr\}, &
   \text{TN}(\mathcal{T}_m) &= \#\bigl\{i : g_i = l_0,\; p^{\text{adj}}_{\mathbb{X}_j(\mathcal{T}_m)} > \alpha\bigr\}.
\end{align*}

From these, we compute the standard classification metrics:
\begin{equation}\label{eq:class_metrics}
\begin{split}
   \text{Recall} &= \frac{\text{TP}}{\text{TP} + \text{FN}}, \\
   \text{Precision} &= \frac{\text{TP}}{\text{TP} + \text{FP}}, \\
   \text{FPR} &= \frac{\text{FP}}{\text{FP} + \text{TN}}, \\
   F\text{-score} &= \frac{\text{TP}}{\text{TP} + \tfrac{1}{2}(\text{FP} + \text{FN})}.
\end{split}
\end{equation}

Recall measures the share of truly heterogeneous units correctly flagged, while Precision is the share of flagged units that are genuinely heterogeneous. We measure the share of null units spuriously flagged with FPR, while the $F$-score is their harmonic mean. All four are averaged across the $M$ simulation runs.

\subsubsection{Estimation precision: bias, MSE, and coverage}\label{append:precision}

For each simulation $m$, we evaluate the per-replication bias, MSE, and
95\% coverage on the truly heterogeneous regions $l_1 \cup l_2$ within
$\mathcal{I}_{\text{inf}}$, denoted
$\mathcal{I}^{\text{het}}_{\text{inf}} = \{ i \in \mathcal{I}_{\text{inf}} : X_i \in l_1 \cup l_2 \}$
with cardinality $N^{\text{het}}_{\text{inf}}$:
\begin{align}
\begin{split}
    \text{|Bias|}_m\bigl(\mathcal{I}^{\text{het}}_{\text{inf}}\bigr)
   &= \frac{1}{N^{\text{het}}_{\text{inf}}}
      \sum_{i \in \mathcal{I}^{\text{het}}_{\text{inf}}}
      | \tau^{\text{CACE}}(X_i) - \widehat\tau^{\text{CACE}}(X_i) |, \\
   \text{Bias}_m\bigl(\mathcal{I}^{\text{het}}_{\text{inf}}\bigr)
   &= \frac{1}{N^{\text{het}}_{\text{inf}}}
      \sum_{i \in \mathcal{I}^{\text{het}}_{\text{inf}}}
      \bigl( \tau^{\text{CACE}}(X_i) - \widehat\tau^{\text{CACE}}(X_i) \bigr), \\
   \text{MSE}_m\bigl(\mathcal{I}^{\text{het}}_{\text{inf}}\bigr)
   &= \frac{1}{N^{\text{het}}_{\text{inf}}}
      \sum_{i \in \mathcal{I}^{\text{het}}_{\text{inf}}}
      \bigl( \tau^{\text{CACE}}(X_i) - \widehat\tau^{\text{CACE}}(X_i) \bigr)^{2}, \\
   \text{Coverage}_m\bigl(\mathcal{I}^{\text{het}}_{\text{inf}}\bigr)
   &= \frac{1}{N^{\text{het}}_{\text{inf}}}
      \sum_{i \in \mathcal{I}^{\text{het}}_{\text{inf}}}
      \mathbf{1}\!\left\{
        \tau^{\text{CACE}}(X_i) \in
        \widehat{\text{CI}}_{95}\bigl(\widehat\tau^{\text{CACE}}(X_i)\bigr)
      \right\},
\end{split}
\end{align}
with Monte Carlo averages
$\text{Bias}\bigl(\mathcal{I}^{\text{het}}_{\text{inf}}\bigr)
 = \tfrac{1}{M} \sum_{m=1}^{M} \text{Bias}_m\bigl(\mathcal{I}^{\text{het}}_{\text{inf}}\bigr)$,
and analogously for $|\text{Bias}|_{\text{m}}$, $\text{MSE}_{m}$ and $\text{Coverage}_m$. $\widehat{\text{CI}}_{95}\bigl(\widehat\tau^{\text{CACE}}(X_i)\bigr)$ is the estimated 95\% confidence interval for unit $i$. Under correct coverage, $\text{Coverage}\bigl(\mathcal{I}^{\text{het}}_{\text{inf}}\bigr) \to 0.95$.

\subsection{Computational details}
\label{append:comp_details}

All simulations were implemented in R (version~4.0.0 or later) \citep{r_lang} and run on CPU
resources only. We note that no GPU acceleration is required. The reference machine was a
small university server with 25 cores and 256\,GB of RAM, though a workstation with at least four cores and 16\,GB of RAM is sufficient to reproduce a single simulation cell. We ran $M=500$ Monte Carlo simulation replications for every combination of effect size and covariate dimension, with effect sizes spanning
$k \in \{0,\,0.2,\,\ldots,\,2.0\}$ and number of covariates $P \in \{10,\,50,\,100\}$, yielding
$500 \times 11 \times 3 = 16{,}500$ runs. Parallelizing across 25 cores, one full setting required approximately six days of wall-clock time. As a single-machine benchmark, an unparallelized run averaged $376$ seconds (with a standard deviation of $37$ seconds) on an Intel i7-1365U laptop. We note that in few sparse covariate settings, \gls{bcf-iv} of \citet{bargagli-stoffi_heterogeneous_2022} occasionally returned
$\widehat{\pi}_C^{\,\text{BCF}}(x) = 0$, which renders the second-stage ratio
$\widehat{\text{ITT}}_Y^{\,\text{BCF}}(x)\,/\,\widehat{\pi}_C^{\,\text{BCF}}(x)$ for \gls{bcf-iv} undefined when forming $\widehat{\tau}^{\,\text{BCF}}(x)$. We replaced exact
zeros with a small positive constant to let the algorithm proceed. The
corresponding \gls{sbcf-iv} estimates were numerically well behaved across all replications and required no such adjustment.

\newpage

\subsection{Results for main simulations}
\label{append:full_precision_results}

\begin{table}[H]
\centering
\footnotesize
\caption{Tree-level subgroup detection.}
\label{tab:rule}
\resizebox{\ifdim\width>\linewidth\linewidth\else\width\fi}{!}{
\begin{tabular}[t]{ll>{\raggedleft\arraybackslash}p{0.25cm}rrrr>{\raggedleft\arraybackslash}p{0.25cm}rrrr}
\toprule
\multicolumn{1}{c}{} & \multicolumn{1}{c}{} & \multicolumn{1}{c}{} & \multicolumn{4}{c}{BCF-IV} & \multicolumn{1}{c}{} & \multicolumn{4}{c}{SBCF-IV} \\
\cmidrule(l{3pt}r{3pt}){4-7} \cmidrule(l{3pt}r{3pt}){9-12}
\multicolumn{1}{c}{} & \multicolumn{1}{c}{} & \multicolumn{1}{c}{} & \multicolumn{3}{c}{DR} & \multicolumn{1}{c}{} & \multicolumn{1}{c}{} & \multicolumn{3}{c}{DR} & \multicolumn{1}{c}{} \\
\cmidrule(l{3pt}r{3pt}){4-6} \cmidrule(l{3pt}r{3pt}){9-11}
$P$ & $k$ &  & Overall & $l_1$ & $l_2$ & $\text{FDR}$ &  & Overall & $l_1$ & $l_2$ & $\text{FDR}$\\
\midrule
 & 0 &  & 0.003 & 0.004 & 0.002 & 0.036 &  & 0.013 & 0.012 & 0.014 & 0.030\\

 & 0.2 &  & 0.006 & 0.006 & 0.006 & 0.036 &  & 0.051 & 0.052 & 0.050 & 0.034\\

 & 0.4 &  & 0.018 & 0.018 & 0.018 & 0.058 &  & 0.211 & 0.210 & 0.212 & 0.062\\

 & 0.6 &  & 0.061 & 0.070 & 0.052 & 0.094 &  & 0.391 & 0.396 & 0.386 & 0.058\\

 & 0.8 &  & 0.103 & 0.098 & 0.108 & 0.122 &  & 0.656 & 0.648 & 0.664 & 0.062\\

 & 1 &  & 0.178 & 0.178 & 0.178 & 0.236 &  & 0.873 & 0.876 & 0.870 & 0.040\\

 & 1.2 &  & 0.222 & 0.220 & 0.224 & 0.336 &  & 0.950 & 0.948 & 0.952 & 0.046\\

 & 1.4 &  & 0.263 & 0.286 & 0.240 & 0.468 &  & 0.984 & 0.986 & 0.982 & 0.028\\

 & 1.6 &  & 0.289 & 0.282 & 0.296 & 0.552 &  & 0.997 & 1.000 & 0.994 & 0.048\\

 & 1.8 &  & 0.322 & 0.336 & 0.308 & 0.624 &  & 1.000 & 1.000 & 1.000 & 0.030\\

\multirow{-11}{*}{\raggedright\arraybackslash 10} & 2 &  & 0.328 & 0.328 & 0.328 & 0.728 &  & 1.000 & 1.000 & 1.000 & 0.060\\
\cmidrule{1-12}
 & 0 &  & 0.000 & 0.000 & 0.000 & 0.046 &  & 0.004 & 0.004 & 0.004 & 0.038\\

 & 0.2 &  & 0.000 & 0.000 & 0.000 & 0.032 &  & 0.019 & 0.020 & 0.018 & 0.040\\

 & 0.4 &  & 0.000 & 0.000 & 0.000 & 0.024 &  & 0.148 & 0.146 & 0.150 & 0.042\\

 & 0.6 &  & 0.001 & 0.000 & 0.002 & 0.048 &  & 0.485 & 0.490 & 0.480 & 0.040\\

 & 0.8 &  & 0.013 & 0.012 & 0.014 & 0.124 &  & 0.838 & 0.836 & 0.840 & 0.054\\

 & 1 &  & 0.021 & 0.014 & 0.028 & 0.186 &  & 0.963 & 0.962 & 0.964 & 0.028\\

 & 1.2 &  & 0.027 & 0.022 & 0.032 & 0.310 &  & 0.984 & 0.986 & 0.982 & 0.046\\

 & 1.4 &  & 0.034 & 0.034 & 0.034 & 0.400 &  & 0.997 & 0.996 & 0.998 & 0.036\\

 & 1.6 &  & 0.037 & 0.026 & 0.048 & 0.516 &  & 0.999 & 0.998 & 1.000 & 0.050\\

 & 1.8 &  & 0.042 & 0.042 & 0.042 & 0.562 &  & 0.999 & 1.000 & 0.998 & 0.052\\

\multirow{-11}{*}{\raggedright\arraybackslash 50} & 2 &  & 0.048 & 0.050 & 0.046 & 0.634 &  & 1.000 & 1.000 & 1.000 & 0.048\\
\cmidrule{1-12}
 & 0 &  & 0.000 & 0.000 & 0.000 & 0.030 &  & 0.000 & 0.000 & 0.000 & 0.040\\

 & 0.2 &  & 0.000 & 0.000 & 0.000 & 0.028 &  & 0.004 & 0.002 & 0.006 & 0.028\\

 & 0.4 &  & 0.000 & 0.000 & 0.000 & 0.036 &  & 0.067 & 0.058 & 0.076 & 0.054\\

 & 0.6 &  & 0.001 & 0.000 & 0.002 & 0.058 &  & 0.391 & 0.390 & 0.392 & 0.086\\

 & 0.8 &  & 0.000 & 0.000 & 0.000 & 0.086 &  & 0.817 & 0.814 & 0.820 & 0.048\\

 & 1 &  & 0.001 & 0.000 & 0.002 & 0.118 &  & 0.954 & 0.952 & 0.956 & 0.046\\

 & 1.2 &  & 0.006 & 0.004 & 0.008 & 0.208 &  & 0.990 & 0.990 & 0.990 & 0.040\\

 & 1.4 &  & 0.009 & 0.006 & 0.012 & 0.256 &  & 0.995 & 0.996 & 0.994 & 0.030\\

 & 1.6 &  & 0.013 & 0.016 & 0.010 & 0.376 &  & 1.000 & 1.000 & 1.000 & 0.030\\

 & 1.8 &  & 0.014 & 0.008 & 0.020 & 0.482 &  & 0.998 & 0.998 & 0.998 & 0.044\\

\multirow{-11}{*}{\raggedright\arraybackslash 100} & 2 &  & 0.015 & 0.020 & 0.010 & 0.516 &  & 0.999 & 0.998 & 1.000 & 0.052\\
\bottomrule
\end{tabular}}
\par
\begin{notes}
    \gls{DR} and \gls{FDR} for \gls{bcf-iv} and \gls{sbcf-iv} across covariate dimensions $P \in \{10, 50, 100\}$ and effect sizes $k \in \{0, 0.2, \ldots, 2\}$. \gls{DR} is reported overall and separately for the two true heterogeneity subgroups $l_1$ and $l_2$. Results are averaged over $M = 500$ Monte Carlo replications with $N = 1{,}000$. This table reports the numerical values underlying Figure~\ref{fig:dr_fdr}.
\end{notes}
\end{table}


\begin{figure}
   \centering
   \caption{Unit-level classification performance.}
   \label{fig:full_classification}
   \resizebox{0.9\linewidth}{!}{%
   \input{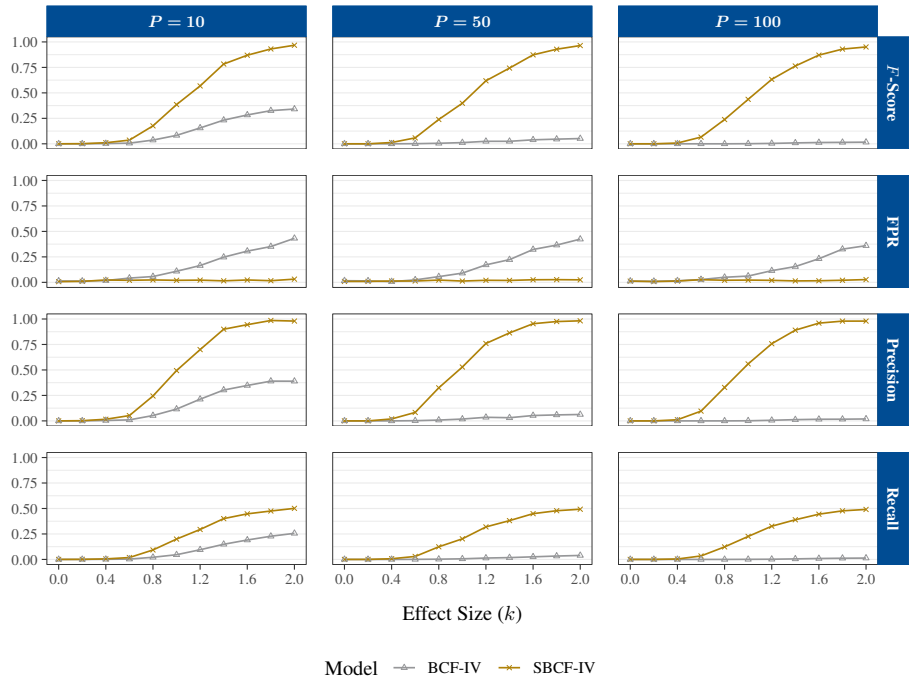}
   }
   \begin{notes}
   Recall, Precision, False Positive Rate (FPR), and $F$-score as a function of effect size $k$, across covariate dimensions $P \in \{10, 50, 100\}$. All four metrics are defined in \eqref{eq:class_metrics} and computed over all units in $\mathcal{I}_{\text{inf}}$, with significance evaluated at $\alpha = 0.05$ using Holm-adjusted $p$-values. Results are averaged over $M = 500$ Monte Carlo replications with $N = 1{,}000$, and compare \gls{sbcf-iv} (orange) with \gls{bcf-iv} (gray). This figure complements Figure~\ref{fig:precision.Fscore}, which reports Precision and $F$-score only, by additionally reporting Recall and FPR across the same $(k, P)$ grid. We report full results in Table \ref{tab:full_classification}.
   \end{notes}
\end{figure}

\clearpage
\newpage

\begin{table}[H]
\centering
\caption{Unit-level classification metrics.}
\label{tab:full_classification}
\centering
\resizebox{\ifdim\width>\linewidth\linewidth\else\width\fi}{!}{
\begin{tabular}[t]{ll>{\raggedleft\arraybackslash}p{0.25cm}rrrr>{\raggedleft\arraybackslash}p{0.25cm}rrrr}
\toprule
\multicolumn{1}{c}{} & \multicolumn{1}{c}{} & \multicolumn{1}{c}{} & \multicolumn{4}{c}{BCF-IV} & \multicolumn{1}{c}{} & \multicolumn{4}{c}{SBCF-IV} \\
\cmidrule(l{3pt}r{3pt}){4-7} \cmidrule(l{3pt}r{3pt}){9-12}
$P$ & $k$ &  & \multicolumn{1}{c}{Recall} & \multicolumn{1}{c}{Precision} & \multicolumn{1}{c}{$F$-Score} & \multicolumn{1}{c}{FPR} &  & \multicolumn{1}{c}{Recall} & \multicolumn{1}{c}{Precision} & \multicolumn{1}{c}{$F$-Score} & \multicolumn{1}{c}{FPR}\\
\midrule
 & 0 &  & \makecell{0.000 (0.000)} & \makecell{0.000 (0.000)} & \makecell{0.000 (0.000)} & \makecell{0.013 (0.087)} &  & \makecell{0.000 (0.000)} & \makecell{0.000 (0.000)} & \makecell{0.000 (0.000)} & \makecell{0.007 (0.046)}\\

 & 0.2 &  & \makecell{0.000 (0.000)} & \makecell{0.000 (0.000)} & \makecell{0.000 (0.000)} & \makecell{0.012 (0.081)} &  & \makecell{0.001 (0.015)} & \makecell{0.002 (0.045)} & \makecell{0.001 (0.030)} & \makecell{0.010 (0.061)}\\

 & 0.4 &  & \makecell{0.001 (0.021)} & \makecell{0.004 (0.063)} & \makecell{0.003 (0.041)} & \makecell{0.018 (0.095)} &  & \makecell{0.005 (0.042)} & \makecell{0.016 (0.126)} & \makecell{0.011 (0.084)} & \makecell{0.023 (0.108)}\\

 & 0.6 &  & \makecell{0.003 (0.033)} & \makecell{0.010 (0.100)} & \makecell{0.006 (0.065)} & \makecell{0.042 (0.144)} &  & \makecell{0.018 (0.080)} & \makecell{0.052 (0.222)} & \makecell{0.036 (0.156)} & \makecell{0.020 (0.088)}\\

 & 0.8 &  & \makecell{0.021 (0.087)} & \makecell{0.052 (0.216)} & \makecell{0.037 (0.149)} & \makecell{0.057 (0.166)} &  & \makecell{0.093 (0.166)} & \makecell{0.244 (0.425)} & \makecell{0.176 (0.310)} & \makecell{0.024 (0.099)}\\

 & 1 &  & \makecell{0.046 (0.124)} & \makecell{0.116 (0.312)} & \makecell{0.083 (0.221)} & \makecell{0.109 (0.214)} &  & \makecell{0.200 (0.209)} & \makecell{0.494 (0.495)} & \makecell{0.385 (0.398)} & \makecell{0.020 (0.102)}\\

 & 1.2 &  & \makecell{0.095 (0.175)} & \makecell{0.213 (0.390)} & \makecell{0.156 (0.280)} & \makecell{0.165 (0.262)} &  & \makecell{0.294 (0.206)} & \makecell{0.700 (0.451)} & \makecell{0.568 (0.387)} & \makecell{0.022 (0.104)}\\

 & 1.4 &  & \makecell{0.149 (0.204)} & \makecell{0.303 (0.424)} & \makecell{0.232 (0.317)} & \makecell{0.247 (0.300)} &  & \makecell{0.401 (0.156)} & \makecell{0.900 (0.290)} & \makecell{0.783 (0.293)} & \makecell{0.015 (0.091)}\\

 & 1.6 &  & \makecell{0.191 (0.222)} & \makecell{0.348 (0.423)} & \makecell{0.283 (0.336)} & \makecell{0.306 (0.322)} &  & \makecell{0.447 (0.122)} & \makecell{0.944 (0.204)} & \makecell{0.870 (0.227)} & \makecell{0.023 (0.105)}\\

 & 1.8 &  & \makecell{0.228 (0.231)} & \makecell{0.390 (0.418)} & \makecell{0.326 (0.338)} & \makecell{0.351 (0.327)} &  & \makecell{0.475 (0.084)} & \makecell{0.985 (0.087)} & \makecell{0.931 (0.142)} & \makecell{0.016 (0.095)}\\

\multirow{-11}{*}{\raggedright\arraybackslash 10} & 2 &  & \makecell{0.256 (0.245)} & \makecell{0.390 (0.398)} & \makecell{0.341 (0.339)} & \makecell{0.432 (0.331)} &  & \makecell{0.501 (0.068)} & \makecell{0.980 (0.081)} & \makecell{0.967 (0.090)} & \makecell{0.031 (0.126)}\\
\cmidrule{1-12}
 & 0 &  & \makecell{0.000 (0.000)} & \makecell{0.000 (0.000)} & \makecell{0.000 (0.000)} & \makecell{0.016 (0.107)} &  & \makecell{0.000 (0.000)} & \makecell{0.000 (0.000)} & \makecell{0.000 (0.000)} & \makecell{0.009 (0.049)}\\

 & 0.2 &  & \makecell{0.000 (0.000)} & \makecell{0.000 (0.000)} & \makecell{0.000 (0.000)} & \makecell{0.015 (0.112)} &  & \makecell{0.000 (0.000)} & \makecell{0.000 (0.000)} & \makecell{0.000 (0.000)} & \makecell{0.009 (0.047)}\\

 & 0.4 &  & \makecell{0.000 (0.000)} & \makecell{0.000 (0.000)} & \makecell{0.000 (0.000)} & \makecell{0.007 (0.069)} &  & \makecell{0.007 (0.049)} & \makecell{0.019 (0.135)} & \makecell{0.013 (0.090)} & \makecell{0.013 (0.066)}\\

 & 0.6 &  & \makecell{0.001 (0.014)} & \makecell{0.002 (0.045)} & \makecell{0.001 (0.029)} & \makecell{0.024 (0.127)} &  & \makecell{0.029 (0.098)} & \makecell{0.083 (0.275)} & \makecell{0.058 (0.195)} & \makecell{0.013 (0.070)}\\

 & 0.8 &  & \makecell{0.003 (0.029)} & \makecell{0.008 (0.089)} & \makecell{0.005 (0.058)} & \makecell{0.056 (0.174)} &  & \makecell{0.124 (0.181)} & \makecell{0.325 (0.464)} & \makecell{0.239 (0.348)} & \makecell{0.022 (0.099)}\\

 & 1 &  & \makecell{0.006 (0.044)} & \makecell{0.018 (0.133)} & \makecell{0.012 (0.088)} & \makecell{0.091 (0.212)} &  & \makecell{0.202 (0.198)} & \makecell{0.529 (0.498)} & \makecell{0.398 (0.389)} & \makecell{0.013 (0.077)}\\

 & 1.2 &  & \makecell{0.014 (0.069)} & \makecell{0.035 (0.181)} & \makecell{0.025 (0.125)} & \makecell{0.172 (0.284)} &  & \makecell{0.319 (0.191)} & \makecell{0.760 (0.418)} & \makecell{0.619 (0.366)} & \makecell{0.020 (0.095)}\\

 & 1.4 &  & \makecell{0.018 (0.086)} & \makecell{0.031 (0.157)} & \makecell{0.025 (0.119)} & \makecell{0.223 (0.314)} &  & \makecell{0.381 (0.170)} & \makecell{0.864 (0.333)} & \makecell{0.743 (0.320)} & \makecell{0.019 (0.100)}\\

 & 1.6 &  & \makecell{0.025 (0.099)} & \makecell{0.053 (0.210)} & \makecell{0.039 (0.151)} & \makecell{0.323 (0.359)} &  & \makecell{0.449 (0.115)} & \makecell{0.954 (0.181)} & \makecell{0.873 (0.209)} & \makecell{0.025 (0.111)}\\

 & 1.8 &  & \makecell{0.033 (0.119)} & \makecell{0.059 (0.213)} & \makecell{0.046 (0.160)} & \makecell{0.366 (0.372)} &  & \makecell{0.478 (0.086)} & \makecell{0.975 (0.112)} & \makecell{0.928 (0.149)} & \makecell{0.027 (0.114)}\\

\multirow{-11}{*}{\raggedright\arraybackslash 50} & 2 &  & \makecell{0.039 (0.126)} & \makecell{0.063 (0.211)} & \makecell{0.052 (0.168)} & \makecell{0.424 (0.384)} &  & \makecell{0.493 (0.070)} & \makecell{0.982 (0.084)} & \makecell{0.964 (0.103)} & \makecell{0.025 (0.112)}\\
\cmidrule{1-12}
 & 0 &  & \makecell{0.000 (0.000)} & \makecell{0.000 (0.000)} & \makecell{0.000 (0.000)} & \makecell{0.012 (0.103)} &  & \makecell{0.000 (0.000)} & \makecell{0.000 (0.000)} & \makecell{0.000 (0.000)} & \makecell{0.012 (0.060)}\\

 & 0.2 &  & \makecell{0.000 (0.000)} & \makecell{0.000 (0.000)} & \makecell{0.000 (0.000)} & \makecell{0.010 (0.094)} &  & \makecell{0.000 (0.000)} & \makecell{0.000 (0.000)} & \makecell{0.000 (0.000)} & \makecell{0.007 (0.047)}\\

 & 0.4 &  & \makecell{0.000 (0.000)} & \makecell{0.000 (0.000)} & \makecell{0.000 (0.000)} & \makecell{0.016 (0.109)} &  & \makecell{0.004 (0.038)} & \makecell{0.011 (0.104)} & \makecell{0.008 (0.071)} & \makecell{0.013 (0.057)}\\

 & 0.6 &  & \makecell{0.000 (0.000)} & \makecell{0.000 (0.000)} & \makecell{0.000 (0.000)} & \makecell{0.027 (0.137)} &  & \makecell{0.034 (0.104)} & \makecell{0.096 (0.293)} & \makecell{0.066 (0.201)} & \makecell{0.026 (0.092)}\\

 & 0.8 &  & \makecell{0.000 (0.000)} & \makecell{0.000 (0.000)} & \makecell{0.000 (0.000)} & \makecell{0.050 (0.185)} &  & \makecell{0.123 (0.178)} & \makecell{0.328 (0.466)} & \makecell{0.238 (0.343)} & \makecell{0.021 (0.099)}\\

 & 1 &  & \makecell{0.001 (0.021)} & \makecell{0.001 (0.025)} & \makecell{0.001 (0.023)} & \makecell{0.062 (0.195)} &  & \makecell{0.226 (0.209)} & \makecell{0.559 (0.491)} & \makecell{0.436 (0.398)} & \makecell{0.022 (0.107)}\\

 & 1.2 &  & \makecell{0.002 (0.026)} & \makecell{0.006 (0.077)} & \makecell{0.004 (0.052)} & \makecell{0.115 (0.252)} &  & \makecell{0.326 (0.197)} & \makecell{0.757 (0.420)} & \makecell{0.632 (0.376)} & \makecell{0.020 (0.098)}\\

 & 1.4 &  & \makecell{0.005 (0.046)} & \makecell{0.012 (0.105)} & \makecell{0.009 (0.077)} & \makecell{0.155 (0.300)} &  & \makecell{0.389 (0.157)} & \makecell{0.890 (0.302)} & \makecell{0.764 (0.299)} & \makecell{0.014 (0.081)}\\

 & 1.6 &  & \makecell{0.009 (0.063)} & \makecell{0.016 (0.113)} & \makecell{0.013 (0.087)} & \makecell{0.232 (0.337)} &  & \makecell{0.444 (0.119)} & \makecell{0.960 (0.179)} & \makecell{0.871 (0.214)} & \makecell{0.016 (0.092)}\\

 & 1.8 &  & \makecell{0.012 (0.074)} & \makecell{0.017 (0.107)} & \makecell{0.014 (0.089)} & \makecell{0.326 (0.384)} &  & \makecell{0.477 (0.084)} & \makecell{0.979 (0.105)} & \makecell{0.930 (0.149)} & \makecell{0.021 (0.098)}\\

\multirow{-11}{*}{\raggedright\arraybackslash 100} & 2 &  & \makecell{0.013 (0.082)} & \makecell{0.019 (0.119)} & \makecell{0.016 (0.099)} & \makecell{0.360 (0.395)} &  & \makecell{0.491 (0.077)} & \makecell{0.979 (0.092)} & \makecell{0.951 (0.121)} & \makecell{0.027 (0.119)}\\
\bottomrule
\end{tabular}
}%
\par
\begin{notes}
Recall, Precision, $F$-score, and FPR for \gls{bcf-iv} and \gls{sbcf-iv} across covariate dimensions $P \in \{10, 50, 100\}$ and effect sizes $k \in \{0, 0.2, \ldots, 2\}$. All four metrics are defined in \eqref{eq:class_metrics} and computed over units in $\mathcal{I}_{\text{inf}}$, with significance evaluated at $\alpha = 0.05$ using Holm-adjusted $p$-values. Results are averaged over $M = 500$ Monte Carlo replications with $N = 1{,}000$; standard deviations are reported in parentheses. This table reports the numerical values underlying Figure~\ref{fig:precision.Fscore} and Figure~\ref{fig:full_classification}.
\end{notes}
\end{table}


\begin{table}[H]
\centering
\caption{Estimation precision and uncertainty quantification for $\widehat{\tau}^{\text{CACE}}(x)$.}
\label{tab:full_precision_results}
\centering
\resizebox{\ifdim\width>\linewidth\linewidth\else\width\fi}{!}{
\begin{tabular}[t]{ll>{\raggedleft\arraybackslash}p{0.25cm}rrrrr>{\raggedleft\arraybackslash}p{0.25cm}rrrrr}
\toprule
\multicolumn{1}{c}{} & \multicolumn{1}{c}{} & \multicolumn{1}{c}{} & \multicolumn{5}{c}{BCF-IV} & \multicolumn{1}{c}{} & \multicolumn{5}{c}{SBCF-IV} \\
\cmidrule(l{3pt}r{3pt}){4-8} \cmidrule(l{3pt}r{3pt}){10-14}
$P$ & $k$ &  & \multicolumn{1}{c}{MSE} & \multicolumn{1}{c}{Bias} & \multicolumn{1}{c}{Abs. Bias} & \multicolumn{1}{c}{Coverage} & \multicolumn{1}{c}{Width} &  & \multicolumn{1}{c}{MSE} & \multicolumn{1}{c}{Bias} & \multicolumn{1}{c}{Abs. Bias} & \multicolumn{1}{c}{Coverage} & \multicolumn{1}{c}{Width}\\
\midrule
 & 0 &  & \makecell{0.188 (0.163)} & \makecell{-0.004 (0.225)} & \makecell{0.286 (0.140)} & \makecell{0.948 (0.176)} & \makecell{1.434 (0.219)} &  & \makecell{0.191 (0.161)} & \makecell{-0.004 (0.222)} & \makecell{0.314 (0.131)} & \makecell{0.957 (0.131)} & \makecell{1.600 (0.176)}\\

 & 0.2 &  & \makecell{0.214 (0.148)} & \makecell{0.011 (0.233)} & \makecell{0.336 (0.114)} & \makecell{0.891 (0.178)} & \makecell{1.448 (0.200)} &  & \makecell{0.221 (0.160)} & \makecell{0.005 (0.229)} & \makecell{0.357 (0.129)} & \makecell{0.912 (0.143)} & \makecell{1.612 (0.183)}\\

 & 0.4 &  & \makecell{0.294 (0.178)} & \makecell{0.014 (0.232)} & \makecell{0.416 (0.122)} & \makecell{0.797 (0.206)} & \makecell{1.474 (0.208)} &  & \makecell{0.306 (0.193)} & \makecell{0.018 (0.234)} & \makecell{0.438 (0.133)} & \makecell{0.847 (0.176)} & \makecell{1.652 (0.185)}\\

 & 0.6 &  & \makecell{0.375 (0.222)} & \makecell{-0.012 (0.237)} & \makecell{0.494 (0.172)} & \makecell{0.708 (0.297)} & \makecell{1.537 (0.205)} &  & \makecell{0.376 (0.256)} & \makecell{-0.002 (0.271)} & \makecell{0.493 (0.198)} & \makecell{0.784 (0.234)} & \makecell{1.706 (0.193)}\\

 & 0.8 &  & \makecell{0.452 (0.326)} & \makecell{-0.004 (0.254)} & \makecell{0.547 (0.230)} & \makecell{0.697 (0.341)} & \makecell{1.591 (0.258)} &  & \makecell{0.370 (0.328)} & \makecell{-0.013 (0.288)} & \makecell{0.485 (0.262)} & \makecell{0.809 (0.283)} & \makecell{1.759 (0.206)}\\

 & 1 &  & \makecell{0.473 (0.333)} & \makecell{0.000 (0.286)} & \makecell{0.569 (0.240)} & \makecell{0.695 (0.337)} & \makecell{1.600 (0.210)} &  & \makecell{0.298 (0.327)} & \makecell{-0.002 (0.308)} & \makecell{0.426 (0.256)} & \makecell{0.899 (0.235)} & \makecell{1.794 (0.160)}\\

 & 1.2 &  & \makecell{0.581 (0.425)} & \makecell{-0.003 (0.325)} & \makecell{0.626 (0.275)} & \makecell{0.631 (0.333)} & \makecell{1.653 (0.205)} &  & \makecell{0.273 (0.317)} & \makecell{-0.008 (0.337)} & \makecell{0.404 (0.245)} & \makecell{0.911 (0.216)} & \makecell{1.823 (0.151)}\\

 & 1.4 &  & \makecell{0.607 (0.483)} & \makecell{0.021 (0.358)} & \makecell{0.641 (0.271)} & \makecell{0.643 (0.318)} & \makecell{1.667 (0.186)} &  & \makecell{0.232 (0.248)} & \makecell{0.008 (0.331)} & \makecell{0.383 (0.208)} & \makecell{0.947 (0.161)} & \makecell{1.824 (0.147)}\\

 & 1.6 &  & \makecell{0.718 (0.574)} & \makecell{-0.006 (0.380)} & \makecell{0.711 (0.314)} & \makecell{0.553 (0.334)} & \makecell{1.660 (0.198)} &  & \makecell{0.221 (0.230)} & \makecell{0.009 (0.328)} & \makecell{0.373 (0.200)} & \makecell{0.954 (0.148)} & \makecell{1.831 (0.148)}\\

 & 1.8 &  & \makecell{0.814 (0.584)} & \makecell{-0.016 (0.406)} & \makecell{0.760 (0.313)} & \makecell{0.536 (0.319)} & \makecell{1.681 (0.174)} &  & \makecell{0.223 (0.225)} & \makecell{-0.029 (0.341)} & \makecell{0.376 (0.203)} & \makecell{0.945 (0.161)} & \makecell{1.829 (0.142)}\\

\multirow{-11}{*}{\raggedright\arraybackslash 10} & 2 &  & \makecell{0.931 (0.729)} & \makecell{0.014 (0.446)} & \makecell{0.811 (0.354)} & \makecell{0.478 (0.335)} & \makecell{1.697 (0.174)} &  & \makecell{0.214 (0.211)} & \makecell{0.012 (0.337)} & \makecell{0.368 (0.199)} & \makecell{0.953 (0.149)} & \makecell{1.840 (0.141)}\\
\cmidrule{1-14}
 & 0 &  & \makecell{0.205 (0.244)} & \makecell{-0.003 (0.223)} & \makecell{0.273 (0.145)} & \makecell{0.953 (0.178)} & \makecell{1.399 (0.244)} &  & \makecell{0.214 (0.151)} & \makecell{-0.003 (0.222)} & \makecell{0.358 (0.140)} & \makecell{0.960 (0.111)} & \makecell{1.769 (0.143)}\\

 & 0.2 &  & \makecell{0.229 (0.180)} & \makecell{0.012 (0.226)} & \makecell{0.321 (0.110)} & \makecell{0.894 (0.192)} & \makecell{1.387 (0.206)} &  & \makecell{0.240 (0.155)} & \makecell{0.014 (0.228)} & \makecell{0.383 (0.129)} & \makecell{0.938 (0.107)} & \makecell{1.780 (0.150)}\\

 & 0.4 &  & \makecell{0.332 (0.164)} & \makecell{0.013 (0.235)} & \makecell{0.449 (0.079)} & \makecell{0.712 (0.210)} & \makecell{1.381 (0.204)} &  & \makecell{0.297 (0.199)} & \makecell{0.003 (0.261)} & \makecell{0.431 (0.144)} & \makecell{0.895 (0.144)} & \makecell{1.786 (0.142)}\\

 & 0.6 &  & \makecell{0.471 (0.220)} & \makecell{0.015 (0.232)} & \makecell{0.574 (0.140)} & \makecell{0.515 (0.314)} & \makecell{1.422 (0.258)} &  & \makecell{0.271 (0.232)} & \makecell{0.020 (0.279)} & \makecell{0.413 (0.193)} & \makecell{0.912 (0.173)} & \makecell{1.807 (0.150)}\\

 & 0.8 &  & \makecell{0.640 (0.338)} & \makecell{-0.004 (0.244)} & \makecell{0.671 (0.222)} & \makecell{0.457 (0.384)} & \makecell{1.445 (0.233)} &  & \makecell{0.269 (0.284)} & \makecell{-0.010 (0.324)} & \makecell{0.410 (0.217)} & \makecell{0.919 (0.193)} & \makecell{1.825 (0.147)}\\

 & 1 &  & \makecell{0.785 (0.472)} & \makecell{0.012 (0.241)} & \makecell{0.744 (0.288)} & \makecell{0.454 (0.404)} & \makecell{1.495 (0.232)} &  & \makecell{0.213 (0.232)} & \makecell{-0.003 (0.309)} & \makecell{0.362 (0.203)} & \makecell{0.949 (0.151)} & \makecell{1.830 (0.140)}\\

 & 1.2 &  & \makecell{0.885 (0.574)} & \makecell{-0.004 (0.257)} & \makecell{0.802 (0.333)} & \makecell{0.470 (0.389)} & \makecell{1.531 (0.228)} &  & \makecell{0.229 (0.253)} & \makecell{-0.030 (0.331)} & \makecell{0.375 (0.206)} & \makecell{0.942 (0.167)} & \makecell{1.819 (0.150)}\\

 & 1.4 &  & \makecell{1.130 (0.750)} & \makecell{0.003 (0.253)} & \makecell{0.923 (0.375)} & \makecell{0.373 (0.349)} & \makecell{1.546 (0.274)} &  & \makecell{0.223 (0.216)} & \makecell{0.001 (0.338)} & \makecell{0.380 (0.204)} & \makecell{0.953 (0.146)} & \makecell{1.831 (0.142)}\\

 & 1.6 &  & \makecell{1.401 (0.968)} & \makecell{-0.011 (0.279)} & \makecell{1.031 (0.427)} & \makecell{0.341 (0.335)} & \makecell{1.566 (0.245)} &  & \makecell{0.228 (0.239)} & \makecell{-0.021 (0.343)} & \makecell{0.377 (0.213)} & \makecell{0.944 (0.161)} & \makecell{1.845 (0.144)}\\

 & 1.8 &  & \makecell{1.695 (1.171)} & \makecell{0.001 (0.296)} & \makecell{1.157 (0.457)} & \makecell{0.260 (0.307)} & \makecell{1.587 (0.225)} &  & \makecell{0.229 (0.242)} & \makecell{-0.011 (0.348)} & \makecell{0.380 (0.206)} & \makecell{0.945 (0.161)} & \makecell{1.839 (0.143)}\\

\multirow{-11}{*}{\raggedright\arraybackslash 50} & 2 &  & \makecell{2.088 (1.446)} & \makecell{-0.008 (0.299)} & \makecell{1.291 (0.508)} & \makecell{0.202 (0.256)} & \makecell{1.593 (0.255)} &  & \makecell{0.204 (0.203)} & \makecell{0.002 (0.320)} & \makecell{0.359 (0.192)} & \makecell{0.957 (0.145)} & \makecell{1.836 (0.138)}\\
\cmidrule{1-14}
 & 0 &  & \makecell{0.196 (0.234)} & \makecell{-0.020 (0.226)} & \makecell{0.253 (0.143)} & \makecell{0.957 (0.180)} & \makecell{1.325 (0.215)} &  & \makecell{0.217 (0.162)} & \makecell{-0.021 (0.226)} & \makecell{0.360 (0.143)} & \makecell{0.951 (0.111)} & \makecell{1.785 (0.136)}\\

 & 0.2 &  & \makecell{0.244 (0.277)} & \makecell{-0.003 (0.226)} & \makecell{0.314 (0.100)} & \makecell{0.882 (0.198)} & \makecell{1.337 (0.216)} &  & \makecell{0.247 (0.151)} & \makecell{-0.002 (0.226)} & \makecell{0.391 (0.123)} & \makecell{0.930 (0.109)} & \makecell{1.780 (0.136)}\\

 & 0.4 &  & \makecell{0.357 (0.232)} & \makecell{-0.002 (0.237)} & \makecell{0.455 (0.070)} & \makecell{0.667 (0.200)} & \makecell{1.346 (0.211)} &  & \makecell{0.320 (0.252)} & \makecell{-0.008 (0.257)} & \makecell{0.450 (0.134)} & \makecell{0.887 (0.143)} & \makecell{1.807 (0.134)}\\

 & 0.6 &  & \makecell{0.535 (0.304)} & \makecell{-0.006 (0.246)} & \makecell{0.608 (0.114)} & \makecell{0.445 (0.279)} & \makecell{1.351 (0.211)} &  & \makecell{0.304 (0.220)} & \makecell{0.002 (0.291)} & \makecell{0.441 (0.188)} & \makecell{0.889 (0.180)} & \makecell{1.797 (0.148)}\\

 & 0.8 &  & \makecell{0.736 (0.430)} & \makecell{0.006 (0.236)} & \makecell{0.750 (0.175)} & \makecell{0.299 (0.338)} & \makecell{1.379 (0.244)} &  & \makecell{0.252 (0.260)} & \makecell{-0.001 (0.313)} & \makecell{0.392 (0.210)} & \makecell{0.919 (0.184)} & \makecell{1.834 (0.156)}\\

 & 1 &  & \makecell{0.954 (0.418)} & \makecell{0.019 (0.220)} & \makecell{0.872 (0.251)} & \makecell{0.273 (0.351)} & \makecell{1.417 (0.243)} &  & \makecell{0.230 (0.215)} & \makecell{0.005 (0.319)} & \makecell{0.387 (0.199)} & \makecell{0.950 (0.149)} & \makecell{1.826 (0.139)}\\

 & 1.2 &  & \makecell{1.206 (0.633)} & \makecell{-0.013 (0.245)} & \makecell{0.975 (0.326)} & \makecell{0.284 (0.343)} & \makecell{1.460 (0.339)} &  & \makecell{0.197 (0.196)} & \makecell{-0.007 (0.310)} & \makecell{0.354 (0.186)} & \makecell{0.954 (0.147)} & \makecell{1.820 (0.143)}\\

 & 1.4 &  & \makecell{1.544 (0.806)} & \makecell{0.021 (0.249)} & \makecell{1.122 (0.369)} & \makecell{0.237 (0.307)} & \makecell{1.461 (0.233)} &  & \makecell{0.229 (0.242)} & \makecell{0.030 (0.336)} & \makecell{0.377 (0.208)} & \makecell{0.945 (0.161)} & \makecell{1.829 (0.142)}\\

 & 1.6 &  & \makecell{1.803 (0.998)} & \makecell{-0.013 (0.253)} & \makecell{1.210 (0.424)} & \makecell{0.221 (0.277)} & \makecell{1.504 (0.269)} &  & \makecell{0.205 (0.214)} & \makecell{-0.021 (0.325)} & \makecell{0.360 (0.195)} & \makecell{0.960 (0.136)} & \makecell{1.832 (0.134)}\\

 & 1.8 &  & \makecell{2.095 (1.283)} & \makecell{-0.011 (0.262)} & \makecell{1.301 (0.489)} & \makecell{0.211 (0.275)} & \makecell{1.516 (0.255)} &  & \makecell{0.227 (0.245)} & \makecell{-0.001 (0.349)} & \makecell{0.373 (0.210)} & \makecell{0.942 (0.166)} & \makecell{1.825 (0.139)}\\

\multirow{-11}{*}{\raggedright\arraybackslash 100} & 2 &  & \makecell{2.551 (1.522)} & \makecell{0.013 (0.257)} & \makecell{1.454 (0.526)} & \makecell{0.157 (0.222)} & \makecell{1.535 (0.235)} &  & \makecell{0.230 (0.238)} & \makecell{0.022 (0.340)} & \makecell{0.376 (0.212)} & \makecell{0.942 (0.164)} & \makecell{1.835 (0.146)}\\
\bottomrule
\end{tabular}}%
\par
\begin{notes}
     MSE, Bias, Absolute Bias, as well as Coverage and Width of 95\% confidence intervals for \gls{bcf-iv} and \gls{sbcf-iv} across covariate dimensions $P \in \{10, 50, 100\}$ and effect sizes $k \in \{0, 0.2, \ldots, 2\}$. Results are averaged over $M = 500$ Monte Carlo replications with $N = 1{,}000$; standard deviations are reported in parentheses. This table reports the full numerical results underlying Table~\ref{tab:short_precision_uncertainty}, which restricts attention to $k \in \{0, 1, 2\}$ and pools $l_1$ and $l_2$.
\end{notes}
\end{table}



\clearpage
\newpage

\subsection{Results for further simulations}
\label{append:further_precision_results}

We complement the simulation study in Section \ref{ch:sim_study} and the \gls{DGP} outlined in Appendix \ref{append:sim_design} with two additions. First, we implement the instrumental forest of the \gls{grf} framework \citep{athey_generalized_2019} as an additional benchmark method. Second, we analyze an ablation that crosses the discovery step in Algorithm \ref{alg:sbcf-iv} with the CART cost weighting of \eqref{eq:cost}. 
We note that \gls{grf} enters only as an alternative method to estimate the heterogeneity signal in line~4 of Algorithm~\ref{alg:sbcf-iv}, not as a
competing partition method. The instrumental forest of \citet{athey_generalized_2019} targets a unit-level conditional \gls{iv} function rather than an interpretable partition of the covariate space. We include it to test whether the high-dimensional degradation documented for \gls{bcf-iv} in Section \ref{ch:sim_study} is specific to that method or more general to other non-sparse competitors.
Table \ref{tab:competitors} lists the six resulting methods that differ in how they estimate the heterogeneity signal $\widehat{\tau}(x)$ of line 4 in Algorithm \ref{alg:sbcf-iv} and whether they use the varying posterior splitting probabilities to guide the CART construction in line 5 of Algorithm \ref{alg:sbcf-iv}. Based on the methods in Table \ref{tab:competitors}, we empirically evaluate whether using CART cost weighting with the posterior splitting probabilities, as described in \eqref{eq:cost}, improves subgroup discovery and estimation of \gls{ccace}.
\begin{table}[H]
\centering
\footnotesize
\caption{The six competitors as instantiations of Algorithm~\ref{alg:sbcf-iv}.}
\label{tab:competitors}
\begin{tabular}{l l c}
\toprule
Method & Estimation of signal $\widehat{\tau}(x)$ & CART cost weighting via \eqref{eq:cost} \\
\midrule
\gls{bcf-iv} (without cost)              & BCF, BART \citep{hahn_bayesian_2020, hill_bayesian_2011}         & no \\
\gls{bcf-iv} (with cost)    & BCF, BART \citep{hahn_bayesian_2020, hill_bayesian_2011}         & yes \\
\addlinespace
\gls{sbcf-iv} (without cost)             & SBCF, SoftBART \citep{caron_shrinkage_2022, linero_bayesian_2018}      & no \\
\gls{sbcf-iv} (with cost)   & SBCF, SoftBART \citep{caron_shrinkage_2022, linero_bayesian_2018}     & yes \\
\addlinespace
\gls{grf-iv} (without cost)                     & GRF \citep{athey_generalized_2019}      & no \\
\gls{grf-iv} (with cost)          & GRF \citep{athey_generalized_2019}    & yes \\
\bottomrule
\end{tabular}
\begin{notes}
The six competitors rely on Algorithm \ref{alg:sbcf-iv} and only differ in the estimation of the discovery signal $\widehat\tau(x)$ and in whether the CART applies the split-frequency cost $c_{\text{psp}}$.
The cost vector $c_{\text{psp}} = \max(\widehat s_{\text{ITT}_Y})/\widehat s_{\text{ITT}_Y}$ of \eqref{eq:cost} is computed \emph{once} from the SBCF posterior split frequencies $\widehat s_{\text{ITT}_Y}$ and applied unchanged to the "with-cost" methods. Thus, \gls{bcf-iv} (with cost) and \gls{grf-iv} (with cost) pair a non-sparse discovery signal with the sparse-derived cost. \gls{bcf-iv} and \gls{sbcf-iv} form $\widehat\tau(x)$ by dividing the estimated $\widehat{\text{ITT}}_Y(x)$ by the corresponding complier-share estimate $\widehat\pi_C$. The instrumental forest of \citet{athey_generalized_2019} targets \eqref{eq:subgroup-cace} directly, so it requires no explicit ratio.  
\end{notes}
\end{table}
%
The results in Table~\ref{tab:abl-pehe} isolate the choice of the estimation method for the heterogeneity signal in Algorithm \ref{alg:sbcf-iv} as the dominant source of the improvement in precision. The differences between MSE values of \gls{sbcf-iv} with and without CART cost weighting via \eqref{eq:cost} are only marginal. MSE performance of both \gls{sbcf-iv} methods remains rather constant across covariate dimension $P$ and effect size $k$ and they outperform \gls{bcf-iv} and \gls{grf-iv} in almost all scenarios. The non-sparse methods (\gls{bcf-iv} and \gls{grf-iv}) only show lower MSE values compared to \gls{sbcf-iv} for rather homogeneous cases where the effect size $k \le 0.4$, leading to true conditional CACE values that are close to zero via \eqref{eq:tau_CACE_sim}. Both non-sparse methods benefit substantially from the inclusion of CART cost weighting, as they are able to use sparsity-information from SBCF's posterior splitting probabilities for subgroup discovery. This is only a partial remedy that never exceeds \gls{sbcf-iv}'s performance on MSE for moderate to large effect sizes $k$ and all covariate dimensions $P$. 
In Table \ref{tab:abl-bias}, all six methods show bias values close to zero, regardless of CART cost weighting. This indicates that the increase in MSE for the non-sparse methods with rising $P$ and $k$ is driven by estimation instability rather than by systematic over- or underestimation.
Table~\ref{tab:abl-cov} shows the same pattern for interval coverage at the nominal level of $95\%$. Only the two \gls{sbcf-iv} methods reach close to the nominal level across $P$, while \gls{bcf-iv}, \gls{grf-iv}, and their cost-weighted variants all undercover progressively with $k$ and $P$. This is not achieved by widening intervals, as reported in Table \ref{tab:abl-width}. \gls{sbcf-iv} and \gls{grf-iv} have similar mean interval widths, but \gls{grf-iv} undercovers heavily, as Table~\ref{tab:abl-cov} shows. 
For instance, in the high-dimensional case of ($P=100$, $k=2$), \gls{bcf-iv} has an average coverage of roughly 15\%, while \gls{sbcf-iv} reports an exact mean coverage of 95\% and its mean interval width increases by only roughly 20\% compared to \gls{bcf-iv}. Moreover, \gls{sbcf-iv} and \gls{grf-iv} carry near-identical widths, yet \gls{grf-iv} possesses mean coverage of only 63\% (without cost) and 73\% (with cost).

Figure~\ref{fig:abl-FDR-DR} displays the Detection Rate of \eqref{eq:DR} at the tree-level. The two \gls{sbcf-iv} variants recover essentially all true effect-modifying splits by $k \approx 1.2$ at every covariate dimension, whereas \gls{bcf-iv} and both \gls{grf-iv} variants detect fewer splits and fall further behind as $P$ grows. The unit-level metrics in Figure~\ref{fig:abl-precision-F1} show the same pattern. The two \gls{sbcf-iv} variants coincide and dominate on both $F$-score and Precision and remain flat across $P$, while every non-sparse method degrades. Precision separates the methods most sharply at $P = 100$, where \gls{sbcf-iv} reaches roughly $0.98$ and the remaining methods plateau well below. As on the estimation side, cost weighting visibly lifts the
non-sparse methods, as the cost-weighted \gls{bcf-iv} and \gls{grf-iv} curves lie above their counterparts without CART cost weighting. The two \gls{sbcf-iv} curves are indistinguishable, which is consistent with Corollary~\ref{cor:coherence}. Using the sparsity-inducing prior to estimate the heterogeneity signal $\widehat{\tau}(x)$ of line 4 in Algorithm \ref{alg:sbcf-iv} accounts for much of the
gains in subgroup discovery, precision, and uncertainty quantification. The cost weighting is a component that helps only where the ensemble is not already sparsity-regularized, as for the \gls{bcf-iv} and \gls{grf-iv} methods of Table \ref{tab:competitors}. The transfer of \gls{sbcf}'s sparsity information to the discovery tree, as for \gls{sbcf-iv}, adds little on top of the sparsity prior when the signal is already sparse, while partially recovering the non-sparse methods (\gls{bcf-iv}, \gls{grf-iv}) that lack this information.

\begin{figure}[H]\centering
    \caption{Estimation error and interval coverage for the methods of Table \ref{tab:competitors}.}\label{fig:abl-estimation}
    \resizebox{0.9\linewidth}{!}{
    \input{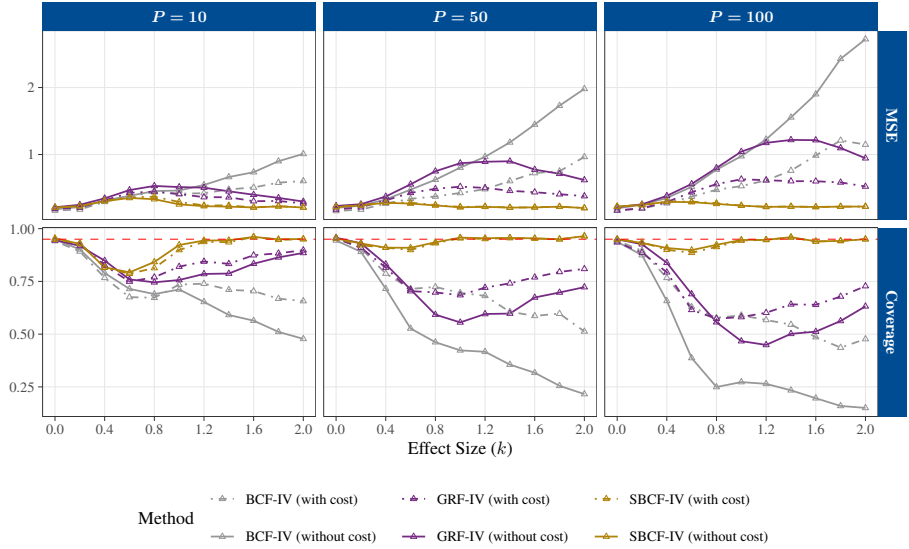}
    }
    \begin{notes}
        MSE (top row; lower is better) and empirical coverage of nominal $95\%$ intervals (bottom row; dashed red line at the nominal coverage rate $0.95$) as a function of the effect size $k$, for covariate dimensions $P\in\{10,50,100\}$ (columns). We compare the six methods of Table \ref{tab:competitors} where solid lines represent methods without CART cost weighting and dot-dashed lines represent methods with CART cost weighting. Values are means over $M=500$ Monte Carlo replications at $N=1{,}000$ and are also reported in Tables~\ref{tab:abl-pehe} and~\ref{tab:abl-cov}.
    \end{notes}
\end{figure}

\begin{figure}[H]\centering
    \caption{Tree-level subgroup discovery for the methods of Table~\ref{tab:competitors}.}\label{fig:abl-FDR-DR}
    \resizebox{0.9\linewidth}{!}{
    \input{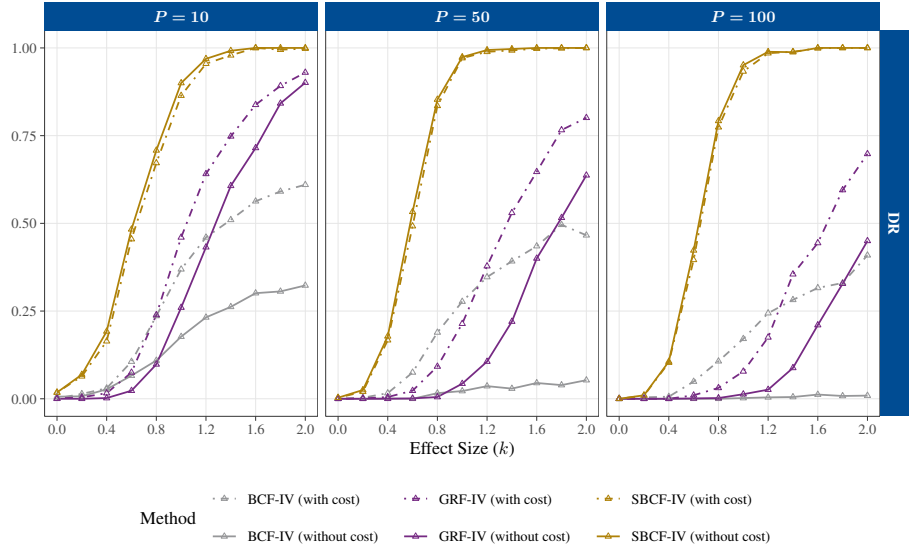}
    }
    \begin{notes}
        Detection Rate (\gls{DR}) as a function of the effect size $k$, for covariate dimensions $P\in\{10,50,100\}$ (columns). \gls{DR} is the share of true effect-modifying splits recovered, defined in \eqref{eq:DR} and evaluated at the tree level with subgroup significance assessed at $\alpha=0.05$ using Holm-adjusted $p$-values. We compare the six methods of Table \ref{tab:competitors} where solid lines represent methods without CART cost weighting and dot-dashed lines represent methods with CART cost weighting. Values are means over $M=500$ Monte Carlo replications at $N=1{,}000$.
    \end{notes}
\end{figure}

\begin{figure}[H]\centering
    \caption{Unit-level discovery performance ($F$-score and precision) for the methods of Table~\ref{tab:competitors}.}\label{fig:abl-precision-F1}
    \resizebox{0.9\linewidth}{!}{
    \input{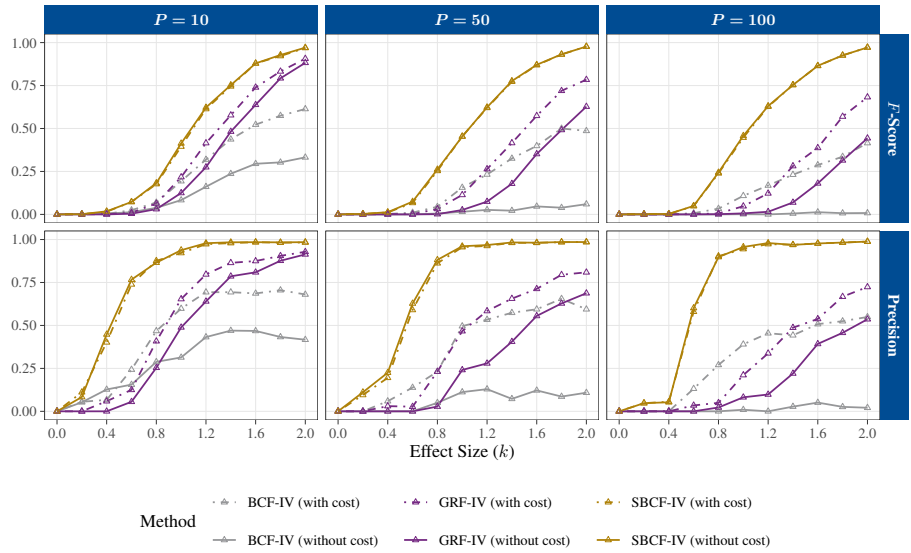}
    }
    \begin{notes}
    $F$-score (top row) and precision (bottom row) for the classification of units into the true effect subgroups, as a function of the effect size $k$, for covariate dimensions $P\in\{10,50,100\}$ (columns). Both metrics are defined in \eqref{eq:class_metrics} and computed over $\mathcal{I}_{\text{inf}}$, with subgroup significance assessed at $\alpha=0.05$ using Holm-adjusted $p$-values. We compare the six methods of Table \ref{tab:competitors} where solid lines represent methods without CART cost weighting and dot-dashed lines represent methods with CART cost weighting. Values are means over $M=500$ Monte Carlo replications at $N=1{,}000$.
    \end{notes}
\end{figure}

\begin{table}[H]
\centering
\caption{Mean squared error (MSE; lower is better) for the methods of Table \ref{tab:competitors}.}
\label{tab:abl-pehe}
\resizebox{\ifdim\width>\linewidth\linewidth\else\width\fi}{!}{
\begin{tabular}[t]{ll>{\raggedleft\arraybackslash}p{0.25cm}rr>{\raggedleft\arraybackslash}p{0.25cm}rr>{\raggedleft\arraybackslash}p{0.25cm}rr}
\toprule
\multicolumn{1}{c}{} & \multicolumn{1}{c}{} & \multicolumn{1}{c}{} & \multicolumn{2}{c}{\gls{bcf-iv}} & \multicolumn{1}{c}{} & \multicolumn{2}{c}{\gls{sbcf-iv}} & \multicolumn{1}{c}{} & \multicolumn{2}{c}{\gls{grf-iv}} \\
\cmidrule(l{3pt}r{3pt}){4-5} \cmidrule(l{3pt}r{3pt}){7-8} \cmidrule(l{3pt}r{3pt}){10-11}
$P$ & $k$ &  & \multicolumn{1}{c}{without cost} & \multicolumn{1}{c}{with cost} &  & \multicolumn{1}{c}{without cost} & \multicolumn{1}{c}{with cost} &  & \multicolumn{1}{c}{without cost} & \multicolumn{1}{c}{with cost}\\
\midrule
 & 0.0 &  & \makecell{0.183 (0.198)} & \makecell{0.164 (0.281)} &  & \makecell{0.209 (0.211)} & \makecell{0.210 (0.223)} &  & \makecell{0.213 (0.159)} & \makecell{0.191 (0.166)}\\
 & 0.4 &  & \makecell{0.313 (0.318)} & \makecell{0.294 (0.164)} &  & \makecell{0.303 (0.222)} & \makecell{0.313 (0.224)} &  & \makecell{0.346 (0.166)} & \makecell{0.316 (0.154)}\\
 & 0.8 &  & \makecell{0.454 (0.309)} & \makecell{0.437 (0.304)} &  & \makecell{0.333 (0.310)} & \makecell{0.371 (0.323)} &  & \makecell{0.531 (0.307)} & \makecell{0.447 (0.310)}\\
 & 1.2 &  & \makecell{0.546 (0.469)} & \makecell{0.423 (0.453)} &  & \makecell{0.229 (0.263)} & \makecell{0.241 (0.289)} &  & \makecell{0.506 (0.450)} & \makecell{0.364 (0.354)}\\
 & 1.6 &  & \makecell{0.736 (0.588)} & \makecell{0.509 (0.462)} &  & \makecell{0.210 (0.213)} & \makecell{0.210 (0.213)} &  & \makecell{0.400 (0.395)} & \makecell{0.304 (0.317)}\\
\multirow{-6}{*}{\raggedright\arraybackslash 10} & 2.0 &  & \makecell{1.010 (0.846)} & \makecell{0.605 (0.591)} &  & \makecell{0.210 (0.210)} & \makecell{0.210 (0.211)} &  & \makecell{0.298 (0.321)} & \makecell{0.271 (0.284)}\\
\cmidrule{1-11}
 & 0.0 &  & \makecell{0.204 (0.222)} & \makecell{0.158 (0.160)} &  & \makecell{0.216 (0.170)} & \makecell{0.221 (0.176)} &  & \makecell{0.233 (0.190)} & \makecell{0.179 (0.178)}\\
 & 0.4 &  & \makecell{0.330 (0.176)} & \makecell{0.268 (0.159)} &  & \makecell{0.279 (0.186)} & \makecell{0.281 (0.183)} &  & \makecell{0.374 (0.168)} & \makecell{0.316 (0.169)}\\
 & 0.8 &  & \makecell{0.624 (0.348)} & \makecell{0.374 (0.290)} &  & \makecell{0.240 (0.249)} & \makecell{0.246 (0.252)} &  & \makecell{0.751 (0.288)} & \makecell{0.488 (0.301)}\\
 & 1.2 &  & \makecell{0.967 (0.599)} & \makecell{0.483 (0.430)} &  & \makecell{0.219 (0.222)} & \makecell{0.224 (0.232)} &  & \makecell{0.893 (0.604)} & \makecell{0.502 (0.446)}\\
 & 1.6 &  & \makecell{1.448 (0.970)} & \makecell{0.724 (0.752)} &  & \makecell{0.211 (0.210)} & \makecell{0.212 (0.211)} &  & \makecell{0.775 (0.760)} & \makecell{0.438 (0.451)}\\
\multirow{-6}{*}{\raggedright\arraybackslash 50} & 2.0 &  & \makecell{1.981 (1.449)} & \makecell{0.965 (1.045)} &  & \makecell{0.202 (0.202)} & \makecell{0.202 (0.202)} &  & \makecell{0.619 (0.744)} & \makecell{0.379 (0.440)}\\
\cmidrule{1-11}
 & 0.0 &  & \makecell{0.197 (0.229)} & \makecell{0.161 (0.182)} &  & \makecell{0.224 (0.165)} & \makecell{0.225 (0.198)} &  & \makecell{0.216 (0.158)} & \makecell{0.166 (0.148)}\\
 & 0.4 &  & \makecell{0.352 (0.269)} & \makecell{0.273 (0.190)} &  & \makecell{0.294 (0.186)} & \makecell{0.299 (0.192)} &  & \makecell{0.383 (0.174)} & \makecell{0.291 (0.150)}\\
 & 0.8 &  & \makecell{0.774 (0.759)} & \makecell{0.474 (0.329)} &  & \makecell{0.263 (0.245)} & \makecell{0.274 (0.251)} &  & \makecell{0.800 (0.236)} & \makecell{0.558 (0.306)}\\
 & 1.2 &  & \makecell{1.227 (0.573)} & \makecell{0.614 (0.507)} &  & \makecell{0.219 (0.217)} & \makecell{0.219 (0.218)} &  & \makecell{1.174 (0.575)} & \makecell{0.617 (0.469)}\\
 & 1.6 &  & \makecell{1.901 (0.987)} & \makecell{0.986 (0.912)} &  & \makecell{0.217 (0.218)} & \makecell{0.218 (0.218)} &  & \makecell{1.215 (0.961)} & \makecell{0.601 (0.567)}\\
\multirow{-6}{*}{\raggedright\arraybackslash 100} & 2.0 &  & \makecell{2.722 (1.561)} & \makecell{1.149 (1.230)} &  & \makecell{0.224 (0.240)} & \makecell{0.224 (0.240)} &  & \makecell{0.945 (0.991)} & \makecell{0.523 (0.566)}\\
\bottomrule
\end{tabular}
}
\begin{notes}
MSE across covariate dimension $P$ and effect size $k$. Means over $M=500$ replications at $N=1{,}000$; standard deviations in parentheses.
\end{notes}
\end{table}

\begin{table}[H]
\centering
\caption{Mean signed bias for the methods of Table \ref{tab:competitors}.}
\label{tab:abl-bias}
\resizebox{\ifdim\width>\linewidth\linewidth\else\width\fi}{!}{
\begin{tabular}[t]{ll>{\raggedleft\arraybackslash}p{0.25cm}rr>{\raggedleft\arraybackslash}p{0.25cm}rr>{\raggedleft\arraybackslash}p{0.25cm}rr}
\toprule
\multicolumn{1}{c}{} & \multicolumn{1}{c}{} & \multicolumn{1}{c}{} & \multicolumn{2}{c}{\gls{bcf-iv}} & \multicolumn{1}{c}{} & \multicolumn{2}{c}{\gls{sbcf-iv}} & \multicolumn{1}{c}{} & \multicolumn{2}{c}{\gls{grf-iv}} \\
\cmidrule(l{3pt}r{3pt}){4-5} \cmidrule(l{3pt}r{3pt}){7-8} \cmidrule(l{3pt}r{3pt}){10-11}
$P$ & $k$ &  & \multicolumn{1}{c}{without cost} & \multicolumn{1}{c}{with cost} &  & \multicolumn{1}{c}{without cost} & \multicolumn{1}{c}{with cost} &  & \multicolumn{1}{c}{without cost} & \multicolumn{1}{c}{with cost}\\
\midrule
 & 0.0 &  & \makecell{0.009 (0.235)} & \makecell{0.011 (0.241)} &  & \makecell{0.005 (0.229)} & \makecell{0.009 (0.229)} &  & \makecell{0.012 (0.238)} & \makecell{0.011 (0.234)}\\
 & 0.4 &  & \makecell{0.026 (0.225)} & \makecell{0.021 (0.228)} &  & \makecell{0.010 (0.228)} & \makecell{0.008 (0.221)} &  & \makecell{0.016 (0.227)} & \makecell{0.018 (0.221)}\\
 & 0.8 &  & \makecell{0.000 (0.275)} & \makecell{0.008 (0.295)} &  & \makecell{0.019 (0.314)} & \makecell{0.017 (0.309)} &  & \makecell{-0.008 (0.252)} & \makecell{0.002 (0.261)}\\
 & 1.2 &  & \makecell{-0.017 (0.306)} & \makecell{-0.021 (0.349)} &  & \makecell{0.006 (0.307)} & \makecell{0.008 (0.304)} &  & \makecell{0.011 (0.297)} & \makecell{0.013 (0.316)}\\
 & 1.6 &  & \makecell{0.015 (0.387)} & \makecell{0.032 (0.436)} &  & \makecell{0.001 (0.334)} & \makecell{0.001 (0.334)} &  & \makecell{0.004 (0.365)} & \makecell{0.005 (0.356)}\\
\multirow{-6}{*}{\raggedright\arraybackslash 10} & 2.0 &  & \makecell{-0.032 (0.432)} & \makecell{0.015 (0.485)} &  & \makecell{-0.005 (0.334)} & \makecell{-0.004 (0.334)} &  & \makecell{-0.002 (0.366)} & \makecell{0.002 (0.359)}\\
\cmidrule{1-11}
 & 0.0 &  & \makecell{-0.006 (0.235)} & \makecell{-0.005 (0.235)} &  & \makecell{-0.004 (0.232)} & \makecell{-0.005 (0.233)} &  & \makecell{-0.004 (0.236)} & \makecell{-0.003 (0.235)}\\
 & 0.4 &  & \makecell{-0.005 (0.227)} & \makecell{-0.010 (0.224)} &  & \makecell{-0.005 (0.247)} & \makecell{-0.008 (0.246)} &  & \makecell{-0.004 (0.225)} & \makecell{-0.004 (0.227)}\\
 & 0.8 &  & \makecell{0.016 (0.244)} & \makecell{0.011 (0.271)} &  & \makecell{0.018 (0.320)} & \makecell{0.018 (0.318)} &  & \makecell{0.010 (0.244)} & \makecell{0.009 (0.255)}\\
 & 1.2 &  & \makecell{-0.016 (0.249)} & \makecell{-0.020 (0.336)} &  & \makecell{-0.011 (0.320)} & \makecell{-0.011 (0.321)} &  & \makecell{-0.016 (0.251)} & \makecell{-0.005 (0.299)}\\
 & 1.6 &  & \makecell{-0.012 (0.278)} & \makecell{-0.033 (0.406)} &  & \makecell{0.007 (0.332)} & \makecell{0.009 (0.333)} &  & \makecell{-0.005 (0.307)} & \makecell{0.004 (0.325)}\\
\multirow{-6}{*}{\raggedright\arraybackslash 50} & 2.0 &  & \makecell{0.004 (0.282)} & \makecell{-0.014 (0.502)} &  & \makecell{0.007 (0.325)} & \makecell{0.007 (0.325)} &  & \makecell{0.003 (0.350)} & \makecell{0.004 (0.346)}\\
\cmidrule{1-11}
 & 0.0 &  & \makecell{-0.020 (0.239)} & \makecell{-0.020 (0.241)} &  & \makecell{-0.018 (0.247)} & \makecell{-0.019 (0.245)} &  & \makecell{-0.019 (0.241)} & \makecell{-0.019 (0.243)}\\
 & 0.4 &  & \makecell{-0.003 (0.238)} & \makecell{-0.004 (0.236)} &  & \makecell{0.002 (0.247)} & \makecell{-0.002 (0.244)} &  & \makecell{-0.004 (0.234)} & \makecell{-0.007 (0.237)}\\
 & 0.8 &  & \makecell{0.006 (0.236)} & \makecell{0.006 (0.256)} &  & \makecell{-0.002 (0.327)} & \makecell{-0.002 (0.322)} &  & \makecell{0.001 (0.237)} & \makecell{0.003 (0.237)}\\
 & 1.2 &  & \makecell{0.002 (0.234)} & \makecell{-0.005 (0.307)} &  & \makecell{-0.009 (0.330)} & \makecell{-0.009 (0.329)} &  & \makecell{-0.003 (0.237)} & \makecell{0.005 (0.280)}\\
 & 1.6 &  & \makecell{0.009 (0.240)} & \makecell{0.020 (0.364)} &  & \makecell{-0.001 (0.329)} & \makecell{-0.002 (0.329)} &  & \makecell{-0.002 (0.275)} & \makecell{-0.015 (0.309)}\\
\multirow{-6}{*}{\raggedright\arraybackslash 100} & 2.0 &  & \makecell{0.023 (0.256)} & \makecell{0.049 (0.464)} &  & \makecell{0.038 (0.336)} & \makecell{0.038 (0.336)} &  & \makecell{0.028 (0.317)} & \makecell{0.064 (0.347)}\\
\bottomrule
\end{tabular}
}
\begin{notes}
Mean signed bias across covariate dimension $P$ and effect size $k$ for the methods of Table \ref{tab:competitors}. Means over $M=500$ replications at $N=1{,}000$; standard deviations in parentheses.
\end{notes}
\end{table}

\begin{table}[H]
\centering
\caption{Mean interval coverage for the methods of Table \ref{tab:competitors}.}
\label{tab:abl-cov}
\resizebox{\ifdim\width>\linewidth\linewidth\else\width\fi}{!}{
\begin{tabular}[t]{ll>{\raggedleft\arraybackslash}p{0.25cm}rr>{\raggedleft\arraybackslash}p{0.25cm}rr>{\raggedleft\arraybackslash}p{0.25cm}rr}
\toprule
\multicolumn{1}{c}{} & \multicolumn{1}{c}{} & \multicolumn{1}{c}{} & \multicolumn{2}{c}{\gls{bcf-iv}} & \multicolumn{1}{c}{} & \multicolumn{2}{c}{\gls{sbcf-iv}} & \multicolumn{1}{c}{} & \multicolumn{2}{c}{\gls{grf-iv}} \\
\cmidrule(l{3pt}r{3pt}){4-5} \cmidrule(l{3pt}r{3pt}){7-8} \cmidrule(l{3pt}r{3pt}){10-11}
$P$ & $k$ &  & \multicolumn{1}{c}{without cost} & \multicolumn{1}{c}{with cost} &  & \multicolumn{1}{c}{without cost} & \multicolumn{1}{c}{with cost} &  & \multicolumn{1}{c}{without cost} & \multicolumn{1}{c}{with cost}\\
\midrule
 & 0.0 &  & \makecell{0.949 (0.162)} & \makecell{0.941 (0.190)} &  & \makecell{0.955 (0.131)} & \makecell{0.948 (0.141)} &  & \makecell{0.946 (0.136)} & \makecell{0.945 (0.151)}\\
 & 0.4 &  & \makecell{0.790 (0.211)} & \makecell{0.767 (0.204)} &  & \makecell{0.817 (0.186)} & \makecell{0.815 (0.181)} &  & \makecell{0.848 (0.142)} & \makecell{0.825 (0.170)}\\
 & 0.8 &  & \makecell{0.690 (0.338)} & \makecell{0.673 (0.348)} &  & \makecell{0.842 (0.266)} & \makecell{0.813 (0.275)} &  & \makecell{0.746 (0.233)} & \makecell{0.770 (0.258)}\\
 & 1.2 &  & \makecell{0.653 (0.340)} & \makecell{0.739 (0.313)} &  & \makecell{0.944 (0.178)} & \makecell{0.937 (0.190)} &  & \makecell{0.785 (0.272)} & \makecell{0.844 (0.255)}\\
 & 1.6 &  & \makecell{0.564 (0.331)} & \makecell{0.704 (0.296)} &  & \makecell{0.961 (0.139)} & \makecell{0.961 (0.139)} &  & \makecell{0.834 (0.257)} & \makecell{0.874 (0.242)}\\
\multirow{-6}{*}{\raggedright\arraybackslash 10} & 2.0 &  & \makecell{0.478 (0.329)} & \makecell{0.657 (0.304)} &  & \makecell{0.951 (0.149)} & \makecell{0.950 (0.150)} &  & \makecell{0.885 (0.230)} & \makecell{0.898 (0.220)}\\
\cmidrule{1-11}
 & 0.0 &  & \makecell{0.944 (0.197)} & \makecell{0.955 (0.161)} &  & \makecell{0.954 (0.112)} & \makecell{0.955 (0.113)} &  & \makecell{0.957 (0.111)} & \makecell{0.954 (0.138)}\\
 & 0.4 &  & \makecell{0.715 (0.217)} & \makecell{0.787 (0.211)} &  & \makecell{0.910 (0.142)} & \makecell{0.910 (0.137)} &  & \makecell{0.832 (0.148)} & \makecell{0.816 (0.180)}\\
 & 0.8 &  & \makecell{0.462 (0.393)} & \makecell{0.724 (0.343)} &  & \makecell{0.935 (0.176)} & \makecell{0.931 (0.181)} &  & \makecell{0.592 (0.214)} & \makecell{0.698 (0.305)}\\
 & 1.2 &  & \makecell{0.416 (0.387)} & \makecell{0.682 (0.338)} &  & \makecell{0.955 (0.143)} & \makecell{0.951 (0.150)} &  & \makecell{0.596 (0.322)} & \makecell{0.720 (0.331)}\\
 & 1.6 &  & \makecell{0.317 (0.328)} & \makecell{0.587 (0.364)} &  & \makecell{0.954 (0.145)} & \makecell{0.953 (0.146)} &  & \makecell{0.674 (0.329)} & \makecell{0.770 (0.316)}\\
\multirow{-6}{*}{\raggedright\arraybackslash 50} & 2.0 &  & \makecell{0.216 (0.274)} & \makecell{0.513 (0.324)} &  & \makecell{0.965 (0.131)} & \makecell{0.965 (0.131)} &  & \makecell{0.723 (0.356)} & \makecell{0.809 (0.328)}\\
\cmidrule{1-11}
 & 0.0 &  & \makecell{0.938 (0.216)} & \makecell{0.943 (0.184)} &  & \makecell{0.950 (0.124)} & \makecell{0.948 (0.123)} &  & \makecell{0.949 (0.134)} & \makecell{0.935 (0.189)}\\
 & 0.4 &  & \makecell{0.658 (0.205)} & \makecell{0.767 (0.215)} &  & \makecell{0.908 (0.125)} & \makecell{0.900 (0.133)} &  & \makecell{0.838 (0.142)} & \makecell{0.793 (0.193)}\\
 & 0.8 &  & \makecell{0.249 (0.321)} & \makecell{0.577 (0.393)} &  & \makecell{0.923 (0.174)} & \makecell{0.913 (0.188)} &  & \makecell{0.556 (0.180)} & \makecell{0.573 (0.342)}\\
 & 1.2 &  & \makecell{0.264 (0.335)} & \makecell{0.567 (0.384)} &  & \makecell{0.947 (0.157)} & \makecell{0.947 (0.157)} &  & \makecell{0.449 (0.303)} & \makecell{0.601 (0.359)}\\
 & 1.6 &  & \makecell{0.196 (0.280)} & \makecell{0.486 (0.379)} &  & \makecell{0.940 (0.163)} & \makecell{0.939 (0.164)} &  & \makecell{0.512 (0.350)} & \makecell{0.639 (0.364)}\\
\multirow{-6}{*}{\raggedright\arraybackslash 100} & 2.0 &  & \makecell{0.150 (0.225)} & \makecell{0.477 (0.350)} &  & \makecell{0.950 (0.155)} & \makecell{0.950 (0.155)} &  & \makecell{0.632 (0.352)} & \makecell{0.728 (0.369)}\\
\bottomrule
\end{tabular}
}
\begin{notes}
Empirical coverage of nominal $95\%$ intervals for the methods of Table \ref{tab:competitors}, across covariate dimension $P$ and effect size $k$. Means over $M=500$ replications at $N=1{,}000$; standard deviations in parentheses.
\end{notes}
\end{table}

\begin{table}[H]
\centering
\caption{Mean interval width for the methods of Table \ref{tab:competitors}.}
\label{tab:abl-width}
\resizebox{\ifdim\width>\linewidth\linewidth\else\width\fi}{!}{
\begin{tabular}[t]{ll>{\raggedleft\arraybackslash}p{0.25cm}rr>{\raggedleft\arraybackslash}p{0.25cm}rr>{\raggedleft\arraybackslash}p{0.25cm}rr}
\toprule
\multicolumn{1}{c}{} & \multicolumn{1}{c}{} & \multicolumn{1}{c}{} & \multicolumn{2}{c}{\gls{bcf-iv}} & \multicolumn{1}{c}{} & \multicolumn{2}{c}{\gls{sbcf-iv}} & \multicolumn{1}{c}{} & \multicolumn{2}{c}{\gls{grf-iv}} \\
\cmidrule(l{3pt}r{3pt}){4-5} \cmidrule(l{3pt}r{3pt}){7-8} \cmidrule(l{3pt}r{3pt}){10-11}
$P$ & $k$ &  & \multicolumn{1}{c}{without cost} & \multicolumn{1}{c}{with cost} &  & \multicolumn{1}{c}{without cost} & \multicolumn{1}{c}{with cost} &  & \multicolumn{1}{c}{without cost} & \multicolumn{1}{c}{with cost}\\
\midrule
 & 0.0 &  & \makecell{1.442 (0.209)} & \makecell{1.343 (0.237)} &  & \makecell{1.575 (0.196)} & \makecell{1.588 (0.182)} &  & \makecell{1.721 (0.123)} & \makecell{1.570 (0.212)}\\
 & 0.4 &  & \makecell{1.496 (0.213)} & \makecell{1.430 (0.225)} &  & \makecell{1.633 (0.200)} & \makecell{1.637 (0.186)} &  & \makecell{1.741 (0.126)} & \makecell{1.626 (0.189)}\\
 & 0.8 &  & \makecell{1.580 (0.212)} & \makecell{1.541 (0.241)} &  & \makecell{1.781 (0.181)} & \makecell{1.776 (0.177)} &  & \makecell{1.788 (0.116)} & \makecell{1.752 (0.168)}\\
 & 1.2 &  & \makecell{1.640 (0.216)} & \makecell{1.635 (0.206)} &  & \makecell{1.828 (0.141)} & \makecell{1.827 (0.144)} &  & \makecell{1.808 (0.144)} & \makecell{1.790 (0.174)}\\
 & 1.6 &  & \makecell{1.662 (0.184)} & \makecell{1.663 (0.184)} &  & \makecell{1.831 (0.138)} & \makecell{1.831 (0.138)} &  & \makecell{1.805 (0.162)} & \makecell{1.781 (0.185)}\\
\multirow{-6}{*}{\raggedright\arraybackslash 10} & 2.0 &  & \makecell{1.683 (0.183)} & \makecell{1.672 (0.191)} &  & \makecell{1.822 (0.139)} & \makecell{1.822 (0.139)} &  & \makecell{1.800 (0.160)} & \makecell{1.796 (0.162)}\\
\cmidrule{1-11}
 & 0.0 &  & \makecell{1.375 (0.222)} & \makecell{1.383 (0.260)} &  & \makecell{1.792 (0.137)} & \makecell{1.773 (0.147)} &  & \makecell{1.759 (0.128)} & \makecell{1.529 (0.260)}\\
 & 0.4 &  & \makecell{1.369 (0.227)} & \makecell{1.411 (0.240)} &  & \makecell{1.798 (0.145)} & \makecell{1.781 (0.157)} &  & \makecell{1.759 (0.118)} & \makecell{1.548 (0.260)}\\
 & 0.8 &  & \makecell{1.427 (0.236)} & \makecell{1.481 (0.228)} &  & \makecell{1.808 (0.149)} & \makecell{1.807 (0.151)} &  & \makecell{1.778 (0.115)} & \makecell{1.632 (0.233)}\\
 & 1.2 &  & \makecell{1.507 (0.272)} & \makecell{1.543 (0.226)} &  & \makecell{1.841 (0.140)} & \makecell{1.841 (0.140)} &  & \makecell{1.815 (0.140)} & \makecell{1.675 (0.244)}\\
 & 1.6 &  & \makecell{1.521 (0.283)} & \makecell{1.558 (0.233)} &  & \makecell{1.830 (0.142)} & \makecell{1.830 (0.142)} &  & \makecell{1.814 (0.155)} & \makecell{1.724 (0.226)}\\
\multirow{-6}{*}{\raggedright\arraybackslash 50} & 2.0 &  & \makecell{1.573 (0.256)} & \makecell{1.581 (0.221)} &  & \makecell{1.830 (0.144)} & \makecell{1.830 (0.144)} &  & \makecell{1.795 (0.184)} & \makecell{1.749 (0.208)}\\
\cmidrule{1-11}
 & 0.0 &  & \makecell{1.308 (0.216)} & \makecell{1.344 (0.280)} &  & \makecell{1.801 (0.131)} & \makecell{1.786 (0.150)} &  & \makecell{1.750 (0.118)} & \makecell{1.444 (0.271)}\\
 & 0.4 &  & \makecell{1.309 (0.259)} & \makecell{1.349 (0.258)} &  & \makecell{1.815 (0.119)} & \makecell{1.803 (0.137)} &  & \makecell{1.765 (0.116)} & \makecell{1.463 (0.274)}\\
 & 0.8 &  & \makecell{1.335 (0.296)} & \makecell{1.402 (0.280)} &  & \makecell{1.826 (0.142)} & \makecell{1.822 (0.147)} &  & \makecell{1.788 (0.116)} & \makecell{1.479 (0.285)}\\
 & 1.2 &  & \makecell{1.425 (0.280)} & \makecell{1.455 (0.252)} &  & \makecell{1.839 (0.141)} & \makecell{1.839 (0.139)} &  & \makecell{1.818 (0.119)} & \makecell{1.587 (0.260)}\\
 & 1.6 &  & \makecell{1.446 (0.261)} & \makecell{1.485 (0.277)} &  & \makecell{1.848 (0.146)} & \makecell{1.848 (0.146)} &  & \makecell{1.843 (0.147)} & \makecell{1.665 (0.245)}\\
\multirow{-6}{*}{\raggedright\arraybackslash 100} & 2.0 &  & \makecell{1.512 (0.293)} & \makecell{1.554 (0.240)} &  & \makecell{1.831 (0.146)} & \makecell{1.831 (0.146)} &  & \makecell{1.847 (0.151)} & \makecell{1.725 (0.222)}\\
\bottomrule
\end{tabular}
}
\begin{notes}
Mean $95\%$ interval width across covariate dimension $P$ and effect size $k$ for the methods of Table \ref{tab:competitors}. Means over $M=500$ replications at $N=1{,}000$; standard deviations in parentheses. At a given $(P,k)$, \gls{sbcf-iv} and \gls{grf-iv} have comparable widths despite large coverage differences (cf.\ Table~\ref{tab:abl-cov}), so the coverage gap is not a width artifact.
\end{notes}
\end{table}

\section{Supplementary materials for empirical applications}

\subsection{Oregon Health Insurance Experiment (OHIE)}
\label{sec:OHIE}

We evaluate \gls{sbcf-iv} on the OHIE data, replicating the empirical setup of \citet{johnson_detecting_2022}. The OHIE assembled administrative records (hospital discharges, credit reports, and mortality), survey data (healthcare utilization, financial strain, and self-reported health), and pre-randomization demographic information. \citet{johnson_detecting_2022} match individuals on the pre-randomization demographic variables, including sex, age, language preference at lottery sign-up (English or other), Metropolitan Statistical Area (MSA) residency, education (less than high school, high school diploma or GED, vocational or two-year degree, four-year degree or higher), and self-identified race. Because Hispanic and Black self-identification as well as education contained missing values, missingness indicators are included in the matching set. The outcome $Y_i$ is the number of days in the past 30 days on which poor physical or mental health did not impair usual activities. The endogenous treatment indicates Medicaid enrollment and the instrument indicates lottery selection \citep{finkelstein_oregon_2012, johnson_detecting_2022}.

Figure \ref{fig:tree_subgroups} shows that \gls{sbcf-iv} partitions on age at the root and on English (lottery sign-up language preference) within both age branches. The education covariate appears as a tertiary split within the younger non-English-preferring subtree. We receive seven leaves with conditional CACEs spanning $-6.017$ to $3.630$ days, of which only one subgroup has a statistically significant $p$-value at the 10\% level: English-preferring compliers aged between 38 and 59, with $\widehat{\tau}^{\mathrm{CACE}}(x) = 2.263$ days (adjusted $p = 0.0945$). This leaf carries roughly half of the inference sample and a complier share of $0.324$. Its positive sign and the middle-age, English-preferring composition closely mirror the Medicaid treatment effects reported by \citet{johnson_detecting_2022}.

The remaining six leaves carry inference-sample shares between 4\% and 14\% and produce estimates that should be read as exploratory rather than substantive findings. Two leaves carry negative point estimates ($\widehat{\tau}^{\mathrm{CACE}}(x) = -6.017$ for non-English-preferring compliers aged $\geq 60$ (4\% share) and $\widehat{\tau}^{\mathrm{CACE}}(x) = -1.403$ for non-English-preferring compliers aged 26--37 with above-median education (9\%)), but are not statistically significant. We do not interpret them as evidence of adverse Medicaid effects, as they are most plausibly small-sample noise in thin demographic strata. The non-significant positive leaves at $\widehat{\tau}^{\mathrm{CACE}}(x) = 3.630$ (English-preferring under 38, 11\%) and $\widehat{\tau}^{\mathrm{CACE}}(x) = 1.563$ (non-English-preferring aged 26--37 with low education, 14\%) suggest that positive effects may extend beyond the central English-preferring middle-aged subgroup, but are not statistically resolvable at this sample size and tree depth. 

Relative to the two complier subgroups identified by \citet{johnson_detecting_2022} ((i) non-Asian, English-preferring males over age 36, and (ii) compliers under age 36 who prefer English with at most a high-school education or GED), \gls{sbcf-iv} recovers a partition close to a coarsened version of their subgroup (i): English preference and middle age are the leading dimensions, but our discovery tree does not split on sex or race and therefore cannot reproduce their male-only refinement. The qualitative location of the dominant positive Medicaid effect (older, English-preferring compliers) is preserved across both methods, while the demographic resolution differs.

\subsection{401(k) eligibility and retirement savings}
\label{sec:401k}

We apply \gls{sbcf-iv} to the 401(k) data of \citet{chernozhukov_doubledebiased_2018} and \citet{bach_2024}, an instrumental variable analysis benchmark in which the endogenous decision to participate in an employer-sponsored retirement plan is instrumented by eligibility. \citet{belloni_program_2017} provide a high-dimensional treatment of this design. Conditional on a set of job-choice covariates, eligibility is plausibly unconfounded and satisfies the exclusion restriction, while participation remains confounded by unobserved preferences for saving. In Subsection \ref{sec:401k}, we use \gls{sbcf-iv} to discover the subpopulations of eligible households for which participation has the largest impact on net financial assets.

The data are drawn from the 1991 Survey of Income and Program Participation, which has become a canonical testbed for instrumental variable methods in labor and public economics. The outcome $Y_i$ is a household's net financial assets (\gls{IRA} and 401(k) balances, checking accounts, savings bonds, and related holdings, net of non-mortgage debt). The endogenous treatment $W_i$ is an indicator for participation in a 401(k) plan, while the instrument $Z_i$ is an indicator for eligibility to enroll in such a plan through one's employer. 
Following \citet{poterba_401k_1992} and \citet{poterba_401k_1995}, the identifying argument is that, conditional on a small set of job-choice covariates $X_i$ (notably income, together with age, family size, education, marital status, two-earner and defined-benefit pension status, \gls{IRA} participation, and home ownership), whether an employer offered a 401(k) plan around the time of the survey can be treated as exogenous to the household's saving behavior, while the decision to actually participate conditional on eligibility remains endogenous. This places us in the irregular assignment mechanism of Section~\ref{ch:PO_irreg}: $Z_i$ is plausibly unconfounded given $X_i$ and satisfies the exclusion restriction (eligibility affects assets only through participation), whereas $W_i$ is confounded by unobserved preferences for saving. The \gls{sbcf-iv} algorithm is used to discover for which subpopulations of eligible households participation has the largest impact on net financial assets.

We show in Figure \ref{fig:tree_subgroups_401k} that \gls{sbcf-iv} partitions the data primarily on the variable \textit{inc} (household income), which appears at the root and at every internal node except one. The variable \textit{pira} (\gls{IRA} holdings) enters only as a secondary split within a thin upper-middle-income window. None of the other portfolio or household-structure binaries (\textit{db}, \textit{hown}, \textit{marr}, \textit{twoearn}) survive into the depth-four discovery tree. The dominance of income aligns \gls{sbcf-iv}'s discovered structure with the lifecycle and earnings-gradient emphasis of the classical 401(k) saving literature \citep{poterba_401k_1992, poterba_401k_1995, engen_illusory_1996, engen_effects_2000}: heterogeneity in $\tau^{\mathrm{CACE}}(x)$ is concentrated along the income margin, with effects rising broadly from the lower-middle-income mass into the upper-middle-income range before becoming statistically indistinguishable from zero in the thin high-income tail.

Six of the seven leaves carry inference-sample shares below 7\%, and four have shares at or below 1\%. These small leaves correspond to a few dozen households each and should be read as small-sample artifacts rather than substantive heterogeneity: none of them display statistically significant Holm-adjusted $p$-values at the $10\%$ level. Only two leaves remain significant at the 10\% level after adjustment: the bulk subgroup with household income lower than $68{,}810$ (89\% of observations in the inference sample and $\widehat{\tau}^{\mathrm{CACE}}(x) = \$17{,}818$ for this subgroup), and the small upper-middle-income subgroup with household income between $92{,}690$ and $110{,}400$ (2.6\%, $\widehat{\tau}^{\mathrm{CACE}}(x) = \$53{,}422$). 
The qualitative ordering (larger effects in the upper-middle-income subgroup than in the lower-income mass) is consistent with the finding in \citet{engen_effects_2000} that 401(k) effects on net financial assets are larger for higher-earnings groups, who hold the bulk of 401(k) assets and for whom contributions are most likely to substitute from taxable accounts rather than represent new saving \citep{chetty_active_2014}.

Two structural parallels with OHIE in Section \ref{sec:OHIE} persist. Complier shares vary only modestly relative to the order-of-magnitude changes in the conditional CACE estimates. Heterogeneity is identified almost entirely from variation in the conditional \gls{itt}, matching the simulation evidence in Section~\ref{ch:sim_study} that \gls{sbcf-iv} reliably recovers ITT-driven partitions. Moreover, both statistically significant leaves carry estimates above the cross-fitted \gls{cace} on the same trimmed sample ($\approx \$9{,}000$--$\$13{,}000$ in \citet{chernozhukov_doubledebiased_2018}): the bulk leaf at \$17{,}818 lies a modest \$5{,}000--\$9{,}000 above that identification benchmark, while the upper-middle-income leaf at \$53{,}422 lies further above and should be read as partly inflated by sample selection. The residual upward bias is consistent with longstanding concerns that 401(k) effects on net financial assets are inflated by selection on saver type \citep{engen_illusory_1996, engen_effects_2000} and by mechanical accumulation in tax-favored accounts that does not represent net new saving \citep{chetty_active_2014}. We therefore read the leaf estimates of the statistically significant subgroups as mildly upward-biased point estimates of the structural participation effect. Further, we treat the high-income tail leaves as exploratory and emphasize the shape of the discovered partition (income-dominated, with effects rising through middle income before destabilizing at the top) over the effect estimate magnitudes of any individual subgroup.


\begin{sidewaysfigure}
   \centering
   \caption{Discovered partition from applying \gls{sbcf-iv} with Algorithm \ref{alg:sbcf-iv} to the Oregon Health Insurance Experiment (OHIE).}
   \label{fig:tree_subgroups}
    \begin{singlespace}
	
	\begin{small}
		
		\resizebox{\textwidth}{!}{
			
			\begin{tikzpicture}[
				level distance=55mm,
				level 1/.style={sibling distance=230mm}, 
				level 2/.style={sibling distance=100mm}, 
				level 3/.style={sibling distance=85mm}, 
				level 4/.style={sibling distance=57mm},
				every node/.style={align=center},
				decisionbox/.style={
					rectangle, draw=blue!60, thick, fill=blue!10,
					align=center,
					inner xsep=-2.5pt, inner ysep=4pt,
					font=\Large
				},
				leaf/.style={
					rectangle, draw=black!60, thick, fill=gray!10,
					align=center, 
					inner xsep=-2.5pt, inner ysep=4pt,
					font=\Large
				},
				highlightleaf/.style={
					rectangle, draw=red!80!black, thick, fill=gold!75,
					align=center,
					inner xsep=-2.5pt, inner ysep=4pt,
					font=\Large
				},
				edge from parent path={
					(\tikzparentnode.south) -- ++(0,-10pt) -| (\tikzchildnode.north)
				}
				]
				
				\node {
					\begin{tikzpicture}[baseline]
						\node[decisionbox] (box) {
							\begin{tabular}{>{\raggedright\arraybackslash}m{3cm} >{\raggedleft\arraybackslash}m{1.5cm}}
								$\widehat{\tau}^{\text{CACE}}(x)$ & 1.711 \\
								$\mathcal{I}_{\text{inf}}$ & $100 \%$ \\
								$\widehat{\pi}_C(x)$ & 0.294 \\
							\end{tabular}
						};
						\node[below=5pt of box] (rule) {\Large Age $\le$ 37.5};
					\end{tikzpicture}
				}
				child {node {
						\begin{tikzpicture}[baseline]
							\node[decisionbox] (box) {
								\begin{tabular}{>{\raggedright\arraybackslash}m{3cm} >{\raggedleft\arraybackslash}m{1.5cm}}
									$\widehat{\tau}^{\text{CACE}}(x)$ & 2.174 \\
									$\mathcal{I}_{\text{inf}}$                  & $63 \%$ \\
									$\widehat{\pi}_C(x)$         & 0.307 \\
								\end{tabular}
							};
							\node[below=5pt of box] {\Large English $=$ 0};
						\end{tikzpicture}
					}
					child { child { child {node[leaf] {
									\begin{tabular}{>{\raggedright\arraybackslash}m{3.1cm} >{\raggedleft\arraybackslash}m{1.5cm}}
										$\widehat{\tau}^{\text{CACE}}(x)$ 	& $0.441$ \\
										$\mathcal{I}_{\text{inf}}$ 					& $5 \%$ \\
										$\widehat{\pi}_C(x)$ 			& $0.202$ \\
									\end{tabular}	 
								}
					}}}
					child {
						node {
							\begin{tikzpicture}[baseline]
								\node[decisionbox] (box) {
									\begin{tabular}{>{\raggedright\arraybackslash}m{3.1cm} >{\raggedleft\arraybackslash}m{1.5cm}}
										$\widehat{\tau}^{\text{CACE}}(x)$ 	& 2.136 \\
										$\mathcal{I}_{\text{inf}}$ 					& $59 \%$ \\
										$\widehat{\pi}_C(x)$ 			& 0.315 \\
									\end{tabular}
								};
								\node[below=5pt of box] {\Large Age $<$ 26};
							\end{tikzpicture}
						}
						child{child {node[leaf] {
									\begin{tabular}{>{\raggedright\arraybackslash}m{3.1cm} >{\raggedleft\arraybackslash}m{1.5cm}}
										$\widehat{\tau}^{\text{CACE}}(x)$	& 3.630 \\
										$\mathcal{I}_{\text{inf}}$					& $11 \%$ \\
										$\widehat{\pi}_C(x)$ 			& 0.247\\
									\end{tabular}
							}} 
						}
						child {node {
								\begin{tikzpicture}[baseline]
									\node[decisionbox] (box) {
										\begin{tabular}{>{\raggedright\arraybackslash}m{3.1cm} >{\raggedleft\arraybackslash}m{1.5cm}}
											$\widehat{\tau}^{\text{CACE}}(x)$	& 0.520 \\
											$\mathcal{I}_{\text{inf}}$					& $23 \%$ \\
											$\widehat{\pi}_C(x)$ 			& 0.298 \\
										\end{tabular}
									};
									\node[below=5pt of box] {\Large Education $<$ 0.67};
								\end{tikzpicture}
							}
							child {node[leaf] {
									\begin{tabular}{>{\raggedright\arraybackslash}m{3.1cm} >{\raggedleft\arraybackslash}m{1.5cm}}
										$\widehat{\tau}^{\text{CACE}}(x)$	& $1.563$\\
										$\mathcal{I}_{\text{inf}}$					& $14 \%$ \\
										$\widehat{\pi}_C(x)$ 			& 0.309 \\
									\end{tabular}
								}
							}
							child {node[leaf] {
									\begin{tabular}{>{\raggedright\arraybackslash}m{3.1cm} >{\raggedleft\arraybackslash}m{1.5cm}}
										$\widehat{\tau}^{\text{CACE}}(x)$ 	& $-1.403$\\
										$\mathcal{I}_{\text{inf}}$					& $9 \%$ \\
										$\widehat{\pi}_C(x)$ 			& 0.281 \\
									\end{tabular}
								}
							}
						}
					}
				edge from parent node[near start, right, font=\Large, draw=black, thick, fill=white, inner sep=3pt, rounded corners=2pt]{TRUE}
				}
				child {
					node {
						\begin{tikzpicture}[baseline]
							\node[decisionbox] (box) {
								\begin{tabular}{>{\raggedright\arraybackslash}m{3.1cm} >{\raggedleft\arraybackslash}m{1.5cm}}
									$\widehat{\tau}^{\text{CACE}}(x)$ 	& 0.977 \\
									$\mathcal{I}_{\text{inf}}$					& $37 \%$ \\
									$\widehat{\pi}_C(x)$ 			& 0.273 \\
								\end{tabular}
							};
							\node[below=5pt of box] {\Large English $=$ 0};
						\end{tikzpicture}
					}
					child{ child{ child {node[leaf] {
									\begin{tabular}{>{\raggedright\arraybackslash}m{3.1cm} >{\raggedleft\arraybackslash}m{1.5cm}}
										$\widehat{\tau}^{\text{CACE}}(x)$	& -6.017 \\
										$\mathcal{I}_{\text{inf}}$					& $4 \%$ \\
										$\widehat{\pi}_C(x)$ 			& 0.192 \\
									\end{tabular}
								}
					}}}
					child{node {
							\begin{tikzpicture}[baseline]
								\node[decisionbox] (box) {
									\begin{tabular}{>{\raggedright\arraybackslash}m{3.1cm} >{\raggedleft\arraybackslash}m{1.5cm}}
										$\widehat{\tau}^{\text{CACE}}(x)$	& 1.415 \\
										$\mathcal{I}_{\text{inf}}$					& $33 \%$ \\
										$\widehat{\pi}_C(x)$ 			& 0.282 \\
									\end{tabular}
								};
								\node[below=5pt of box] {\Large Age $\geq$ 60};
							\end{tikzpicture}
						}
						child{child {node[leaf] {
									\begin{tabular}{>{\raggedright\arraybackslash}m{3cm} >{\raggedleft\arraybackslash}m{1.5cm}}
										$\widehat{\tau}^{\text{CACE}}(x)$	& $0.749$\\
										$\mathcal{I}_{\text{inf}}$					& $7 \%$ \\
										$\widehat{\pi}_C(x)$ 			& 0.249 \\
									\end{tabular}
								}
						}}
						child{child {node[leaf] { 
									\begin{tabular}{>{\raggedright\arraybackslash}m{3cm} >{\raggedleft\arraybackslash}m{1.5cm}}
										$\widehat{\tau}^{\text{CACE}}(x)$	& $2.263^{*}$\\
										$\mathcal{I}_{\text{inf}}$					& $52 \%$ \\
										$\widehat{\pi}_C(x)$ 			& 0.324 \\
									\end{tabular}%
								}
						}}
					}	
				edge from parent node[near start, left, font=\Large, draw=black, thick, fill=white, inner sep=3pt, rounded corners=2pt]{FALSE}
				};	
		\end{tikzpicture}}

	\end{small}
	
\end{singlespace}
    \begin{notes}
    Each internal and terminal node reports (top) the estimated conditional CACE $\widehat\tau^{\text{CACE}}(x)$, (middle) the share of the inference sample $\mathcal{I}_{\text{inf}}$ falling into the node, and (bottom) the estimated complier share $\widehat\pi_C(x)$. 
    Splits are made on \textit{Age}, \textit{English}, and \textit{Education}, with the splitting rule indicated at each branch. At terminal nodes, $\widehat\tau^{\text{CACE}}(x)$ coincides with the subgroup-level 2SLS estimator $\widehat\tau^{\,\text{2SLS}}_{\mathbb{X}_j}$ in Definition \ref{defn:cCACE_estimator}; an asterisk marks leaves whose subgroup CACE is statistically significant at the 10\% level.
    \end{notes}
\end{sidewaysfigure}
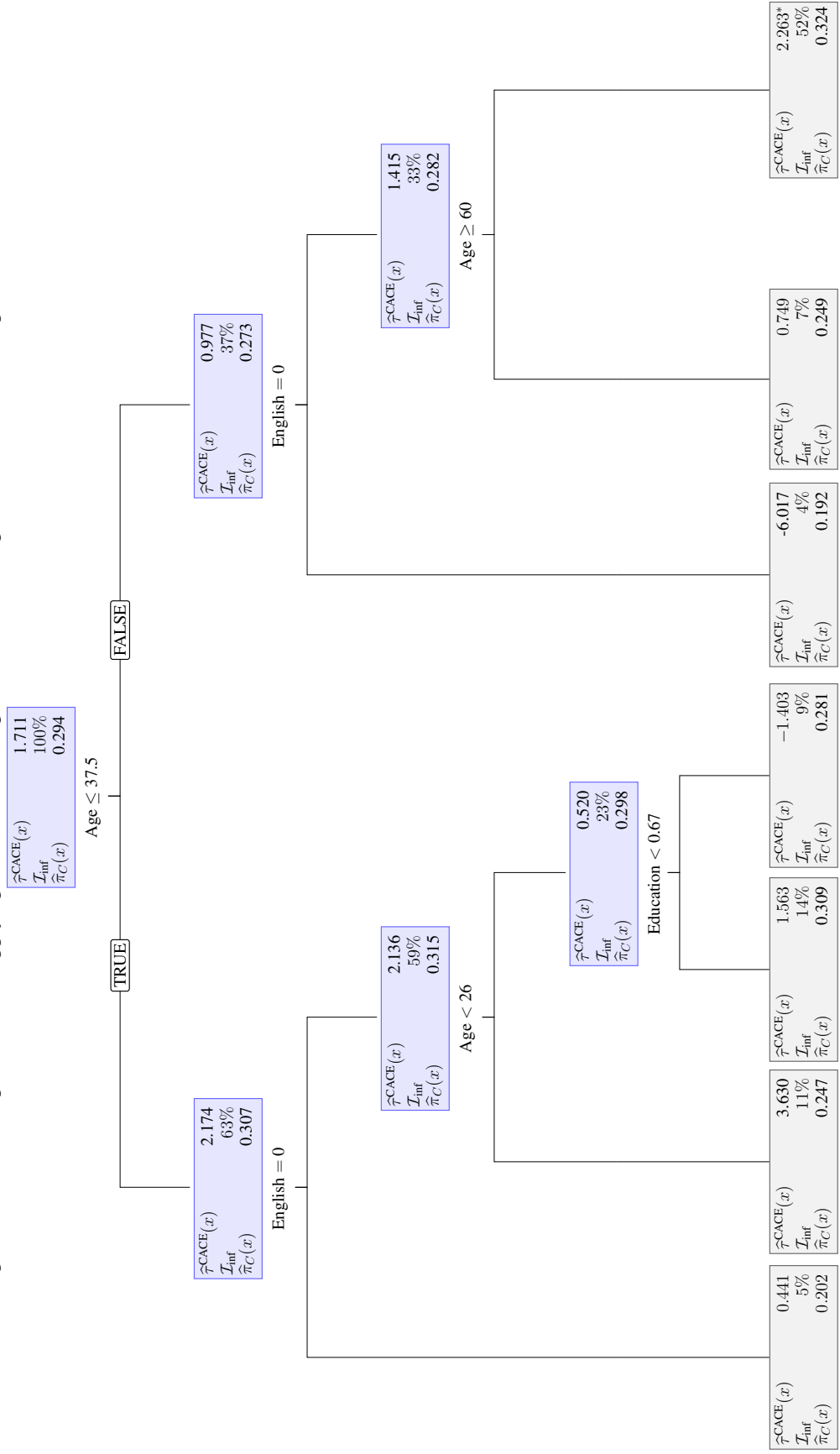

\clearpage
\newpage

\begin{sidewaysfigure}
   \centering
   \caption{Discovered partition from applying \gls{sbcf-iv} with Algorithm \ref{alg:sbcf-iv} to the 401(k) dataset.} 
    \label{fig:tree_subgroups_401k}
    \begin{singlespace}
	
	\begin{small}
		
		\resizebox{\textwidth}{!}{
			
			\begin{tikzpicture}[
				level distance=55mm,
				level 1/.style={sibling distance=325mm},
				level 2/.style={sibling distance=200mm},
				level 3/.style={sibling distance=100mm},
				level 4/.style={sibling distance=50mm},
				every node/.style={align=center},
				decisionbox/.style={
					rectangle, draw=blue!60, thick, fill=blue!10,
					align=center,
					inner xsep=-2.5pt, inner ysep=4pt,
					font=\Large
				},
				leaf/.style={
					rectangle, draw=black!60, thick, fill=gray!10,
					align=center, 
					inner xsep=-2.5pt, inner ysep=4pt,
					font=\Large
				},
				highlightleaf/.style={
					rectangle, draw=red!80!black, thick, fill=gold!75,
					align=center,
					inner xsep=-2.5pt, inner ysep=4pt,
					font=\Large
				},
				edge from parent path={
					(\tikzparentnode.south) -- ++(0,-10pt) -| (\tikzchildnode.north)
				}
				]
				
				\node {
					\begin{tikzpicture}[baseline]
						\node[decisionbox] (box) {
							\begin{tabular}{>{\raggedright\arraybackslash}m{2cm} >{\raggedleft\arraybackslash}m{1.7cm}}
								$\widehat{\tau}^{CACE}(x)$ & 27,826 \\
								$\mathcal{I}_{\text{inf}}$ & $100 \%$ \\
								$\widehat{\pi}_C(x)$ & 0.712 \\
							\end{tabular}
						};
						\node[below=5pt of box] (rule) {\Large inc $<$ 147{,}700};
					\end{tikzpicture}
				}
				child {node {
						\begin{tikzpicture}[baseline]
							\node[decisionbox] (box) {
								\begin{tabular}{>{\raggedright\arraybackslash}m{2cm} >{\raggedleft\arraybackslash}m{1.7cm}}
									$\widehat{\tau}^{CACE}(x)$ & 26,470 \\
									$\mathcal{I}_{\text{inf}}$                  & $99.7 \%$ \\
									$\widehat{\pi}_C(x)$         & 0.711 \\
								\end{tabular}
							};
							\node[below=5pt of box] {\Large inc $\geq$ 110{,}400};
						\end{tikzpicture}
					}
					child {node {
							\begin{tikzpicture}[baseline]
								\node[decisionbox] (box) {
									\begin{tabular}{>{\raggedright\arraybackslash}m{2cm} >{\raggedleft\arraybackslash}m{1.7cm}}
										$\widehat{\tau}^{CACE}(x)$ 	& 43,358 \\
										$\mathcal{I}_{\text{inf}}$ 					& $1.2 \%$ \\
										$\widehat{\pi}_C(x)$ 			& 0.839 \\
									\end{tabular}
								};
								\node[below=5pt of box] {\Large inc $\geq$ 135{,}000};
							\end{tikzpicture}
						}
						child { child {node[leaf] {
									\begin{tabular}{>{\raggedright\arraybackslash}m{2cm} >{\raggedleft\arraybackslash}m{1.7cm}}
										$\widehat{\tau}^{CACE}(x)$ 	& 144,376 \\
										$\mathcal{I}_{\text{inf}}$ 					& $0.3 \%$ \\
										$\widehat{\pi}_C(x)$ 			& 1.000 \\
									\end{tabular}	 
								}
						}}
						child {node {
								\begin{tikzpicture}[baseline]
									\node[decisionbox] (box) {
										\begin{tabular}{>{\raggedright\arraybackslash}m{2cm} >{\raggedleft\arraybackslash}m{1.7cm}}
											$\widehat{\tau}^{CACE}(x)$ 	& 7,643 \\
											$\mathcal{I}_{\text{inf}}$ 					& $0.9 \%$ \\
											$\widehat{\pi}_C(x)$ 			& 0.762 \\
										\end{tabular}
									};
									\node[below=5pt of box] {\Large pira $=$ 1};
								\end{tikzpicture}
							}
							child {node[leaf] {
									\begin{tabular}{>{\raggedright\arraybackslash}m{2cm} >{\raggedleft\arraybackslash}m{1.7cm}}
										$\widehat{\tau}^{CACE}(x)$	& $-6{,}783$ \\
										$\mathcal{I}_{\text{inf}}$					& $0.7 \%$ \\
										$\widehat{\pi}_C(x)$ 			& 0.824 \\
									\end{tabular}
								}
							}
							child {node[leaf] {
									\begin{tabular}{>{\raggedright\arraybackslash}m{2cm} >{\raggedleft\arraybackslash}m{1.7cm}}
										$\widehat{\tau}^{CACE}(x)$ 	& $-29{,}366$ \\
										$\mathcal{I}_{\text{inf}}$					& $0.2 \%$ \\
										$\widehat{\pi}_C(x)$ 			& 0.500 \\
									\end{tabular}
								}
							}
						}
					}
					child {node {
							\begin{tikzpicture}[baseline]
								\node[decisionbox] (box) {
									\begin{tabular}{>{\raggedright\arraybackslash}m{2cm} >{\raggedleft\arraybackslash}m{1.7cm}}
										$\widehat{\tau}^{CACE}(x)$ 	& 25,327 \\
										$\mathcal{I}_{\text{inf}}$ 					& $98.5 \%$ \\
										$\widehat{\pi}_C(x)$ 			& 0.709 \\
									\end{tabular}
								};
								\node[below=5pt of box] {\Large inc $\geq$ 92{,}690};
							\end{tikzpicture}
						}
						child { child {node[leaf] {
									\begin{tabular}{>{\raggedright\arraybackslash}m{2cm} >{\raggedleft\arraybackslash}m{1.7cm}}
										$\widehat{\tau}^{CACE}(x)$	& $53{,}422^{*}$ \\
										$\mathcal{I}_{\text{inf}}$					& $2.6 \%$ \\
										$\widehat{\pi}_C(x)$ 			& 0.913 \\
									\end{tabular}
								}
						}}
						child {node {
								\begin{tikzpicture}[baseline]
									\node[decisionbox] (box) {
										\begin{tabular}{>{\raggedright\arraybackslash}m{2cm} >{\raggedleft\arraybackslash}m{1.7cm}}
											$\widehat{\tau}^{CACE}(x)$	& 22,670 \\
											$\mathcal{I}_{\text{inf}}$					& $95.9 \%$ \\
											$\widehat{\pi}_C(x)$ 			& 0.700 \\
										\end{tabular}
									};
									\node[below=5pt of box] {\Large inc $<$ 68{,}810};
								\end{tikzpicture}
							}
							child {node[leaf] {
									\begin{tabular}{>{\raggedright\arraybackslash}m{2cm} >{\raggedleft\arraybackslash}m{1.7cm}}
										$\widehat{\tau}^{CACE}(x)$	& $17{,}818^{*}$ \\
										$\mathcal{I}_{\text{inf}}$					& $89.1 \%$ \\
										$\widehat{\pi}_C(x)$ 			& 0.689 \\
									\end{tabular}
								}
							}
							child {node[leaf] {
									\begin{tabular}{>{\raggedright\arraybackslash}m{2cm} >{\raggedleft\arraybackslash}m{1.7cm}}
										$\widehat{\tau}^{CACE}(x)$ 	& 25,271 \\
										$\mathcal{I}_{\text{inf}}$					& $6.8 \%$ \\
										$\widehat{\pi}_C(x)$ 			& 0.789 \\
									\end{tabular}
								}
							}
						}
					}
				edge from parent node[near start, right, font=\Large, draw=black, thick, fill=white, inner sep=3pt, rounded corners=2pt]{TRUE}
				}
				child {
					child{ child{ child {node[leaf] {
									\begin{tabular}{>{\raggedright\arraybackslash}m{2cm} >{\raggedleft\arraybackslash}m{1.7cm}}
										$\widehat{\tau}^{CACE}(x)$	& 250,409 \\
										$\mathcal{I}_{\text{inf}}$					& $0.3 \%$ \\
										$\widehat{\pi}_C(x)$ 			& 0.857 \\
									\end{tabular}
								}
					}}}
				edge from parent node[near start, left, font=\Large, draw=black, thick, fill=white, inner sep=3pt, rounded corners=2pt]{FALSE}
				};	
		\end{tikzpicture}}

	\end{small}
	
\end{singlespace}
    \begin{notes}
        Each internal and terminal node reports (top) the estimated conditional CACE $\widehat\tau^{\text{CACE}}(x)$, (middle) the share of the inference sample $\mathcal{I}_{\text{inf}}$ falling into the node, and (bottom) the estimated complier share $\widehat\pi_C(x)$. At terminal nodes, $\widehat\tau^{\text{CACE}}(x)$ coincides with the subgroup-level 2SLS estimator $\widehat\tau^{\,\text{2SLS}}_{\mathbb{X}_j}$ in Definition \ref{defn:cCACE_estimator}. An asterisk marks leaves whose subgroup CACE is statistically significant at the 10\% level.
    \end{notes}
\end{sidewaysfigure}
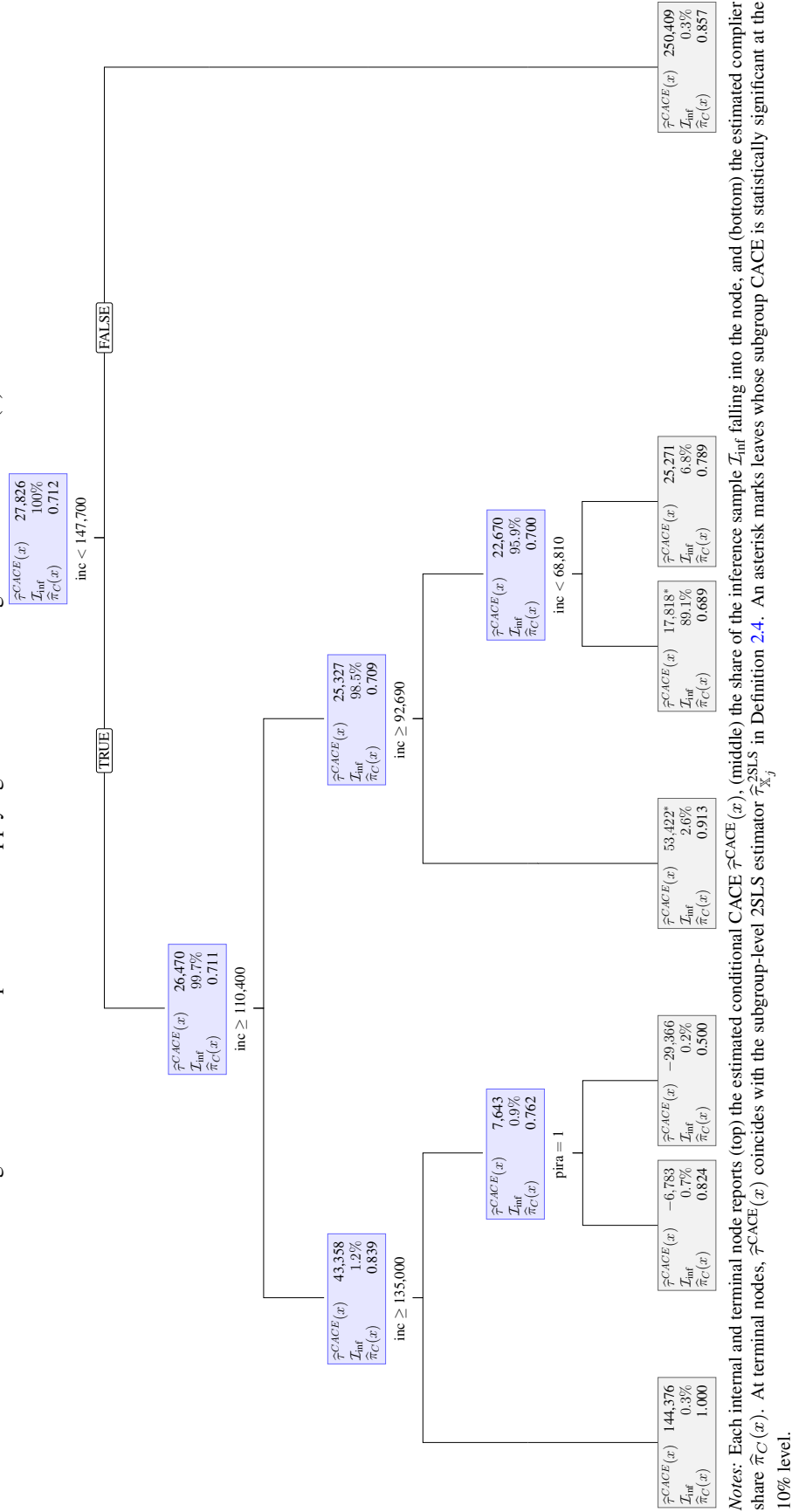

\clearpage
\newpage

\end{document}